\documentclass[noinfoline,11pt]{imsart}

\usepackage[OT1]{fontenc}
\usepackage{amsthm,amsmath,graphicx,latexsym,amssymb,array,mathabx,mathrsfs}
\usepackage{natbib}
\usepackage[colorlinks,citecolor=blue,urlcolor=blue]{hyperref}
\usepackage{adjustbox,blindtext}
\usepackage{float}
\usepackage[dvipsnames]{xcolor}

\startlocaldefs
\theoremstyle{plain}
\newtheorem{theorem}{Theorem}
\newtheorem{counterexample}{Counterexample}

\newtheorem{definition}{Definition}
\newtheorem{corollary}{Corollary}
\newtheorem{proposition}{Proposition}
\newtheorem{lemma}{Lemma}
\newtheorem{remark}{Remark}
\endlocaldefs

\usepackage[normalem]{ulem}
\allowdisplaybreaks

\usepackage{fullpage}

\begin{document}

\begin{frontmatter}

\title{{Functional Registration and Local Variations: Identifiability, Rank, and Tuning}} 
\runtitle{Functional Registration and Local Variations}

\begin{aug}
\author{\fnms{Anirvan} \snm{Chakraborty}\thanksref{a}\ead[label=e1]{anirvan.c@iiserkol.ac.in}}
\and
\author{\fnms{Victor M.} \snm{Panaretos}\thanksref{b}
\ead[label=e2]{victor.panaretos@epfl.ch}}

\address[a]{Indian Institute of Science Education \& Research (IISER) Kolkata, India\\\printead{e1}}
\address[b]{Institut de Math\'ematiques,
\'Ecole Polytechnique F\'ed\'erale de Lausanne (EPFL), Switzerland\\\printead{e2}}

\runauthor{A. Chakraborty and V.~M. Panaretos}



\end{aug}

\begin{abstract}
We develop theory and methodology for the problem of nonparametric registration of functional data that have been subjected to random deformation (warping) of their time scale. The separation of this phase variation (``horizontal" variation) from the amplitude variation (``vertical" variation) is crucial in order to properly conduct further analyses, which otherwise can be severely distorted. We determine precise nonparametric conditions under which the two forms of variation are identifiable. These show that the identifiability delicately depends on the underlying rank. By means of several counterexamples, we demonstrate that our conditions are sharp if one wishes a genuinely nonparametric setup; and in doing so we caution that popular remedies such as structural assumptions or roughness penalties can easily fail. We then propose a nonparametric registration method based on a ``local variation measure'', the main element in elucidating identifiability. A key advantage of the method is that it is free of any tuning or penalisation parameters regulating the amount of alignment, thus circumventing the problem of over/under-registration often encountered in practice. We provide asymptotic theory for the resulting estimators under the identifiable regime, but also under mild departures from identifiability, quantifying the resulting bias in terms of the amplitude variation's spectral gap.



\end{abstract}



\end{frontmatter}


{ \begin{footnotesize}
\tableofcontents
\end{footnotesize}
}

\section{Background and Contributions} \label{sec0}

\subsection*{Background}

\indent Functional observations can fluctuate around their mean structure in broadly two ways: (a) amplitude variation, and (b) phase variation. The first type of variation is analysed using functional principal component analysis, which stratifies the variation in amplitude (or variation in the ``vertical axis'') across the different eigenfunctions of the covariance operator of the underlying distribution. The second kind of variation, if present, is more subtle and can drastically distort the analysis of a functional dataset. It typically manifests itself in functional data representing physiological processes or physical motion, and consists in deformations of the time scale of the functional data (or variation in the ``horizontal axis''), associating to each observation its own unobservable time scale resulting from a transformation of the original time scale by a time warp. Specifically, instead of observing curves $\{X_{i}(t):[0,1]\rightarrow \mathbb{R}\}_{i=1}^{n}$, one actually observes warped versions $\widetilde{X}_{i} = X_{i} \circ T_{i}^{-1}$, where the $T_{i}$'s are unobservable (random) homeomorphisms termed \emph{warp maps}. In the presence of phase variation, the mean of the warped data conditional on the warping, $E(X_{i}|T_{i}) = \mu \circ T_{i}^{-1}$, is a distortion of the true mean $\mu$ by the warp map. Failing to account for the time transformation will yield deformed mean estimates, converging to $E[\mu \circ T_{i}^{-1}]$ rather than $\mu$. More dramatic still will be the effect on the estimation of the covariance of the latent process, inflating its essential rank, and yielding uninterpretable principal components. We refer to Section 2 in \cite{PZ16} for a detailed discussion of these effects. Consequently, in the presence of phase variation in the data, the natural first step in the analysis should be to \emph{register} the data, i.e., to simultaneously transform/synchronise the curves back to the \emph{objective} time scale.

\indent Owing to the rather complex nature of the registration problem, a variety of different assumptions on the latent process $X_{i}$ and the warp maps $T_{i}$ have been considered, and correspondingly a multitude of methods have been investigated: landmark based registration \citep{KG92}; template/target based registration \citep{RL98}; registration using dynamic time warping \citep{WG97,WG99}; registration based on local regression \citep{KLMR00}; a ``self-modelling'' approach by \cite{GG04} for warp maps expressible as linear combinations of B-splines; related registration procedures under assumptions on functional forms of the warp maps that result in a finite dimensional family of deformations \citep{Ronn01,GG05}; a functional convex synchronization approach to registration \citep{LM04}; registration using ``moments'' of the data curves \citep{Jame07}; registration based on a parsimonious representation of the registered observations by the principal components \citep{KR08}; pairwise registration of the warped functional data under monotone piecewise-linear warp maps \citep{TM08}; a joint amplitude-phase analysis with this pairwise registration procedure but considering step-function (thus finite dimensional) approximations of the warp maps using finite difference of their log-derivatives (centered log-ratio transform) \citep{HAME15}; registration when the warp maps are generated as compositions of elementary ``warplets'' \citep{CSS10}; and registration using a warp-invariant metric between curves when the warp functions are diffeomorphisms on an interval \citep{SWKKM11}. The above list is not exhaustive and we refer to \cite{MRSS15} for an oveview and comparison of some of the registration procedures mentioned above. More recently, \cite{PHCA17} applied the pairwise registration procedure of \cite{TM08} for two-dimensional curves, where the warping is in only one of the dimensions, while \cite{LA17} generalized the pairwise registration method for manifold valued data. 

Several of the above contributions consider the case when the warp maps are themselves random, and in such cases, a canonical set of assumptions is usually required: 
\begin{itemize}
\item[(a)] $T$ is a strictly increasing homeomorphism with probability one, and 
\item[(b)] $E(T) = Id$, where $Id$ is the identity map, $Id(x)=x$. 
\end{itemize}
The first assumption rules out ``time-reversal'' or ``time-jumps", while the second disallows an overall speed-up or slow-down of time. Further to these natural assumptions, most of the above cited papers impose additional smoothness and structural assumptions on the warp maps, which require tuning parameters to be selected. However, it is unclear whether these additional assumptions are either necessary or indeed sufficient for identifiability to hold. It is an open problem to determine what assumptions must one minimally impose on the latent functional data generating process so that the registration problem be identifiable under conditions (a) and (b) on the warp maps. This is of importance to understand since, in practice, one rarely has more detailed insights regarding the underlying warping phenomenon.
\\
\indent Consider the model
\begin{equation}\label{standard_model}
 X_{i}(t) = \xi_{i}\phi(t) + \delta\epsilon_{i}(t), \quad i=1,2,\ldots,n
 \end{equation}
for the latent process, with $\phi$ a unit norm deterministic function, $\xi_i$ random scalars, and $\epsilon_{i}(t)$ zero-mean random functions of unit variance (i.e. $E\left(\int_{0}^{1} \epsilon_i^{2}(t)dt\right)=1$).   When $\delta$ is unrestricted, the model \eqref{standard_model} spans any possible functional datum. The value of $\delta$ then regulates the balance between an (effectively) low rank model ($\delta^2\ll\mbox{var}\{\xi_i\}$) or a higher rank model (larger $\delta^2 \sim \mbox{var}\{\xi_i\}$). When one has exactly $\delta=0$ one has a rank 1 model. Several well-known approaches for registration available in the literature (see, e.g., \cite{Ronn01}, \cite{GG04,GG05}, \cite{TM08}, \cite{SWKKM11}) have considered variants of model \eqref{standard_model}, with the assumption that $\delta^2$ is small relative to $\mbox{var}\{\xi_i\}$ (for this reason, and for ease of reference, we thus henceforth refer to Model \eqref{standard_model} as the ``standard model"). In other words, it is postulated that if it were not for phase variation, important landmark features such as peaks and valleys of the latent process would not drastically change from realisation to realisation. In effect, there seems to be a certain concordance that identifiability (and hence consistency in the usual sense) rests crucially on an implicit assumption that  the amplitude variation of the synchronised functions \emph{is of low rank}, and phase variation is \emph{dominant} over amplitude variation. 

Observe that the dominating component $\xi_{i}\phi(T_{i}^{-1}(t))$ in the warped process $X_{i}(T_{i}^{-1}(t))$ obtained by warping model \eqref{standard_model} forms a sub-class of the so-called general non-linear shift models (NLSM). These models find extensive use in comparison of semi-parametric regression models (see, e.g., \cite{HM90}), and have been studied in the context of landmark and dynamic registration techniques by \cite{KG92} and \cite{WG97,WG99}. Also note that the landmark principle of registration essentially stipulates that the true curves have similar shape (thus having the same landmarks) but possibly differ in their amplitude component. Although some of the earlier papers, e.g., \cite{RL98}, \cite{KLMR00}, \cite{KR08}, \cite{CSS10} consider higher rank models for the latent process corresponding to nontrivial $\delta$ (with additional structural assumptions on warp maps), it is not known whether these procedures are truly identifiable/consistent. Indeed, \cite{KR08} (see p. 1160) acknowledged the fact that for such higher rank models, one can have \emph{different valid registrations} based on the degree of complexity of the warp maps that one allows (cf. Counterexample \ref{counter_6}). Further, as hinted in \cite{TM08}, who consider model \eqref{standard_model}, identifiable (consistent) registration appears not to be guaranteed unless one lets $\delta \rightarrow 0$ as $n \rightarrow \infty$. {Recently, \cite{WK19} studied a nonparametric registration procedure that registers the warped curves to a low dimensional subspace provided that the number of ``feature points'' are bounded almost surely.}

\subsection*{Our Contributions}
 
We contribute to the nonparametric synchronisation problem with theory, methodology, and asymptotics, and corroborate our findings with simulations and a data analysis: 

\begin{enumerate}
\item Firstly, we provide a rigorous study of the issue of identifiability, providing conditions for the standard model \ref{standard_model} to be identifiable. Conversely, we show by means of several counterexamples that our conditions are relatively sharp.

\item Secondly, we develop methodology to address the problem of nonparametric and consistent recovery of the warp maps from discretely warped curves without structural assumptions on the warp maps further to (a) and (b). A salient feature of our methods is that they do not require penalisation or tuning related to the warp maps. 
Our methodology is adapted to cover all three standard observation settings: \emph{complete observation}, \emph{discrete observation}, and \emph{discrete observation with measurement error}.

\item We carry out an asymptotic analysis in all three observation settings. 
We also investigate the setting when the model deviates from identifiability and derive results quantifying the amount of asymptotic bias incurred in terms of the spectral gap of the amplitude variation (Theorem \ref{thm5}). 

\item We probe the finite sample performance of our methodology (Section \ref{sec3}), for all possible observation regimes, and compare to other popular registration techniques, including in settings departing from the identifiable regime, observing a noteworthy stability of our method to mild such departures. 

\item The method is further illustrated by analysis of a functional dataset of \emph{Triboleum} beetle larvae growth curves (Section \ref{sec4}), yielding biologically interpretable results. 
\end{enumerate}

The key to our results is the novel use of a criterion that measures the local amount of deformation of the time scale (Section  \ref{sec2}), the \emph{local variation measure} of $X$, with associated cumulative distribution $J_X(t)=\int_{0}^{t} |X'(u)|du$. Here, and throughout the rest of the paper, for any function $f$ (random or deterministic), its first order derivative will be denoted by $f'$, i.e., $f'(t) = df(t)/dt$. Further, its second, third and fourth order derivatives will be denoted by $f''$, $f'''$ and $f^{(4)}$, respectively. The local variation measure reflects how the total amount of variation of the curve is distributed on the real axis. The simple but consequential insight is that by a change-of-variable argument, the total variation measure remains invariant under any strictly increasing deformation $T$ of the time scale of $X$, namely, $J_{X}(1)=J_{\widetilde{X}}(1),$ where $\widetilde{X} = X \circ T^{-1}$. However, it is the local amount of deformation that provides the information about the warping mechanism. 
 {
 This insight is what circumvents the need to introduce registration tuning parameters -- even when the curves are observed over a discrete grid\footnote{Of course, once the warp maps are estimated, one would have to smooth the warped discrete data in order to register them, since the warped data are not observed at all points of their domain. And, if there is measurement error in the observations, then some pre-smoothing will be needed. But in either case, this smoothing will be on the data itself (either as a pre-processing or post-processing step), and no smoothing penalties or structural assumptions will be required on the registration maps themselves.}. This connection also guides us in the construction of counterexamples, illustrating where caution should be taken. 

\section{Identifiability and Counterexamples}  \label{sec1} 

Recall that the \emph{standard model} for the latent/synchronised process prior to warping (Equation \ref{standard_model}) takes the general form
$$
X(t) = \xi\phi(t) + \delta\epsilon(t).
$$
This, depending on the constraints imposed on the random variable $\xi$ and the scalar $\delta$, can be of arbitrarily large rank, and indeed can span any functional datum. Usually $\mbox{var}\{\xi\}$ is expected to be the dominant effect relative to $\delta$ (i.e. $\delta^2\ll\mbox{var}\{\xi\}$), corresponding to an effectively low rank model. We now give sufficient conditions on the standard model for that identifiability will hold in a genuine nonparametric sense. In simple terms, these require the process must be \emph{exactly} of rank 1 (i.e. $\delta=0$ or $\epsilon(t) \in \mbox{span}\{\phi(t)\})$.

\begin{theorem}[Identifiability] \label{thm1}
Let $\{X_1,X_2\}$ be a random elements in $C^{1}[0,1]$ of rank one, i.e., $X_i(t) = \xi_i\phi_i(t)$ for deterministic functions $\phi_i$ with $\int_{0}^{1} \phi_i^{2}(t)dt = 1$, and with $\phi'_i$ vanishing on at most a countable set. Assume that $\{T_1,T_2\}$  are strictly increasing homeomorphisms in $C^{1}[0,1]$, and such that $E(T_i) = Id$. Write $\widetilde{X}_{i}=X_i(T_i^{-1}(t))$. Then,
$$\widetilde{X}_{1}\stackrel{d}{=}\widetilde{X}_{2} \iff \Big\{T_{1} \stackrel{d}{=} T_{2},\quad \phi_{1} = \pm\phi_{2},\quad  \xi_{1} = \pm\xi_{2}\Big\}.$$ 
\end{theorem}

The proof is given in Section \ref{proofs} of the Appendix (Supplementary Material). The assumption that $\phi'$ does not vanish except perhaps on a countable set excludes the possibility of constant functions, in which case the problem is vacuous and identifiability trivially fails. Note that the identifiability result in Theorem \ref{thm1} \textit{does not} require that $\xi$ and $T$ be independent.


One might understandably argue that the rank 1 assumption in the previous theorem is restrictive. Perhaps surprisingly, though, the condition can be seen to be rather sharp. We construct a series of counterexamples, demonstrating how unstable identifiability is to higher ranks (even rank 2). These illustrate that the situation cannot be rectified at a genuinely nonparametric level, e.g. by assuming specific classes of models on the synchronised processes (such as splines or trigonometric functions); or by imposing qualitative non-parametric constraints, such as roughness penalties, Sobolev norm bounds or rank restrictions on the warp maps (or combinations of these). 

\begin{counterexample}\label{counter_1}\normalfont
Our first counterexample shows that the same rank 2 process can arise either as warped rank 1 process, or as a syncrhonised rank 2 process. Both the process itself and the warp maps can be taken to be of rank at most 2 (notice that a rank 1 warp map would need to be the identity almost surely). Define $f(t) = (3t + t^{2})/4$ and $g(t) = (5t - t^{2})/4, ~t \in [0,1].$ Take $\xi$ to be a standard Gaussian random variable and $\phi(t) = t/\sqrt{3}$ for $t \in [0,1]$. Now define a random warp map $T$ such that $P[T=f]=P[T=g]=1/2$. Then $T$ satisfies (a) and (b). Now define $\widetilde{X} = \xi\phi\circ T^{-1} = \xi_{1} T^{-1} = \xi_{1} (f^{-1} U + g^{-1}(1-U))$, where $U$ is a Bernoulli random variable with success probability $1/2$ and $\xi_{1} = \xi/\sqrt{3}$. Let $V = \xi_{1} U$ and $W = \xi_{1} (1-U)$ so that $\widetilde{X} = Vf_{1} + Wg_{1},$ 
where $f_{1}(t) = f^{-1}(t) = (\sqrt{9+16t} - 3)/2$ and $g_{1}(t) = g^{-1}(t) = (5 - \sqrt{25-16t})/2$, $t \in [0,1]$. Since $f$ and $g$ are $C^{\infty}$, and $f'$ and $g'$ are bounded away from zero on $[0,1]$, so are their inverses. Also, the inverses are $C^{\infty}$ as well. It is easy to check that $\mathrm{Cov}(V,W) = 0$. Further, it is easy to show that $f_{1}$ and $g_{1}$ are linearly independent. Consequently, we may define a new process $Y = Vf_{1} + Wg_{1}$, which is a rank two process. Define $\widetilde{Y} = Y \circ Id^{-1} = Y$. Then, $\widetilde{X}\stackrel{d}{=}\widetilde{Y}$ (in fact $\widetilde{X}{=}\widetilde{Y}$) but they have been generated using two different $C^{\infty}$ latent processes, namely $X$ and $Y$, and $C^{\infty}$ warp maps, namely $T$ and $Id$, which of course do not have the same distribution.
\end{counterexample} 

\begin{counterexample}[{Rank 2 warped to rank 1 -- smooth latent process and warp maps}]
\label{counter_3}\normalfont
We will give two constructions demonstrating that the same rank one process can arise in one of \textit{infinitely many} ways: (i) as a rank one analytic process with no warping,  and (ii) one of an infinite collection of rank two analytic processes subjected to warping by one of an infinite collection of non-trivial analytic warp maps $T$ satisfying (a) and (b). 

 (A) First take the latent model class to consist of linear combinations of trigonometric functions and polynomials. Define $\mu(t) = 2t - 1$ and $\phi_{k}(t) = \sin((2k-1){\pi}t)/[(2k-1)\pi], t \in [0,1]$ for some $k \geq 1$. Let $T_{k}(t) = t - (2U_{k} - 1)\phi_{k}(t)$, where $U_{k} \sim \mathrm{Unif}(a,b)$. Here $a = (1/2)(1 - M^{-1})$ and $b = (1/2)(1 + M^{-1})$ with $M$ satisfying $M > 1$. It can be checked that $T_{k}$ satisfies (a) and (b) for all $k \geq 1$. Let $\xi$ be a random variable independent of $U_{k}$. Define $X(t) = \xi\mu(t)$ and $Y_{k}(t) = \xi\mu(t) + \xi(2-4U_{k})\phi_{k}(t)$. It can be checked that $X = \widetilde{Y}_{k} := Y_{k} \circ T_{k}^{-1}$ for all $k \geq 1$. Since $\xi$ an $U_{k}$ are independent, it follows that $\mathrm{Cov}(\xi,\xi(2-4U_{k})) = 0$. Also, since $\langle\mu,\phi_{k}\rangle = 0$ (by direct calculation), the form of $Y_{k}$ given above is in fact its Karhunen-Lo\'eve (KL) expansion, which is of rank 2, and this holds for all $k \geq 1$. The plots of sample paths of $Y_{1}$ and $Y_{2}$ along with the warp maps $T_{1}$ and $T_{2}$ are shown in Figure \ref{Fig-counter}.

(B) For the second construction, we take the latent model class to consist of linear combinations of polynomials only. 
Define $\mu(t) = t$. Fix $R \in \mathbb{N}$ and any finite subset $\{k_{1},k_{2},\ldots,k_{R}\}$ of $\mathbb{N}$. Also, fix reals $a_{1},a_{2},\ldots,a_{R}$ satisfying $\sum_{l=1}^{R} a_{l} = 0$. Consider the Legendre polynomials $P_{2k_{l} + 1}$ on $[-1,1]$. Since these satisfy $P_{2k_{l} + 1}(-t) = P_{2k_{l} + 1}(t)$ for $t \in [0,1]$, it follows that $\int_{0}^{1} tP_{2k_{l} + 1}(t)dt = (1/2)\int_{-1}^{1} tP_{2k_{l} + 1}(t)dt = 0$. Define $\phi(t) = \sum_{l=1}^{R} a_{l}P_{2k_{l} + 1}(t)$ and $T(t) = t - (2U - 1)\phi(t)$, where $U \sim \mathrm{Unif}(a,b)$, where $M > ||\phi'||_{\infty}$. Here and elsewhere in the paper, $||\cdot||_{\infty}$ will denote the usual sup-norm on $C[0,1]$. The above construction ensures that $T(0) = 0$, $T(1) = 1$, and $T$ satisfies (a). It is clear that $T$ satisfies (b). Let $X(t) = {\xi}t$ and $Y(t) = {\xi}t - \xi(2U-1)\phi(t)$, where $\xi$ is as in the first construction. Then, it can be shown that $X = \widetilde{Y} := Y \circ T^{-1}$. Also, $Y$ is rank 2, and the above form is in fact its KL expansion because $\mathrm{Cov}(\xi,\xi(2U-1)) = 0$ and $\langle\mu,\phi\rangle = 0$, which follows as earlier. 

By taking $\xi$ to be a constant random variable, this counterexample also shows that one cannot extend the identifiable regime from $\xi\phi(t)$ to $\mu(t) + \xi\phi(t)$, where $\mu \notin \mathrm{span}\{\phi\}$.
\end{counterexample} 
\begin{figure}[t!]
\begin{center}
{
\includegraphics[scale=0.46]{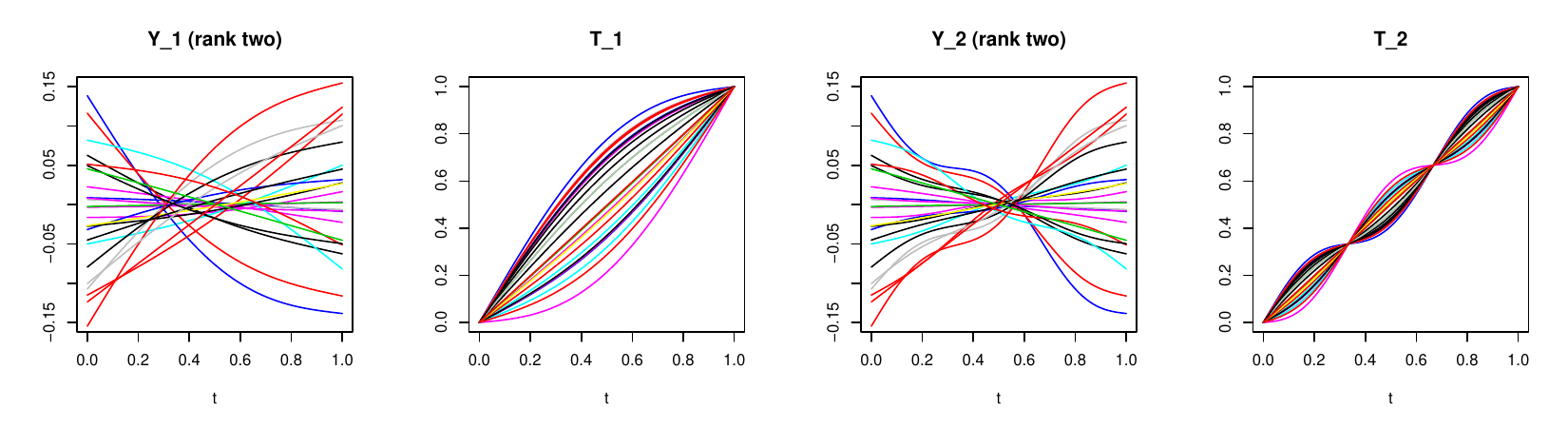}
}
\end{center}
\vspace{-0.3in}
\caption{Plots of some sample paths of the rank two latent processes $Y_{1}$ and $Y_{2}$ in part (A) of Counterexample \ref{counter_3} along with the warp maps $T_{1}$ and $T_{2}$ mentioned there, which warp them into the same rank one process.}
\label{Fig-counter}
\end{figure}
\begin{counterexample}[{Penalising the Warp Maps}]
\label{counter_4}\normalfont
We will show that even if one penalises the warp maps, e.g., by one or both of $\int_{0}^{1} E([T(t) - t]^{2})dt$ and $\int_{0}^{1} E[(T''(t))^{2}]dt$, still one can get \textit{infinitely many} possible solutions for the registration problem. Under the setup of (A) in Counterexample \ref{counter_3}, $\int_{0}^{1} E([T(t) - t]^{2})dt = [\sqrt{6}M{\pi}(2k-1)]^{-2}$ and $\int_{0}^{1} E[(T''(t))^{2}]dt = (2k-1)^{2}{\pi}^{2}/(6M^{2})$. For (B) in the previous counterexample, it can be shown using the orthogonality of the Legendre polynomials that $\int_{0}^{1} E([T(t) - t]^{2})dt = \{\sum_{l=1}^{R} a_{l}^{2}/(2k_{l}+1)\}/(3M^{2})$ and $\int_{0}^{1} E[(T''(t))^{2}]dt = ||\sum_{l=1}^{R} a_{l}P_{2k_{l}+1}''||^{2}/(3M^{2})$, where $||\cdot||$ denotes the usual norm on $L_2[0,1]$. Thus, in both cases, for any $\epsilon > 0$, the sum of the two penalty terms can be made arbitrarily small by choosing large enough $M$ (depending on the choices of the other parameters -- $k$, $R$, $k_{l}$'s and $a_{l}$'s).

The above facts imply that if one wants to carry out the registration using the penalization procedure $\min_{h \in \mathscr{T}} \int_{0}^{1} E\{[W_{h}(t) - X(h(t))]^{2} + \lambda_{1}[T(t) - t]^{2} + \lambda_{2}(T''(t))^{2}\}dt$, where  $\mathscr{T}$ is a class of $C^{\infty}$ warp maps, and $W_{h}$ takes values in an appropriate synchronized space $\mathscr{S}$ of linear combinations of $C^{\infty}$ functions, then we have \textit{infinitely many} registrations valid registrations as follows: \\
(i) under setup (A) -- if we allow $\mathscr{T}$ to include monotone homoemorphisms on $[0,1]$ whose deviation from the identity is a trigonometric function, and even if $\mathscr{S}$ is restricted to linear combinations of linear and trigonometric functions (both $X$ and $Y_{k}$ belong to this class). \\
(ii) under setup (B) -- even if we allow $\mathscr{T}$ and $\mathscr{S}$ to only include polynomials. \\
Note that for both (i) and (ii), the ``fit'' term $E\{[W_{h}(t) - X(h(t))]^{2}$ becomes zero.
\end{counterexample}

\begin{counterexample}[{Spline models for latent process and warp map}]
\label{counter_5}\normalfont
Our next counterexample shows that structural restrictions on the latent synchronised process, such as spline models, will also fail if the rank is higher than 1. We will consider cubic splines but one can similarly construct more elaborate counterexamples involving higher order splines and more knots. Let $\phi$ be a cubic spline with a single knot at $a_0 \in (0,1)$, i.e., $\phi(t) = \sum_{i=0}^{3} \theta_{i}t^{i} + \delta(t - a_0)_{+}^{3}$, and define $s(t) = c(a_1-a_0)^{-1}(t-a_0)I\{a_0 \leq t \leq a_1\} + c(1-a_1)^{-1}(1-t)I\{a_1 < t \leq 1\}, ~t \in [0,1]$, where $c \in \mathbb{R}$ and $a_1 \in (a_0,1)$ are fixed. Let $X(t) = \xi\phi(t)$ and $T(t) = t - (2U-1)s(t)$ with $U$ and $\xi$ as before, and choose $M > |c|/\min\{(a_1-a_0),(1-a_1)\}$. This ensures that $T$ satisfies (a) and (b). Define 
$$Y(t) = \xi\phi(t) + V_{1}s(t)\{\theta_{1} + \theta_{2}t - 3\theta_{3}t^{2} - 3\delta(t - a_0)_{+}^{2}\} + V_{2}s^{2}(t)\{\theta_{2} + 3\theta_{3}t + 3\delta(t - a_0)_{+}\} + V_{3}s^{3}(t),$$
where $V_{1} = \xi(1 - 2U), V_{2} = \xi(2U-1)^{2}$ and $V_{3} = \xi(1-2U)^{3}(\theta_{3} + \delta)$. Note that $s$ is a linear spline with knots at $a_0$ and $a_1$. Also, $p_{1}(t) := \theta_{1} + \theta_{2}t - 3\theta_{3}t^{2} - 3\delta(t - 1/2)_{+}^{2}$ and $p_{2}(t) := \theta_{2} + 3\theta_{3}t + 3\delta(t - 1/2)_{+}$ are splines (quadratic and linear, respectively) with knots at $1/2$. Hence, these can be considered as elements of the cubic spline space $\mathscr{S}_{0}$ with knots at $a_0$ and $a_1$. So, by repeated application of Theorem 3.1 in \cite{Mork91}, the functions $\phi$, $sp_{1}$, $s^{2}p_{2}$ and $s_{3}$ are elements of the space $\mathscr{S}_{1}$ of cubic splines with a finite set of knots (including $a_0$ and $a_1$). So, both $X$ and $Y$ lie in $\mathscr{S}_{1} \supset \mathscr{S}_{0}$. If we assume that $\phi(1) \neq 0$, then it follows that $\phi$ is linearly independent of $sp_{1}$, $s^{2}p_{2}$ and $s^{3}$ (since these three functions equal zero at $t = 1$). Thus, $Y$ is of rank at least two. Now, it can be checked that $\widetilde{Y}(t) := Y(T^{-1}(t)) = X(t)$. Thus, two distinct processes $X$ and $Y$ can be warped (by the maps $Id$ and $T$, respectively) to produce the same process.

If we choose $a_0 = 0$, i.e., take $\phi$ to be a cubic polynomial (which also lies in $\mathscr{S}_{0}$ trivially), then we can choose $s$ to be a spline on $[0,1]$ of degree $\geq 2$ with a fixed set of knots. So, in this case, we can have differentiable (instead of a.e. differentiable) warp maps. In this case, we choose $M > ||s'||_{\infty}$. Then, for the same $Y$, the conclusion of the above counterexample holds.
\end{counterexample}

\begin{counterexample}[{Latent process with geometric properties}]
\label{counter_6}\normalfont
Our last counterexample illustrates that even a priori knowledge of landmarks does not help rectify identifiability if the rank 1 condition is violated. Let $X(t) = {\xi}t(1-t), ~t \in [0,1]$ so that the latent process has a unique maximum at $t = 1/2$. A priori knowledge of existence of a unique maximum in synchronized space can be utilized to carry out a landmark/peak alignment of the warped curves. Let us denote the vector space of functions with unique maximum at $t = 1/2$ by $\mathscr{U}$, and the vector space of functions proportional to the bell-shaped curve $f(t) = t(1-t)$ by $\mathscr{S}_{f}$. Obviously, $X \in \mathscr{S}_{f} \subset \mathscr{U}$. Let $T$ be any warp map independent of $\xi$ and satisfying (a) and (b). Define a new warp map $S$ as follows: $S(t) = 2tT(1/2)I\{0 \leq t \leq 1/2\} + T(1/2) + (2t-1)[1-T(1/2)]I\{1/2 \leq t \leq 1\}$. Note that $S$ satisfies (a) and (b). Define $Y(t) = {\xi}T^{-1}(S(t))[1 - T^{-1}(S(t))], ~t \in [0,1]$. It can be checked that the process $Y$ has a unique maximum at $t_0$, where $t_0$ satisfies $T^{-1}(S(t_0)) = 1/2$, equivalently, $t_0 = S^{-1}(T(1/2))$. However, from the construction of $S$, it is easy to check that $S^{-1}(T(1/2)) = 1/2$. So, $Y \in \mathscr{U}$. Defining $\widetilde{X} = X \circ T^{-1}$ and  $\widetilde{Y} = Y \circ S^{-1}$, it follows that $\widetilde{X} = \widetilde{Y}$ although $X$ and $Y$ are different processes. Further, although $X \in \mathscr{S}_{f}$, it holds that $Y \notin \mathscr{S}_{f}$ provided $S \neq T$, and $Y$ has rank at least two. This counterexample (without explicit constructions of the latent processes or of the warp maps) is mentioned in \cite{KR08}. {Note that in both cases, one achieves ``consistent registration'' in the sense of that paper.} 
\end{counterexample}

What we learn from these counterexamples is that identifiability crucially rests upon constructing a synchronised space of processes $\mathscr{S}$ (contained within continuous processes on $[0,1]$) and a warp map space of processes $\mathscr{T}$ (contained within strictly monotone homeomorphisms onto $[0,1]$ with identity expectation) such that: 
\begin{itemize}
\item[(I)] Warping causes the latent process to exit the synchronised space, i.e. $X\in\mathscr{S}$ but $\widetilde{X} \notin \mathscr{S}$.
\item[(II)] There exists a unique process  $X \in \mathscr{S}$ such that $\widetilde{X} = X \circ T^{-1}$ for some random $T \in \mathscr{T}$. 
\end{itemize}
Theorem \ref{thm1} informs us that such a construction is possible by taking $\mathscr{S}$ to essentially be $C^1$ rank 1 non-constant processes, and otherwise not restricting $\mathscr{T}$ except for a $C^1$ assumption. The counterexamples demonstrate that allowing higher ranks can have severe effect on identifiability, even if $\mathscr{S}$ is modeled more concretely, or indeed if $\mathscr{T}$ is restricted to be smoother. In light of this, we will introduce the terminology of ``identifiable regime" to mean the pair $(\mathscr{S},\mathscr{T})$ implied by the context of Theorem \ref{thm1}:

\begin{definition}[Identifiable Regime]\label{identifiable_regime}
We define the identifiable regime to involve latent synchronised processes $X\in\mathscr{S}$, warp maps $T\in\mathscr{T}$, and warped processes $\widetilde{X}(t)=X(T^{-1}(t))$, where:
\begin{itemize}
\item[(I1)] The synchronised process space is $\mathscr{S}=\{X\in C^1[0,1]: X(t)=\xi \varphi(t)\}$, for $\xi$ a real-valued random variable of finite variance and $\varphi\in C^{1}[0,1]$ is a deterministic function of unit $L^2$-norm, whose derivative vanishes at most on a countable subset of $[0,1]$. 

\item[(I2)] The warp map space is $\mathscr{T}=\{T\in C^1[0,1]: \mathbb{E}[T]=Id$ and $T$ is a strictly increasing homeomorphism$\}$.


\end{itemize}
\end{definition}

\begin{remark}[Unidentifiable vs Identifiable Regimes] \label{identifiable-rank-two-model}
Deviations from the identifiable regime will be generally termed as a ``potentially unidentifiable regime". We say ``potentially" because the conditions (I1) and (I2), though sharp in a fully nonparametric setting, are not necessary. One can presumably produce identifiable models escaping the framework in Definition \ref{identifiable_regime} by introducing explicit parametric assumptions (and/or perhaps considering weaker forms of identifiability). For example one could modify Definition \ref{identifiable_regime} into a rank-two synchronised model $X(t)=\xi\phi(t)+\zeta \psi(t)$ under the parametric assumption that $\psi(t)=1$. This model is an arguably trivial extension of the rank one model with an inclusion of an additive random vertical shift, which is obviously invariant to phase variation. If $\int_0^1\phi(u)du=0$ and $\mathrm{cov}\{\xi,\zeta\}=0$, this model can be shown to be identifiable in the sense of Theorem \ref{thm1}. We refer to the Supplementary Material for a formal statement (Theorem \ref{identifiable-rank-two-model-thm}) and further details. 
 Note that any method that can estimate the warp maps in the rank 1 model, can also accommodate the presence of a vertical shift as above (see Remark \ref{remark-regis-rank-two}), which is why we call the inclusion of a vertical shift a trivial modification.
\end{remark}

With identifiability clarified, we now turn to nonparametric methods of estimation. Our goal will be to construct methods that perform provably well in the identifiable regime, remain stable under small departures (e.g. effectively rank 1 rather than precisely rank 1 models), and do not rely on tuning (which adds a layer of arbitrariness and in any case was seen to be unavailing). For these, we will require the notion of \emph{local variation measure}, introduced in the next section.

{In closing this section we remark that \cite{WK19} also independently arrive at the conclusion that, under the standard assumptions on the warp maps, identifiability cannot be guaranteed beyond rank one models (though our results are antecedent, see \cite{arxiv}).}

\section{Tuning-Free Methodology}  \label{sec2}

Recall that the total variation of a continuous function $h(x):[0,1]\rightarrow\mathbb{R}$ measures the total distance sweeped by the ordinate $y=h(x)$ of its graph, as the abscissa $x$ moves from $0$ to $1$. By distorting functions ``in the $x$-domain" through an increasing homeomorphism, phase variation will not affect the total amount of variation accrued over the interval $[0,1]$. However, it will \emph{redistribute} this total variation over the subintervals of $[0,1]$. This redistribution can be measured by focussing on \emph{local variation}:

\begin{definition}[Local Variation Distribution]\label{local_variation}
Given any real function $h\in C[0,1]$, we define 
\begin{equation}\label{general_definition}
J_h(t)=\sup_{K \in {\cal K}_t}\sum_{k=0}^{|K|} |h(\tau_{k+1}) - h(\tau_{k})| 
\end{equation}
where $K_t = \{\tau_{0},\tau_{1},\ldots,\tau_{|K|}\}$ is a partition of $[0,t]$ and ${\cal K}_t$ is the collection of all finite partitions of $[0,t]$. Noting that $J_h(1)$ is the total variation of $h$, define the local variation distribution as $F_h(t)={J_h(t)}\big/{J_h(1)}$.
\end{definition}

\begin{remark}
Recall that when $h\in C^1([0,1])$, it holds that $J_h(t)=\int_0^t|h'(u)|du$. The general definition comes handy under discrete observation, this one under continuous observation.
\end{remark}

We now show that, in the identifiable regime, warping affects the local variation of the underlying process in a rather predictable manner -- one that can be used to motivate estimators. We will write $\widetilde{F}=F_{\widetilde{X}}$ and $F=F_{X}$ for simplicity.

\begin{lemma}[Local Variations and Warp Maps] \label{expectation_lemma} When $\widetilde{X}=X\circ T^{-1}$ fall under the Identifiable Regime (\ref{identifiable_regime}), $F$ and $\widetilde{F}$ are strictly monotone almost surely, 
$\mathbb{E}\left\{\widetilde{F}^{-1}\right\} = F^{-1}=F_{\phi}^{-1}$, and $T=\widetilde{F}^{-1}\circ F=\widetilde{F}^{-1}\circ F_{\phi}.$
\end{lemma}


\begin{remark}\label{wasserstein_remark}
In the language of transportation of measure, Lemma \ref{expectation_lemma} says that the warp map pushes forward the original local variation distribution to the warped local variation distribution, in fact \emph{optimally} so in terms of quadratic transportation cost; and that the synchronised local variation measure is the \emph{Fr\'echet mean} of the (random) warped local variation measure in the 2-Wasserstein distance.
\end{remark}


Now suppose we have an i.i.d. sample $\{\widetilde{X}_{i}: i=1,2,\ldots,n\}$ of randomly warped functional data that we wish to register, i.e. we wish to construct nonparametric estimators of the $\{X_i\}_{i=1}^{n}$ and the $\{T_i\}_{i=1}^{n}$ on the basis of $\{\widetilde{X_i}\}_{i=1}^{n}$. If we expect the data to (at least approximately) conform to the identifiable regime (\ref{identifiable_regime}), we can rely on Lemma (\ref{expectation_lemma}) as inspiration for tuning-free methodology. We would like to emphasize that this methodology will be applicable whatever the ``true model", of course, but the point is for it to be accurate under the identifiable regime, and stable when mildly departing from identifiability. We construct such methodology under all three different observation regimes on $\{\widetilde{X_i}\}_{i=1}^{n}$: complete observations (Section \ref{subsec2-1}), discrete noiseless observations (Section \ref{subsec2-2}), and discrete observations with measurement error (Section \ref{subsec2-3}). We then study the performance under identifiability/unidentifiability theoretically in Section \ref{asymptotics} and numerically in Section \ref{sec3}, where we indeed observe a certain stability to mild departures from identifiability.

\subsection{Fully Observed Functions}\label{subsec2-1}

\noindent Assuming the functions $\{\widetilde{X_i}\}$ are fully observed, we may proceed as follows:

\begin{itemize}
\item[Step 1:] 
Set
$$\widehat{F} = \left(n^{-1} \sum_{i=1}^{n} \widetilde{F}_i^{-1}\right)^{-1},$$
noting that the $\{\widetilde{F}_i\}$ are immediately available by complete observation of the $\{\widetilde{X}_i\}$.
\end{itemize}
Note that under the identifiable regime (\ref{identifiable_regime}), $\widehat{F}$ estimates $F_{\phi}$.
\begin{itemize}
\smallskip
\item[Step 2:] Estimate the warp map $T_{i}$ by $\widehat{T}_{i} =\widetilde{F}_i^{-1} \circ \widehat{F}$, and the registration map $T_i^{-1}$ by $\widehat{T}_i^{-1}$. 

\smallskip
\item[Step 3:] Register the observed warped functional data, by means of $\widehat{X}_{i} = \widetilde{X}_{i} \circ \widehat{T}_{i}$. 
\end{itemize}
If we suspect to be in the identifiable regime (\ref{identifiable_regime}), we may also want to estimate the pairs $\{\phi,\xi_i\}$. In this case, the obvious additional steps will be:
\begin{itemize}
\item[Step 4:] Compute the empirical covariance operator, say, $\widehat{\mathscr{K}}_{r}$ of the registered data $\{\widehat{X}_{i}\}$ and estimate $\phi$ by the leading eigenfunction $\widehat{\phi}$ of $\widehat{\mathscr{K}}_{r}$ (as a convention, assume that this estimator is aligned with the true $\phi$, i.e., $\langle\widehat{\phi},\phi\rangle \geq 0$). 

\smallskip
\item[Step 5:] Estimate $\xi_{i}$ by $\widehat{\xi}_{i} = \langle\widehat{X}_{i},\widehat{\phi}\rangle$. \\
\end{itemize}
In Step 4 above as well as in all other instances in the paper, $\langle\cdot,\cdot\rangle$ denotes the usual inner product on $L_2[0,1]$.


\subsection{Discretely Observed Functions} \label{subsec2-2}

In the discretely observed setting, the $\widetilde{X}_{i}$'s are not fully observed. Instead, we observe point evaluations
$$\widetilde{X}_{i,d} = (\widetilde{X}_{i}(t_{1}),\widetilde{X}_{i}(t_{2}),\ldots,\widetilde{X}_{i}(t_{r}))',\qquad i=1,...,n.$$ 
Here, $0 \leq t_{1} < t_{2} < \ldots < t_{r} \leq 1$ is a grid over $[0,1]$, assumed asymptotically homogeneous in that  $\max_{1 \leq j \leq r-1} (t_{j+1} - t_{j}) = O(r^{-1})$ as $r \rightarrow \infty$. 
The latent discrete process is denoted by $X_{i,d} = (X_{i}(t_{1}),X_{i}(t_{2}),\ldots,X_{i}(t_{r}))'$.

Our strategy will be to mimic Steps 1--5 from the fully observed setup. Since the $X_{i}$'s are no longer fully observed, though, in order to have versions of the $F_{i}$ and $\widetilde{F}_i$, we will draw inspiration from the general definition of the local variation distribution (Equation \ref{general_definition} in Definition \ref{local_variation}). First, define 
\begin{eqnarray*}
F_{i,d}(t) = \sum_{j \in \mathscr{I}_{t}} |X_{i}(t_{j+1}) - X_{i}(t_{j})| \bigg/ \sum_{j=1}^{r-1} |X_{i}(t_{j+1}) - X_{i}(t_{j})|
\end{eqnarray*}
for $t \in [0,1]$ and each $i=1,2,\ldots,n$, where $\mathscr{I}_{t}$ is the set of all $j$'s satisfying $t_{j+1} \leq t$. Note that because we only observe each curve over the grid $0 \leq t_{1} < t_{2} < \ldots < t_{r} \leq 1$, we have replaced the supremum over all grids in Equation \ref{general_definition} of Definition \ref{local_variation} by just this one (the finest grid we get to observe). 
Clearly, $F_{d}$ has jump discontinuities at the grid points $t_{j}$'s, is c\`adl\`ag, and satisfies $F_{d}(t) = 0$ for all $t \in [0,t_{2})$ and $F_{d}(t) = 1$ for all $t \in [t_{r},1]$. \\ 
\indent For the (discretely) observable warped process, we define
\begin{eqnarray}
\widetilde{F}_{i,d}(t) &=& \sum_{j \in \mathscr{I}_{t}} |\widetilde{X}_{i}(t_{j+1}) - \widetilde{X}_{i}(t_{j})| \bigg/ \sum_{j=1}^{r-1} |\widetilde{X}_{i}(t_{j+1}) - \widetilde{X}_{i}(t_{j})|,  \label{eqn0} 
\end{eqnarray}
The $\widetilde{F}_{i,d}$'s also have jump discontinuities at the grid points, and are c\`adl\`ag. \\
\indent Under the identifiable regime, in particular, we would have $F_{i,d}(t) = F_{d}(t)$ for all $i=1,2,\ldots,n$, where 
\begin{eqnarray*}
F_{d}(t) = \sum_{j \in \mathscr{I}_{t}} |\phi(t_{j+1}) - \phi(t_{j})| \bigg/ \sum_{j=1}^{r-1} |\phi(t_{j+1}) - \phi(t_{j})|.
\end{eqnarray*}
Its jumps are at most of size $a_{r} = \max_{1 \leq j \leq r-1} |\phi(t_{j+1}) - \phi(t_{j})| / \sum_{j=1}^{r-1} |\phi(t_{j+1}) - \phi(t_{j})|$. Moreover, in the identifiable regime,
\begin{eqnarray*}
\widetilde{F}_{i,d}(t) &=& \sum_{j \in \mathscr{I}_{t}} |\phi(s_{i,j+1}) - \phi(s_{i,j})| \bigg/ \sum_{j=1}^{r-1} |\phi(s_{i,j+1}) - \phi(s_{i,j})|,
\end{eqnarray*}
where $s_{i,j} = T_{i}^{-1}(t_{j})$ for each $i$ and $j$ are unobserved random variables. The maximum jump size of $\widetilde{F}_{i,d}$ is $A_{i,r} = \max_{1 \leq j \leq r-1} |\phi(s_{i,j+1}) - \phi(s_{i,j})| / \sum_{j=1}^{r-1} |\phi(s_{i,j+1}) - \phi(s_{i,j})|$.

With the general definitions of $F_{i,d}$ and $\widetilde{F}_{i,d}$ in place, we can now adapt Steps 1--5 to the discrete case. In what follows, the generalized inverse of a function $G$ is denoted by $G^{-}$, i.e., $G^{-}(t) = \inf\{u : G(u) \geq t\}$.
The first two steps will remain invariant, except for the fact that they will now employ the discrete local variation measures. This means that we will not require any tuning parameters or smoothness assumptions to estimate the warp and registration maps. The registration itself (the last three steps) will require some smoothing, of course, if it is to make sense:

\begin{itemize}
\item[Step $1^*$:] 
Set $\widehat{F}_{d} = \{n^{-1} \sum_{i=1}^{n} \widetilde{F}_{i,d}^{-}\}^{-}$ and $\widehat{F}_{d}^{*} = n^{-1} \sum_{i=1}^{n} \widetilde{F}_{i,d}^{-}$. 
\end{itemize}
Note that under the identifiable regime (\ref{identifiable_regime}), $\widehat{F}_{d}$ mimics $F_{d}$.
\begin{itemize}
\smallskip
\item[Step $2^*$:] Predict the random warp map $T_{i}$ by $\widehat{T}_{i,d} = \widetilde{F}_{i,d}^{-} \circ \widehat{F}_{d}$ and the registration map $T_{i}^{-1}$ by $\widehat{T}_{i,d}^{*} = \widehat{F}_{d}^{*} \circ \widetilde{F}_{i,d} = \{n^{-1} \sum_{i=1}^{n} \widetilde{F}_{i,d}^{-}\} \circ \widetilde{F}_{i,d}$. 

\smallskip
\item[Step $3^*$:] Since the $\widetilde{X}_{i}$'s are observed discretely, we do not have information about their values between grid points. Thus, we first smooth each of the $\widetilde{X}_{i,d}$ using the Nadaraya-Watson kernel regression estimator for an appropriately chosen kernel $k$ and bandwidth $h$, denoting resulting smoothed functions by $X^{\dagger}_{i}$, 
$$X^{\dagger}_{i}(t) = \sum_{j=1}^{r} k\left(\frac{t - t_{j}}{h}\right)\widetilde{X}_{i}(t_{j}) \bigg/ \sum_{j=1}^{r} k\left(\frac{t - t_{j}}{h}\right).$$
Define 
$$\widehat{X}^{*}_{i}(t) = X^{\dagger}_{i}(\widehat{T}_{i,d}(t)),\quad i=1,2,\ldots,n$$ 
to be the registered functional observations and write $\overline{X}_{r*} = n^{-1}\sum_{i=1}^{n} \widehat{X}^{*}_{i}$ for their mean. 
\end{itemize}
As in the fully observed situation, if we suspect to be in the identifiable regime (\ref{identifiable_regime}), we estimate the pairs $\{\phi,\xi_i\}$ as follows:
\begin{itemize}
\smallskip
\item[Step $4^*$:] Compute the empirical covariance operator $\widehat{\mathscr{K}}_{r*}$ of the registered curves $\widehat{X}^{*}_{i}$, and use its leading eigenfunction $\widehat{\phi}_{*}$ as the estimator of $\phi$ (again, assume the convention that the sign is correctly identified, i.e., $\langle\widehat{\phi}_{*},\phi\rangle \geq 0$). 

\smallskip
\item[Step $5^*$:] Finally, estimate $\xi_{i}$ by $\widehat{\xi}_{i*} = \langle\widehat{X}^{*}_{i},\widehat{\phi}_{*}\rangle$ for each $i \geq 1$. 
\end{itemize}

We should point out here that our method is also straightforwardly applicable in the situation where the grid over which the $\widetilde{X}_{i}$'s are observed, say, $0 \leq t_{i,1} < t_{i,2} < \ldots < t_{i,r_{i}} \leq 1$, differs with $i$. The reason for this compatibility is the fact that our approach considers only one curve at a time. We formulate it in the notationally simpler case of a common grid, in order to alleviate the notation in the statement of our asymptotic results in Section \ref{asymptotics}. Additional remarks on the practical implementation are included in the Supplementary material (Section \ref{sec:implementation}).

\subsection{Discrete Observation With Measurement Error} \label{subsec2-3}

\indent It can often happen that the discretely observed functional data be additionally contaminated by measurement error. In this case, one has to suitably adapt the registration procedure. In the presence of measurement error, we observe $Y_{i,d} = \widetilde{X}_{i,d} + e_{i}$, where $\widetilde{X}_{i,d}$ was defined in Section \ref{subsec2-2}, and $e_{i} = (\epsilon_{i,1},\epsilon_{i,2},\ldots,\epsilon_{i,r})'$ with the $\{\epsilon_{i,j} : j=1,2,\ldots,r, \ i=1,2,\ldots,n\}$ being a collection of i.i.d. error variables with zero mean and variance $\sigma^{2}$, independent of the processes and warp maps. \\
\indent We will modify the registration procedure as follows. First, construct a non-parametric function estimator of $\widetilde{X}_{i}'$, which is the derivative of the warped process $\widetilde{X}_{i}$, using the observation $Y_{i,d}$ for each $i$, and call this estimator $\widehat{X}_{i,w}^{(1)}(\cdot)$. Define analogues of the  $\widetilde{F}_{i}$'s as $$\widetilde{F}_{i,w}(t) = \int_{0}^{t} |\widehat{X}_{i,w}^{(1)}(u)|du \bigg/ \int_{0}^{1}|\widehat{X}_{i,w}^{(1)}(u)|du, \ \ \ t \in [0,1].$$ Note that unlike the discrete observation case described in the previous section, we now have fully functional versions of $\widetilde{X}_{i}'$ for each $i$, which allows us to mimic the algorithm in the fully observed scenario in Section \ref{subsec2-1}.

\begin{itemize}
\item[Step $1^{**}$:] Set $\widehat{F}_{e} = \left(n^{-1} \sum_{i=1}^{n} \widetilde{F}_{i,w}^{-1}\right)^{-1}$.
\end{itemize}
Under the identifiable regime (\ref{identifiable_regime}), in particular, we have $\widehat{F}_{e}$ estimates $F_{\phi}$.
\begin{itemize}
\smallskip
\item[Step $2^{**}$:] Predict the warp map $T_{i}$ by $\widehat{T}_{i,e} = \widetilde{F}_{i,w}^{-1} \circ \widehat{F}_{e}$, and the registration map by $\widehat{T}_{i,e}^{-1}$.

\smallskip
\item[Step $3^{**}$:] Construct non-parametric function estimators of the $\widetilde{X}_{i}$'s using the $Y_{i,d}$'s, and call them $\widehat{X}_{i,w}(\cdot)$'s. Define $\widehat{X}_{i,e}^{*}(t) = \widehat{X}_{i,w}(\widehat{T}_{i,e}(t)), i=1,2,\ldots,n$ to be the registered functional observations.
\end{itemize}
If we suspect to be in the identifiable regime (\ref{identifiable_regime}), we estimate the pairs $\{\phi,\xi_i\}$ as follows:
\begin{itemize}
\smallskip
\item[Step $4^{**}$:] Write $\overline{X}_{e*} = n^{-1} \sum_{i=1}^{n} \widehat{X}_{i,e}$ for the mean of the registered observations and let $\widehat{\mathscr{K}}_{e*}$ denote their empirical covariance operator. Take its leading eigenfunction, denoted by $\widehat{\phi}_{e*}$, as the estimator of $\phi$ (assuming the same sign convention as earlier).

\smallskip
\item[Step $5^{**}$:] Finally, estimate $\xi_{i}$ by $\widehat{\xi}_{i*,e} = \langle\widehat{X}_{i,e},\widehat{\phi}_{e*}\rangle$ for each $i \geq 1$.
\end{itemize}

There are two smoothing steps involved in the above algorithm. Given the large literature on non-parametric smoothing techniques, one can choose any smoother. However, the asymptotic results will depend on the efficiency of the chosen smoothing techniques. From now on in this paper, we will use a local quadratic regression approach with kernel $k_{1}(\cdot)$ and bandwidth $h_{1}(\cdot)$ for finding $\widehat{X}_{i,w}^{(1)}$. We will then use a local linear estimator with kernel $k_{2}(\cdot)$ and bandwidth $h_{2}(\cdot)$ for estimating $\widehat{X}_{i,w}$. These choices are motivated by the advantages of local polynomial estimators in dealing with boundary effects (see, e.g., \cite{FG96} and \cite{WJ95} for further details on various smoothing techniques). More details on the choices of smoothing parameters are given in Remark 4 after Theorem 5.

\begin{remark}[Presence of an additional vertical shift, as in Remark \ref{identifiable-rank-two-model}] \label{remark-regis-rank-two}
In Remark \ref{identifiable-rank-two-model}, we discussed a trivial rank two identifiable model given by $X(t) = \zeta + \xi\phi(t)$, which is simply the inclusion of a random vertical shift in the rank 1 model. The methods presented still apply in this case. For simplicity, consider fully observed warped data $\widetilde{X}_i := X_i \circ T_i^{-1}$, $1 \leq i \leq n$, from the model with vertical shifts. Since the registration procedure based on local variation measure depends on the process only through its derivative, the estimators $(\hat\xi_i,\hat\phi,\hat T_i)$ of $(\xi_i,\phi,T_i)$ obtained from $\widetilde{X}_i(t) = \zeta_i + \xi_i\phi(T^{-1}_i(t))$ will be exactly the same as those obtained from $\widetilde{X}_i(t) =  \xi_i\phi(T^{-1}_i(t))$ considered earlier. The vertical shifts $\zeta_i$ can then be easily estimated by $\widehat{\zeta}_i=\int_0^1 \{\widetilde{X}_i(t) - \widehat{\xi}_i\widehat{\phi}(\widehat{T}_i^{-1}(t))\}dt$. 
Indeed, $|\widehat{\zeta}_i - \zeta_i|$ will converge to zero almost surely as $n \rightarrow \infty$ (as a corollary to Theorem \ref{thm2} in Section \ref{asymptotics}).
\end{remark}

\section{Asymptotic Theory}\label{asymptotics}

We next study the asymptotic properties of the estimators obtained above. 
We develop separate results for each of the three observation regimes considered (full observation, discrete observation, discrete observation with measurement errors). 
In what follows, the space $C^{1}[0,1]$ is equipped with the norm $|||f|||_{1} = ||f||_{\infty} + ||f'||_{\infty}$, where $||\cdot||_{\infty}$ is the usual sup-norm. The 2-Wasserstein distance between distributions $G_1$ and $G_2$ will be denoted by $d_W(G_1,G_2)=\sqrt{\int_{0}^1 \big(G_1^{-}(u)-G^{-}_2(u)\big)^2du}$. All the proofs are developed in Section \ref{proofs} of the Appendix (Supplementary Material). 

\subsection{Identifiable Regime}\label{specified_asymptotics}

We first focus on the identifiable regime as given in Definition \ref{identifiable_regime}. Our first two results concern the fully observed case, as described in Section \ref{subsec2-1}. Write $\mu = E(X_{1}) = E(\xi_{1})\phi$ for the mean function and $\mathscr{K} = E(X_{1} \otimes X_{1}) - \mu \otimes \mu$ for the covariance operator, where $(f\otimes g)h=\langle g,h\rangle f $ for any triple $f,g,h\in L_2[0,1]$. Let $|||\cdot|||$ denote the trace norm for operators on $L_{2}[0,1]$. The covariance kernel of $X$ is denoted by $K(\cdot,\cdot)$ and the empirical covariance kernel of the $\widehat{X}_{i}$'s is denoted by  $\widehat{K}_{r}(\cdot,\cdot)$.
\begin{theorem}[Strong Consistency -- Fully Observed Case] \label{thm2}
Further to the assumptions in Definition \ref{identifiable_regime}, assume also that $\phi'$ is H\"older continuous with exponent $\alpha \in (0,1]$. Then, the estimators in Section \ref{subsec2-2} satisfy the following asymptotic results, where convergence is always with probability one:
\begin{itemize}
\item[(a)] $d^{2}_{W}(\widehat{F},F_{\phi}) \rightarrow 0$ as $n \rightarrow \infty$. 
\item[(b)] $||\widehat{T}_{i}^{-1} - T_{i}^{-1}||_{\infty} \rightarrow 0$ and $||\widehat{T}_{i} - T_{i}||_{\infty} \rightarrow 0$ as $n \rightarrow \infty$ for each $i \geq 1$. 
\item[(c)] $||\widehat{X}_{i} - X_{i}||_{\infty} \rightarrow 0$ as $n \rightarrow \infty$ for each $i \geq 1$.  
\item[(d)] $d^{2}_{W}(\widehat{F}_{i},F_{\phi}) \rightarrow 0$ as $n \rightarrow \infty$ for each $i \geq 1$, where $\widehat{F}_{i}$ is the local variation measure associated with $\widehat{X}_{i}$.
\item[(e)] $||\overline{X}_{r} - \mu||_{\infty} \rightarrow 0$ as $n \rightarrow \infty$, where $\overline{X}_{r} = n^{-1} \sum_{i=1}^{n} \widehat{X}_{i}$. 
\item[(f)] $|||\widehat{\mathscr{K}}_{r} - \mathscr{K}||| \rightarrow 0$ and $||\widehat{K}_{r} - K||_{\infty} = \sup_{s,t \in [0,1]} |\widehat{K}_{r}(s,t) - K(s,t)| \rightarrow 0$ as $n \rightarrow \infty$. Moreover, $||\widehat{\phi} - \phi||_{\infty} \rightarrow 0$ and $|\widehat{\xi}_{i} - \xi_{i}| \rightarrow 0$ as $n \rightarrow \infty$ for each $i \geq 1$.
\end{itemize}
Furthermore, if we additionally assume that $E(||T_{1}'||_{\infty}) < \infty$ and $\inf_{t \in [0,1]} T'(t) \geq \delta > 0$ almost surely for a deterministic constant $\delta$ (call this ``Condition 1"), then the following stronger results hold with probability one, in lieu of (b), (c), and (e):
\begin{itemize}
\item[(b')] $|||\widehat{T}_{i}^{-1} - T_{i}^{-1}|||_{1} \rightarrow 0$ and $|||\widehat{T}_{i} - T_{i}|||_{1} \rightarrow 0$ as $n \rightarrow \infty$ for each $i \geq 1$.

\item[(c')] $|||\widehat{X}_{i} - X_{i}|||_{1} \rightarrow 0$ as $n \rightarrow \infty$ for each $i \geq 1$.

\item[(e')]  $|||\overline{X}_{r} - \mu|||_{1} \rightarrow 0$ as $n \rightarrow \infty$, where $\overline{X}_{r} = n^{-1} \sum_{i=1}^{n} \widehat{X}_{i}$. 
\end{itemize}
\end{theorem}

\begin{remark} \label{rem1}
\normalfont A few remarks concerning the last theorem are as follows:
\begin{description}
\item[\normalfont{\emph{Independence}}]
The strong consistency results in Theorem \ref{thm2} \textit{do not} require that $\xi_{i}$ and $T_{i}$ are independent. 

\item[\normalfont{\emph{Uniformity}}] It is observed from the proof of the uniform convergence of $\widehat{T}_{i}^{-1}$ in part (b) of the above theorem that 
$$\max_{1 \leq i \leq n} ||\widehat{T}_{i}^{-1} - T_{i}^{-1}||_{\infty} \rightarrow 0,\qquad \mbox{as }\,n \rightarrow \infty$$ 
almost surely. Under Condition 1, the same conclusion is true now with the finer norm $|||\cdot|||_{1}$. The convergence in part (d) also holds uniformly for all $i=1,2,\ldots,n$.

\item[\normalfont{\emph{Fisher Consistency}}] It can be directly verified that $\widehat{F}^{-1} = \overline{T} \circ F_{\phi}^{-1}$ so that $\widehat{F} = F_{\phi} \circ \overline{T}^{-1}$, where  $\overline{T} = n^{-1}\sum_{i=1}^{n} T_i$. Also, $\widehat{T}_{i} = T_{i} \circ \overline{T}^{-1}$,  $\widehat{T}_{i}^{-1} = \overline{T} \circ T_{i}^{-1}$, and $\widehat{X}_{i} = \xi_{i}\phi \circ \overline{T}^{-1}$ for each $i$. Further, $\widehat{\mathscr{K}}_{r} = n^{-1}\sum_{i=1}^{n} (\widehat{X}_{i} - \overline{X}_{r}) \otimes (\widehat{X}_{i} - \overline{X}_{r}) = \{n^{-1}\sum_{i=1}^{n} \xi_{i}^{2} - \overline{\xi}^{2}\} (\phi \circ \overline{T}^{-1}) \otimes (\phi \circ \overline{T}^{-1})$, where $\overline{\xi} = n^{-1}\sum_{i=1}^{n} \xi_{i}$. Thus, $\widehat{\phi} = (\phi \circ \overline{T}^{-1})/||(\phi \circ \overline{T}^{-1})||$, and $\widehat{\xi}_{i} = \langle \widehat{X}_{i},\widehat{\phi}\rangle = \xi_{i}||\phi \circ \overline{T}^{-1}||$. Since all of the above estimators are measurable functions of the sample averages of the $T_{i}$'s, the $\xi_{i}$'s and the $\xi_{i}^{2}$'s, it follows that all of the above estimators are Fisher consistent for their population counterpart.

\item[\normalfont{\emph{A Concrete Example}}] The condition $\inf_{t \in [0,1]} T'(t) \geq \delta > 0$ almost surely for a deterministic constant $\delta$ can be relaxed to $\inf_{t \in [0,1]} T'(t) \geq \delta_{i}$ almost surely for i.i.d. positive random variables $\delta_{i}$ provided we assume that $E(\delta_{1}^{-1}) < \infty$. An example of random warp functions that satisfy $\inf_{t \in [0,1]} T'(t) \geq \delta > 0$ can be found Section 8 of \cite{PZ16}. Define $\zeta_{0}(t) = t$ and for $k \neq 0$, define $\zeta_{k}(t) = t - sin({\pi}kt)/(|k|{\pi}\beta)$ for some $\beta > 0$. If $K$ is an integer-valued, symmetric random variable, then $E(\zeta_{K}) = Id$. For a fixed $J \geq 2$, let $\{K_{j}\}_{j=1}^{J}$ be i.i.d. integer-valued, symmetric random variables, and $\{U_{j}\}_{j=1}^{J-1}$ be i.i.d. $Unif[0,1]$ random variables independent of the $K_{j}$'s. Define $T(t) = U_{(1)}\zeta_{K_{1}}(t) + \sum_{j=1}^{J-1} (U_{(j)} - U_{(j-1)})\zeta_{K_{j}}(t) + (1 - U_{(J-1)})\zeta_{K_{J}}(t)$. Then, $T$ is a strictly increasing homeomorphism on $[0,1]$, $T \in C^{1}[0,1]$ surely, $E(T) = Id$. Further, it can be easily shown that $\inf_{t \in [0,1]} T'(t) \geq 1 - \beta^{-1}$. Thus, the condition $\inf_{t \in [0,1]} T'(t) \geq \delta > 0$ holds if we choose $\beta = (1-\delta)^{-1}$. 
\end{description}
\end{remark}

\noindent Further to strong consistency, we also derive weak convergence of the estimators: 

\begin{theorem}[Weak Convergence -- Fully Observed Case] \label{thm3}
Further to assumptions in Definition \ref{identifiable_regime}, assume also that $\phi'$ is H\"older continuous with exponent $\alpha \in (0,1]$, that $\xi_{i}$ and $T_{i}$ are independent for each $i$, and that $E(||T_{1}'||_{\infty}^{2}) < \infty$. Then, the estimators in Section \ref{subsec2-1} satisfy the following asymptotic results,   
\begin{itemize}
\item[(a)] $nd^{2}_{W}(\widehat{F},F_{\phi})$ converges weakly as $n \rightarrow \infty$. 
\item[(b)] $\sqrt{n}(\widehat{T}_{i}^{-1} - T_{i}^{-1})$ and $\sqrt{n}(\widehat{T}_{i} - T_{i})$ converge weakly in the $C[0,1]$ topology as $n \rightarrow \infty$ for each $i \geq 1$. 
\item[(c)] $\sqrt{n}(\widehat{X}_{i} - X_{i})$ converges weakly in the $C[0,1]$ topology as $n \rightarrow \infty$ for each $i \geq 1$. 
\item[(d)] $nd^{2}_{W}(\widehat{F}_{i},F_{\phi})$ converges weakly as $n \rightarrow \infty$ for each $i \geq 1$. 
\item[(e)] $\sqrt{n}(\overline{X}_{r} - \mu)$ converges weakly to a zero mean Gaussian distribution in the $C[0,1]$ topology as $n \rightarrow \infty$. 
\item[(f)] $\sqrt{n}(\widehat{\mathscr{K}}_{r} - \mathscr{K})$ converges weakly in the topology of Hilbert-Schmidt operators, and $\sqrt{n}(\widehat{K}_{r} - K)$ converges weakly in the $C([0,1]^{2})$ topology as $n \rightarrow \infty$. In both cases, the limits are zero mean Gaussian distributions. Moreover, $\sqrt{n}(\widehat{\phi} - \phi)$ converges weakly to a zero mean Gaussian distribution in the $C[0,1]$ topology, and $\sqrt{n}(\widehat{\xi}_{i} - \xi_{i})$ converges weakly as $n \rightarrow \infty$ for each $i \geq 1$.
\end{itemize}
\end{theorem}
\indent Since $C([0,1]^{k})$ is a stronger topology than $L_{2}([0,1]^{k})$ for any finite $k = 1,2,\ldots$, it follows that the weak convergence results in the above theorem which hold in the $C([0,1]^{k})$ topology also hold in the $L_{2}([0,1]^{k})$ topology by virtue of the continuous mapping theorem.

\indent We shall now study some the asymptotic properties of the estimators in the discrete observation setup (without measurement error). 
\begin{theorem}[Limit Theory -- Discretely Observed Case Without Measurement Error] \label{thm4}
Further to the conditions of Theorem \ref{thm3}, assume that $\phi \in C^{2}[0,1]$, $\int_{0}^{1} |\phi'(u)|^{-\epsilon} < \infty$ for some $\epsilon > 0$, and that $\inf_{t \in [0,1]} T'(u) \geq \delta > 0$ almost surely for a deterministic constant $\delta$. Define $\alpha = \epsilon/(1+\epsilon)$. Assume that $\xi_{i}$ and $T_{i}$ are independent for each $i$ (only for the weak convergence statements). The kernel $k(\cdot)$ is assumed to be supported on $[-1,1]$. If $h = h(n) = o(n^{-1/2})$ and $r = r(n)$ satisfies $r >> n^{1/2\alpha}$ as $n \rightarrow \infty$, then the estimators introduced in Section \ref{subsec2-2} satisfy
\begin{itemize}
\item[(a)] $d^{2}_{W}(\widehat{F}_{d}^{*},F_{\phi}) \rightarrow 0$ as $n \rightarrow \infty$ almost surely, and $d^{2}_{W}(\widehat{F}_{d}^{*},F_{\phi}) = O_{P}(n^{-1})$ as $n \rightarrow \infty$. 
\item[(b)] $||\widehat{T}_{i,d}^{*} - T_{i}^{-1}||_{\infty} \rightarrow 0$ and $||\widehat{T}_{i,d} - T_{i}||_{\infty} \rightarrow 0$ as $n \rightarrow \infty$ almost surely. Further, $\sqrt{n}(\widehat{T}_{i,d}^{*} - T_{i}^{-1})$ and $\sqrt{n}(\widehat{T}_{i,d} - T_{i})$ converge weakly in the $L_{2}[0,1]$ topology as $n \rightarrow \infty$ for each $i \geq 1$. 
\item[(c)] $||\widehat{X}^{*}_{i} - X_{i}||_{\infty} \rightarrow 0$ as $n \rightarrow \infty$ almost surely, and $\sqrt{n}(\widehat{X}^{*}_{i} - X_{i})$ converges weakly in the $L_{2}[0,1]$ topology as $n \rightarrow \infty$ for each $i \geq 1$. 
\item[(d)] $d^{2}_{W}(\widehat{F}_{i}^{*},F_{\phi}) \rightarrow 0$ as $n \rightarrow \infty$ almost surely, and $d^{2}_{W}(\widehat{F}_{i}^{*},F_{\phi}) = O_{P}(n^{-1})$ as $n \rightarrow \infty$ for each $i \geq 1$. 
\item[(e)] $||\overline{X}_{r*} - \mu||_{\infty} \rightarrow 0$ as $n \rightarrow \infty$ almost surely, and $\sqrt{n}(\overline{X}_{r*} - \mu)$ converges weakly in the $L_{2}[0,1]$ topology as $n \rightarrow \infty$. 
\item[(f)] $|||\widehat{\mathscr{K}}_{r*} - \mathscr{K}||| \rightarrow 0$ as $n \rightarrow \infty$ almost surely, and $\sqrt{n}(\widehat{\mathscr{K}}_{r*} - \mathscr{K})$ converges weakly in the topology of Hilbert-Schmidt operators. Further, $||\widehat{K}_{r*} - K||_{\infty} \rightarrow 0$ as $n \rightarrow \infty$, and $\sqrt{n}(\widehat{K}_{r*} - K)$ converges weakly in the $L_{2}([0,1]^{2})$ topology as $n \rightarrow \infty$. Moreover, $||\widehat{\phi}_{*} - \phi|| \rightarrow 0$ as $n \rightarrow \infty$ almost surely, and $\sqrt{n}(\widehat{\phi}_{*} - \phi)$ converges weakly in the $L_{2}[0,1]$ topology. Also, $|\widehat{\xi}_{i*} - \xi_{i}| \rightarrow 0$ as $n \rightarrow \infty$ almost surely, and $\sqrt{n}(\widehat{\xi}_{i*} - \xi_{i})$ converges weakly as $n \rightarrow \infty$ for each $i \geq 1$. \\
\noindent In all the weak convergence results stated above, the limits are identical to the corresponding limits obtained in the fully observed scenario in Theorem \ref{thm3}.
\end{itemize}
\end{theorem}

\begin{remark} \label{rem2}
\normalfont 
\begin{description}
\item[\normalfont{\emph{Independence}}] As in the fully observed setting in Theorem \ref{thm2}, the strong consistency results in the discrete, noiseless observation setting in Theorem \ref{thm4} \textit{do not} require $\xi_{i}$ and $T_{i}$ to be independent. 

\item[\normalfont{\emph{Variable observation grids}}] The asymptotic results remain valid in the case where the grid over which the $\widetilde{X}_{i}$'s are observed, say, $0 \leq t_{i,1} < t_{i,2} < \ldots < t_{i,r_{i}} \leq 1$, differs with $i$. The proof, however, will be notationally quite cumbersome. In this case, the requirement on the grid will be as follows: $\max_{1 \leq j \leq r_{i}-1} (t_{j+1} - t_{j}) = O(r_{i}^{-1})$ as $r_{i} \rightarrow \infty$ for each $i$, and $\widetilde{r}_{n} := \min_{1 \leq i \leq n} r_{i}$ satisfies $\widetilde{r}_{n} >> n^{1/2\alpha}$ as $n \rightarrow \infty$. 

\item[\normalfont{\emph{Smoothing parameter choice}}] The choice of $h$ in Theorem \ref{thm4} is an under-smoothing choice. It is made on account of the absence of measurement errors in the observations, which enables us to under-smooth the data without damaging $\sqrt{n}$-consistency. This is unlike what happens in classical non-parametric regression due to the presence of errors in that scenario. Also, the boundary points inflate the bias of the Nadaraya-Watson estimator to an order of $h$ (the same order as that obtained in Theorem \ref{thm4} for all points). However, these issues are of no consequence in this scenario. It is also natural to under-smooth in this situation since appropriate under-smoothing retains the features of the curves better and allows estimation at a parametric rate even under non-parametric smoothing. If instead of the Nadaraya-Watson estimator, one uses a local linear estimator with bandwidth $h$, then the bias is of order $h^{2}$ (even at the boundaries). In this case, $h$ has to be $o(n^{-1/4})$ to achieve parametric rates of convergence, which is again an under-smoothing choice. Thus, the choice of smoothing method does not play a crucial role in this setup.

\item[\normalfont{\emph{Topology of convergence}}] Unlike Theorem \ref{thm3}, the weak convergence results are all in the $L_{2}$ topology. This is because unlike the fully observed case, the estimators involved are not continuous functions in $[0,1]$. We could not consider the weaker $D[0,1]$ topology since not all estimators will be c\`adl\`ag functions. However, we still retain the strong consistency results in parts (b), (c) and (e) in the sup norm similar to Theorem \ref{thm2}. This is due to the fact that those estimators are uniformly bounded almost surely, and thus have finite sup-norm. Further, in all cases, there is no issue with the measurability of the supremum. 

\item[\normalfont{\emph{Assumption on eigenfunction}}] The condition $\phi \in C^{2}[0,1]$ can be relaxed to requiring that $\phi'$ is Lipschitz continuous. Moreover, the requirement $\int_{0}^{1} |\phi'(u)|^{-\epsilon} < \infty$ for some $\epsilon > 0$ is not restrictive. Of course, it holds if $\phi'$ is bounded away from zero on $[0,1]$, in which case one can choose $\alpha = 1$. Consider the case when $\phi \in C^{2}[0,1]$ and let $t_{0} \in (0,1)$ be such that $\phi'(t_{0}) = 0$. If $\phi''(t_{0}) > 0$, then we can choose an interval $A_{\delta} = (t_{0} - \delta,t_{0} + \delta) \subset (0,1)$ such that $\inf_{u \in A_{\delta}} |\phi''(u)| \geq \beta > 0$. Then, a first order Taylor expansion yields $\int_{A_{\delta}} |\phi'(t)|^{-\epsilon}dt \leq \beta^{-\epsilon} \int_{A_{\delta}} |t - t_{0}|^{-\epsilon}dt < \infty$ for any $\epsilon < 1$. Here, we have used the fact that $\int_{0}^{\delta} t^{-\epsilon}dt < \infty$ for any $\delta > 0$ iff $\epsilon < 1$. Thus, if none of the zeros of $\phi'$ and $\phi''$ coincide, then the condition $\int_{0}^{1} |\phi'(u)|^{-\epsilon} < \infty$ holds for any $\epsilon < 1$. In general, if $\phi \in C^{m}[0,1]$ for some $m \geq 2$, and $m'$ be the least integer between $2$ and $m$ such that none of the zeros of $\phi'$ and $\phi^{(m')}$ coincide, then $\int_{0}^{1} |\phi'(u)|^{-\epsilon} < \infty$ holds for any $\epsilon < 1/(m'-1)$.  
\end{description}
\end{remark}

\indent We finally study the asymptotic properties of the estimators in the modified registration procedure employed when one has contamination by measurement error (described in Section \ref{subsec2-3}).
\begin{theorem}[Limit Theory -- Measurement Error Case] \label{thm-error}
In addition to the assumptions of Theorem \ref{thm3}, assume that $\phi \in C^{4}[0,1]$, $\int_{0}^{1} |\phi'(u)|^{-\epsilon}du < \infty$ for some $\epsilon > 0$. Define $\alpha = \epsilon/(1+\epsilon)$. Assume that $\xi_{i}$ and $T_{i}$ are independent for each $i$. Suppose that $T \in C^{4}[0,1]$ a.s. and $\inf_{t \in [0,1]} T'(u) \geq \delta > 0$ almost surely for a deterministic constant $\delta$. The kernels $k_{1}(\cdot)$ and $k_{2}(\cdot)$ are assumed to be supported on $[-1,1]$, symmetric and continuously differentiable. The errors $\{\epsilon_{ij}\}$ are assumed to be a.s. bounded. Also assume that $E\{|\xi_{1}|^{-2\alpha/(2-\alpha)}\} < \infty$ as well as each  of $E(||T_{1}''||_{\infty}^{2})$, $E(||T_{1}'''||_{\infty}^{2})$ and $E(||T_{1}^{(4)}||_{\infty}^{2})$ are finite. The bandwidths satisfy $h_{1}, h_{2} \rightarrow 0$, $rh_{1}^{3}, rh_{2} \rightarrow \infty$. Then, the estimators in Section \ref{subsec2-3} satisfy the following properties.
\begin{itemize}
\item[(a)] $d^{2}_{W}(\widehat{F}_{e},F_{\phi}) = O_{P}(h_{1}^{4\alpha} + (rh_{1}^{3})^{-\alpha} + n^{-1})$ as $n \rightarrow \infty$. 
\item[(b)] Both $||\widehat{T}_{i,e}^{-1} - T_{i}^{-1}||_{\infty}$ and $||\widehat{T}_{i,e} - T_{i}||_{\infty}$ are $O_{P}(h_{1}^{2\alpha} + (rh_{1}^{3})^{-\alpha/2} + n^{-1/2})$ as $n \rightarrow \infty$. 
\item[(c)] $||\widehat{X}_{i,e}^{*} - X_{i}||_{\infty} = O_{P}(h_{1}^{2\alpha} + (rh_{1}^{3})^{-\alpha/2} + h_{2}^{2} + (rh_{2})^{-1/2} + n^{-1/2})$ as $n \rightarrow \infty$. 
\item[(d)] $||\overline{X}_{e*} - \mu||_{\infty} = O_{P}(h_{1}^{2\alpha} + (rh_{1}^{3})^{-\alpha/2} + h_{2}^{2} + (rh_{2})^{-1/2} + n^{-1/2})$ as $n \rightarrow \infty$. 
\item[(e)] $|||\widehat{\mathscr{K}}_{e*} - \mathscr{K}||| = O_{P}(h_{1}^{2\alpha} + (rh_{1}^{3})^{-\alpha/2} + h_{2}^{2} + (rh_{2})^{-1/2} + n^{-1/2})$ as $n \rightarrow \infty$. Consequently, $||\widehat{\phi}_{e*} - \phi||$ and $|\widehat{\xi}_{i*,e} - \xi|$ have the same rates of convergence for each fixed $i$. 
\end{itemize}
\end{theorem}

\begin{remark} \label{rem-err}
\normalfont 
\begin{description}
\item[\normalfont{\emph{Smoothing techniques}}] Analogous rates of convergence can also be obtained if one uses different non-parametric smoothing techniques than the ones in the theorem. One may, e.g., use a Nadaraya-Watson estimator in Step 3** with boundary kernels to alleviate the boundary bias problem that is well-known for this estimator (see, e.g., \cite{WJ95}). Also, to estimate $\widetilde{X}_{i}'$, one may use higher order local polynomials with even orders. However, these will be computationally more intensive as well as need additional smoothness assumptions on the latent process and the warp maps.

\item[\normalfont{\emph{Rates of convergence}}] It is observed in the above theorem that the rates of convergence are slower than the parametric rates achieved in the earlier settings due to the non-parametric smoothing steps involved -- especially the estimation of derivatives, which is known to have quite slow rates of convergence. Further, the contributions of the two smoothing steps in the convergence rates are clear. It is well known in local linear regression that the optimal rate for $h_{1}$ is $r^{-1/7}$ and that for $h_{2}$ is $r^{-1/5}$. With these rates, we have $d^{2}_{W}(\widehat{F}_{e},F_{\phi}) = O_{P}(r^{-4\alpha/7} + n^{-1})$, and the remaining quantities are $O_{P}(r^{-2\alpha/7} + n^{-1/2})$. Thus, parametric rates of convergence is achieved if $r > n^{7/4\alpha}$.

\item[\normalfont{\emph{Assumption on principal component}}] Let $\beta = 2\alpha/(2-\alpha)$ and observe that $\beta < 2$ since $\alpha < 1$. The condition $E\{|\xi_{1}|^{-\beta}\} < \infty$ in Theorem \ref{thm-error} is obviously satisfied if $|\xi_{1}|$ is bounded away from zero. Suppose that $\xi_{1}$ has a continuous density $f_{\xi}$, say, either on $[0,\infty)$ or on $(-\infty,\infty)$ in which case it is assumed to be symmetric about zero. If $\sup_{y \in [0,a)} f_{\xi}(y) < \infty$ for some $a > 0$, then it is easy to show that $E\{|\xi_{1}|^{-\beta}\} < \infty$ if $\beta < 1 \Leftrightarrow \epsilon < 2$, which is quite general in view of point (4) in Remark \ref{rem2}. If $\beta \in [1,2)$, then this expectation is finite if $\sup_{y \in [0,a)} y^{-1}f_{\xi}(y) < \infty$. 
\end{description}
\end{remark}

\subsection{{Potentially} Unidentifiable Regime}\label{misspecified_asymptotics}

\indent {As emphasized before (Section \ref{subsec2-1}), our procedure can be used whether or not the latent process falls in the identifiable regime of Definition \ref{identifiable_regime}. In this section, we carry out a theoretical analysis of the stability of our registration procedure when the distribution of the latent process deviates from the identifiable regime.} Since identifiability is lost, it is clear that consistency is no longer achievable. However, we can quantify how much the estimators deviate from their population counterparts, at least asymptotically. {Since the model may fail to be identifiable}, strictly speaking there is no unique setting corresponding to the law of the data. For this reason, as a convention, we will assume that a ``true" underlying distribution is known and fixed. For simplicity of exposition, we focus on the rank two case. This will be seen to carry the essence of the underlying effects, as we discuss in the third point of Remark \ref{rem3}. To obtain more transparent results, we focus on the case where the underlying functions are completely observable as continuous objects.

Let $X_{i} = \xi_{i1}\phi_{1} + \xi_{i2}\phi_{2}$ for $i=1,2,\ldots,n$, where $\xi_{i1}$ and $\xi_{i2}$ are uncorrelated, and $\phi_1$ and $\phi_2$ are two orthonormal elements of $L_2[0,1]$.  Let $\mu = E(X_{1}) = E(\xi_{11})\phi_{1} + E(\xi_{12})\phi_{2}$. Denote $\gamma_{l}^{2} = Var(\xi_{1l})$ and $Y_{il} = [\xi_{il} - E(\xi_{il})]/\gamma_{il}$ for $l=1,2$. Then, 
\begin{equation}\label{misspecified_model}
X_{i} = \mu + \gamma_{1}Y_{i1}\phi_{1} + \gamma_{2}Y_{i2}\phi_{2}
\end{equation}
gives the Karhunen-Lo\`eve expansion of $X_{i}$. The (random) local variation distribution induced by $X_{i}$ is $F_{i}(t) = \int_{0}^{t} |X_{i}'(u)|du/\int_{0}^{1} |X_{i}'(u)|du$ for $t \in [0,1]$. Note that contrary to the rank one case, where $\mu$ did not play a role in $F_{i}$ (due to cancellation of the term $\xi_{1}$ from the numerator and the denominator), here it cannot be neglected. We will later see that it will play a role in the performance of the estimators. Defining $\eta = \gamma_{2}/\gamma_{1}$, which is the square root of the inverse of the condition number, it follows that
\begin{eqnarray*}
F_{i}(t) = \frac{\int_{0}^{t} |\gamma_{1}^{-1}\mu'(u) + Y_{i1}\phi_{1}'(u) + {\eta}Y_{i2}\phi_{2}'(u)|du}{\int_{0}^{1} |\gamma_{1}^{-1}\mu'(u) + Y_{i1}\phi_{1}'(u) + {\eta}Y_{i2}\phi_{2}'(u)|du}.
\end{eqnarray*}
The local variation distribution induced by the observed warped data $\widetilde{X}_{i} = X_{i} \circ T_{i}^{-1}$ is given by
\begin{eqnarray*}
\widetilde{F}_{i}(t) = \frac{\int_{0}^{t} |\widetilde{X}_{i}'(u)|du}{\int_{0}^{1} |\widetilde{X}_{i}'(u)|du} = \frac{\int_{0}^{T_{i}^{-1}(t)} |\gamma_{1}^{-1}\mu'(u) + Y_{i1}\phi_{1}'(u) + {\eta}Y_{i2}\phi_{2}'(u)|du}{\int_{0}^{1} |\gamma_{1}^{-1}\mu'(u) + Y_{i1}\phi_{1}'(u) + {\eta}Y_{i2}\phi_{2}'(u)|du} = F_{i}(T_{i}^{-1}(t)).
\end{eqnarray*}

The idea is that if under suitable conditions the $F_{i}$'s manifest small variability, then the registration procedure will work quite well. We will illustrate two different situations where this is the case. The estimators of the population parameters will be the same as those considered earlier. The next theorem gives bounds on the estimation errors.
\begin{theorem} \label{thm5}
In the setting of the model provided in (\ref{misspecified_model}), define 
$$Z_i=\begin{cases}
      2\int_{0}^{1} |X_{i}'(u) - \mu'(u)|du/\int_{0}^{1} |X_{i}'(u)|du& \text{ if } \mu' \neq 0, \\
      2\eta\int_{0}^{1} |Y_{i2}\phi_{2}'(u)|du/\int_{0}^{1} |Y_{i1}\phi_{1}'(u) + {\eta}Y_{i2}\phi_{2}'(u)|du & \text{ if }\mu' = 0
\end{cases}$$
for $i=1,2,\ldots,n$. If $\mu' \neq 0$, assume that $\int_{0}^{1} |\mu'(u)|^{-\epsilon}du < \infty$ for some $\epsilon > 0$, and if $\mu' = 0$, assume that $\int_{0}^{1} |\phi_{1}'(u)|^{-\epsilon}du < \infty$ for some $\epsilon > 0$. Set $\alpha = \epsilon/(1+\epsilon)$. Suppose that assumption (I2) from Definition (\ref{identifiable_regime}) holds and that for each $i=1,2$, $\phi_{i}$ lie in $C^{1}[0,1]$ with the derivative being $\alpha_{i}$-H\"older continuous for some $\alpha_{i} \in [0,1]$. Assume that $X_{i}$ and $T_{i}$ are independent for each $i$. Also assume that $E(Z_{1}^{\alpha}) < \infty$. Then:
\begin{itemize}
\item[(a)] $\limsup_{n \rightarrow \infty} ||\widehat{T}_{i}^{-1} - T_{i}^{-1}||_{\infty} \leq const.\{E(Z_{1}^{\alpha}) + Z_{i}\}$, and $\limsup_{n \rightarrow \infty} ||\widehat{T}_{i} - T_{i}||_{\infty} \leq const. ||T_{i}'||_{\infty}\{Z_{i}^{\alpha} + E^{\alpha}(Z_{1}^{\alpha})\}$ almost surely, where the constant term is uniform in $i$. 
\item[(b)] $\limsup_{n \rightarrow \infty} ||\widehat{X}_{i} - X_{i}||_{\infty} \leq O_{P}(1)\{E(Z_{1}^{\alpha}) + Z_{i}\}$  almost surely. 
\end{itemize}
\end{theorem} 

\begin{remark} \label{rem3}
\normalfont
\begin{description}
\item[\normalfont{\emph{Model misspecification}}] Theorem \ref{thm5} reveals that if the $Z_{i}$ are small, the effect of misspecification is also small. Here are two such cases: \\
\indent (a) \indent When $\mu' \neq 0$, $Z_{i} = \int_{0}^{t} |Y_{i1}\phi_{1}'(u) + {\eta}Y_{i2}\phi_{2}'(u)|du/\int_{0}^{1} |\gamma_{1}^{-1}\mu'(u) + Y_{i1}\phi_{1}'(u) + {\eta}Y_{i2}\phi_{2}'(u)|du$. So, in this case, if $|\gamma_{1}^{-1}\mu'|$ has a large enough contribution compared to $|Y_{i1}\phi_{1}' + {\eta}Y_{i2}\phi_{2}'|$ for all $i$, then the $Z_{i}$'s are small. 

\indent (b) \indent On the other hand, if $\mu' = 0$, then if $\eta$ is small, i.e., the condition number of the process is large (which essentially implies that the process is ``close" to a rank one process provided $E(\xi_{12}) = 0$), then the $Z_{i}$'s are small. This can be compared to the minimum eigenvalue registration principle of \cite{RS05}, where one tries to find the warp function that minimises the second eigenvalue of the cross-product matrix between the target function and the registered function. Assume that $E(\xi_{i1}) = E(\xi_{i2}) = 0$ and without loss of generality that $\gamma_{1} = 1$. If in reality the true unobserved curves are rank one, i.e., the $\xi_{i1}\phi_{1}$ component, and we observe warped versions of the rank two curves $X_{i}$'s, then (in the population case) correct registration is achieved by $T_{i}$ if the minimum eigenvalue, namely $\gamma_{2}^{2} = \eta^{2}$, of the expected cross-product matrix equals zero. Thus, in the empirical case, if $\eta$ is close to zero, we may expect $\widehat{T}_{i}$ to be close to $T_{i}$ and consequently expect the registration procedure to have good performance. 

\item[\normalfont{\emph{Convergence of other model parameters}}] Bounds similar to those in (a) and (b) of Theorem \ref{thm5} can also be obtained for the mean, the covariance, the $\gamma_{l}$'s and the $\phi_{l}$'s as well as the principal components $Y_{il}$'s. We do not include them in the statement of the theorem because they need more complicated conditions involving the parameters. 

\item[\normalfont{\emph{General (possibly infinite) rank situation}}] Let $X_{i} = \mu + \sum_{j=1}^{M} \gamma_{j}Y_{ij}\phi_{j}$ for some $1 \leq M \leq \infty$, where the $\{Y_{ij} : j =1,2,\ldots,M\}$ are uncorrelated with zero mean and unit variance. Without loss of generality, we assume that $\gamma_{1} > \gamma_{2} > \ldots \geq 0$. The errors in estimation when $\mu' \neq 0$ remain the same as in Theorem \ref{thm5}. When $\mu' = 0$, then we define $Z_{i} = 2\eta\int_{0}^{1} |Y_{i2}\phi_{2}'(u) + \sum_{k \geq 3} \delta_{k}Y_{ik}\phi_{k}'(u)|du/\int_{0}^{1} |Y_{i1}\phi_{1}'(u) + {\eta}[Y_{i2}\phi_{2}'(u) + \sum_{k \geq 3} \delta_{k}Y_{ik}\phi_{k}'(u)]|du$ for $i=1,2,\ldots,n$, where $\delta_{k} = \gamma_{k}/\gamma_{2}$ for $k \geq 3$. In this case, under the conditions of Theorem \ref{thm5}, the bounds as in that theorem still hold true. Note that $\delta_{k} \leq 1$ for all $k \geq 3$. So, in the general case, the performance of the registration procedure studied in the paper will only depend on how small $\eta$ is and does not in general depend on the values of the $\delta_{k}$'s (or the $\gamma_{j}$'s for $j \geq 3$). In other words, only the behaviour of the second frequency component relative to the first one matters (which elucidates the role of $\delta$ in the standard model, i.e. Equation  \ref{standard_model}, whose role is precisely to tune this behaviour). Of course, the magnitude of the error in estimation for the same value of $\eta$ will now differ from the rank $2$ case because of the presence of the additional terms. We have investigated these issues in a simulation study in Section \ref{misspecified_simulations} (see, in particular, Figure \ref{Fig6}).

\item[\normalfont{\emph{Comparison with pairwise warping}}] In the setup of the infinite rank latent model considered in item (3) above, we now compare the bounds obtained in Theorem \ref{thm5} to those obtained by \cite{TM08}. Denoting $\sum_{j=1}^{M} \gamma_{j}Y_{ij}\phi_{j} = {\kappa}W_{i}$, it follows that the latent model is exactly the same as considered in that paper (see p. 877 with $\delta$ there replaced by $\kappa$). So, if $\mu' \neq 0$, it follows that $Z_{i} = 2\kappa\int_{0}^{1} |W_{i}'(u)|du/\int_{0}^{1} |\mu'(u) + {\kappa}W_{i}'(u)|du = O_{P}(\kappa)$, which is similar to the bound obtained in \cite{TM08}. Our analysis nevertheless refines the results of \cite{TM08} in the sense that it reveals the impact of $\mu$ on the asymptotic bias -- larger magnitudes of $\mu'$ yield smaller asymptotic bias.  Further refinements can be offered by differentiating between the cases  $\mu' \neq 0$ and $\mu' = 0$. Specifically, when $\mu' = 0$, it can be shown that $Z_{i} = 2\int_{0}^{1} |W_{i}'(u) - Y_{i1}\phi_{1}'(u)|du/\int_{0}^{1} |W_{i}'(u)|du$. Thus, in this case, the error bounds on the warp maps in Theorem \ref{thm5} do not depend on $\kappa$. This is to be expected for the following reason. Note that $\mu' = 0$ means that the latent process in this case is $X(t) = c + {\kappa}W(t)$ for a constant $c$, and hence, the warped process is $\widetilde{X}(t) = c + {\kappa}W(T^{-1}(t))$. Thus, the warped version of the process $X$ differs from the warped version of the process $W$ only by a constant shift and a scale factor. Ideally, any proper registration procedure should be invariant with respect these transformations since they do not affect the time scale. This is clearly true for our procedure. We should thus get the same estimates of the warp maps if we work with the warped process $W(T^{-1}(t))$ (which does not involve $\kappa$) instead of $\widetilde{X}$. 
\end{description}
\end{remark}

\section{Numerical Experiments} \label{sec3}

We now carry out simulation experiments to probe the finite-sample performance of our registration procedure. First we treat the case of a well-specified identifiable regime without error, and then separately the case when there are measurement errors in the observations. Finally, we consider the setup when the rank of the latent process is more than one (departure from identifiability). In all cases, we have compared the performance of the proposed registration method to the continuous monotone registration (CMR) method by \cite{RL98}, the pairwise registration (PW) technique of \cite{TM08} and registration using the Fisher-Rao metric (FMR) studied in \cite{SWKKM11}. The CMR procedure is implemented using the ``register.fd'' function in the \texttt{R} package \texttt{fda}. The PW procedure is implemented using the \texttt{Matlab} codes in the \texttt{PACE} package. The FMR method is implemented using the ``time\_warping'' function in the \texttt{R} package \texttt{fdasrvf}. The tuning parameters in the PW method are always chosen to be the default ones since the other choices were found to be computationally extremely intensive. For the CMR procedure, we compared its performance by using different numbers of B-spline basis functions in the structure of the warp maps (see \cite{RL98}). This varies their complexity. However, we found that the best performance was obtained when the warp maps are simple. As will be seen in the  simulations, the registration procedures involving structural assumptions on warp maps and consequently more tuning parameters (CMR and PW) encounter difficulties in several of the models considered, which is probably due to the mis-specification of the true warping mechanism.

\subsection{Identifiable Regime Without Measurement Error} \label{well-specified}

Let $X(t) = \xi\phi(t), t \in [0,1]$, and consider two models: 
 \begin{itemize}
\item[Model 1:]  $\xi \sim N(1.5,1)$, $\phi(t) = \exp\{cos(2{\pi}t-\pi)\}$; 

\item[Model 2:] $\xi \sim 1 + Beta(2,2)$, $\phi(t) = \{1 - (t-0.25)^{2}\}\cos(3{\pi}t)$. 
\end{itemize}
In either case, the sample size is $n = 50$ and the curves are observed at $r = 101$ equally spaced points in $[0,1]$. The warp maps are chosen according to the last point of Remark \ref{rem1} with the parameters $J = 2$, $K = V_{1}V_{2}$, where $V_{1} \sim Poisson(3)$, $P(V_{2} = \pm 1) = 1/2$ with $V_{2}$ independent of $V_{1}$, and $\beta = 1.01$. 

The kernel for the Nadaraya-Watson estimator as well as the one used to smooth the $\widehat{T}_{i,d}$'s is the Epanechnikov kernel on $[-1,1]$. For both the models, the bandwidths used in the registration procedure were chosen to under-smooth the data so that the features (maxima, minima, etc.) are not smeared out. In order to provide smooth registered curves, we have smoothed the $\widehat{T}_{i,d}$'s using cubic splines with $11$ equi-spaced knots on $[0,1]$, prior to synchronising the data. \\
\indent Figure \ref{Fig2} shows the plots of the true, warped and registered data curves; the true, warped and registered means; and the true, warped and registered leading eigenfunctions under Model 1 and Model 2. Figure \ref{Fig2} suggests that the procedure studied in this paper has been able to adequately register the discretely observed and warped sample curves. Moreover, it is clear that the cross-sectional mean and the leading eigenfunction of the warped curves differ from the true mean and leading eigenfunction in either amplitude or phase (under either model), while the registration procedure corrects the problem, and the resulting estimates (whether smoothed or raw) are very close to the true functions. \\
\begin{figure}[ht!]
\vspace{-0.2in}
\begin{center}
{
\includegraphics[scale=0.46]{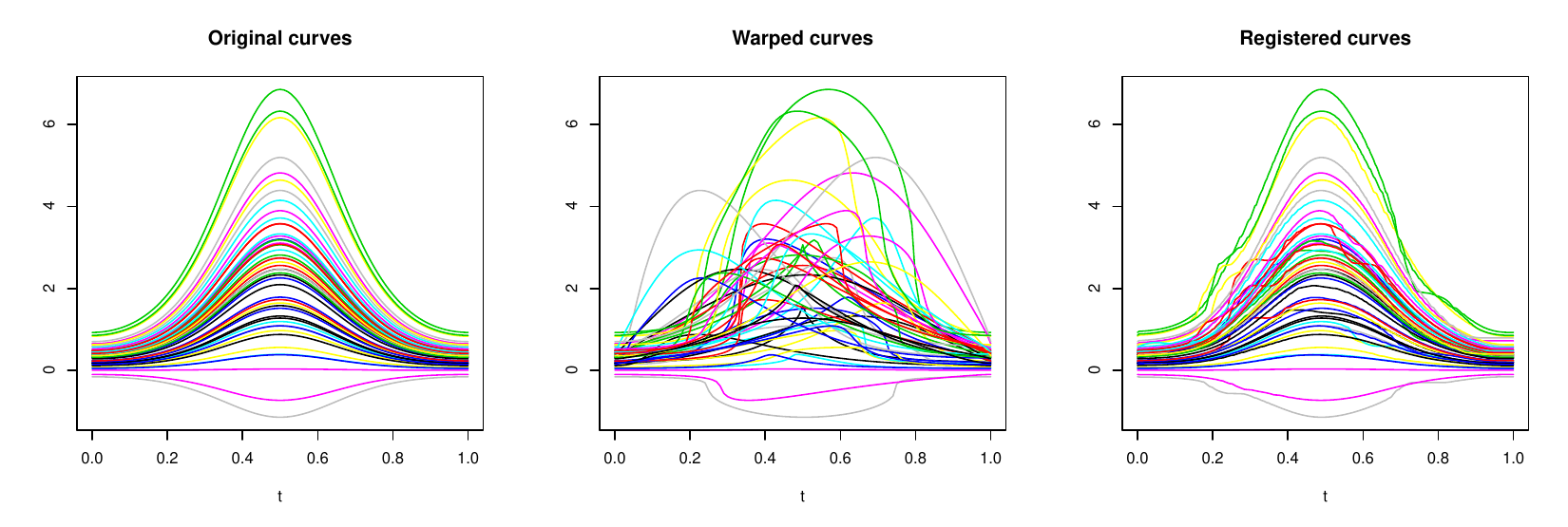}
\includegraphics[scale=0.3]{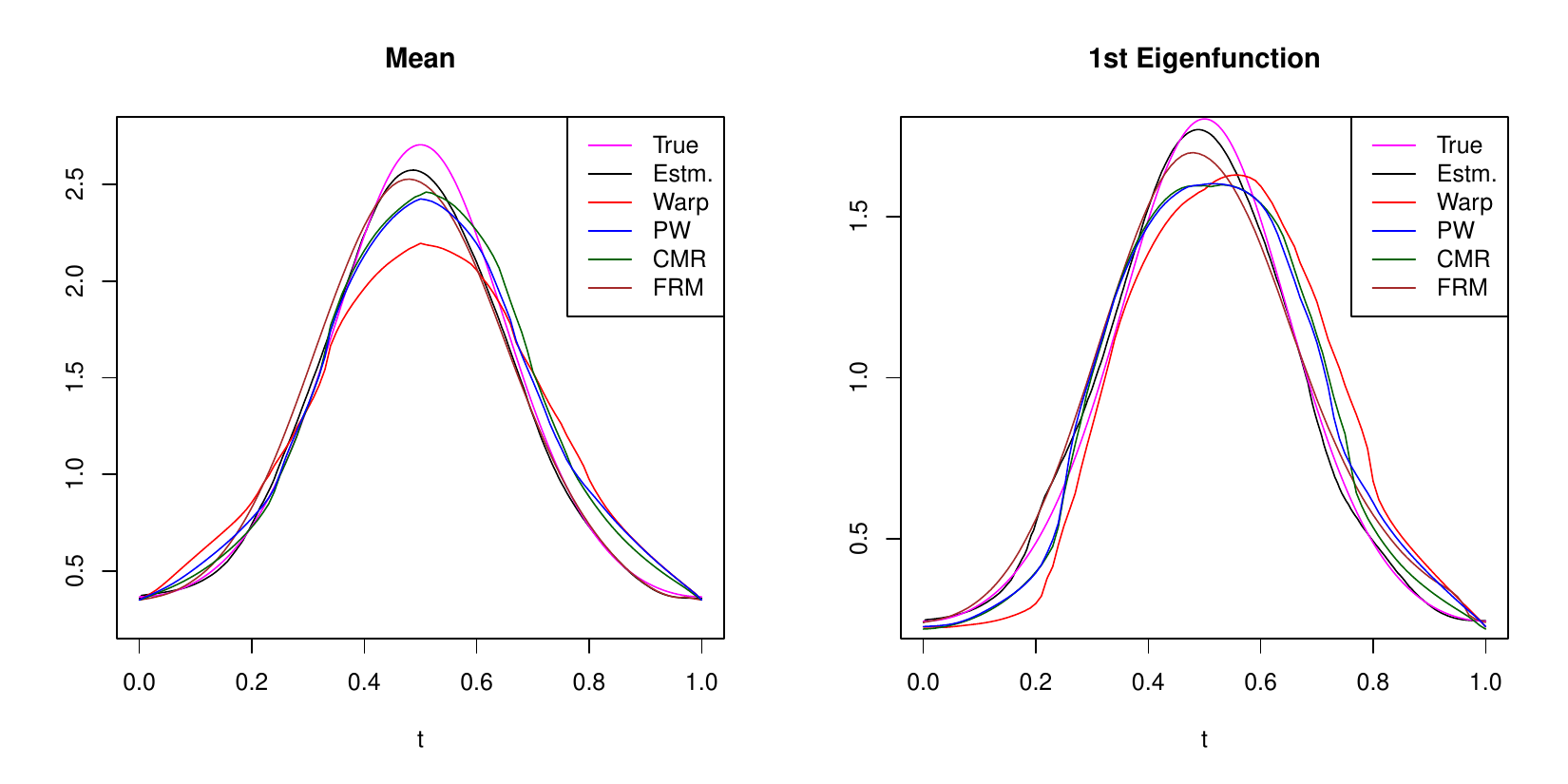}
\includegraphics[scale=0.46]{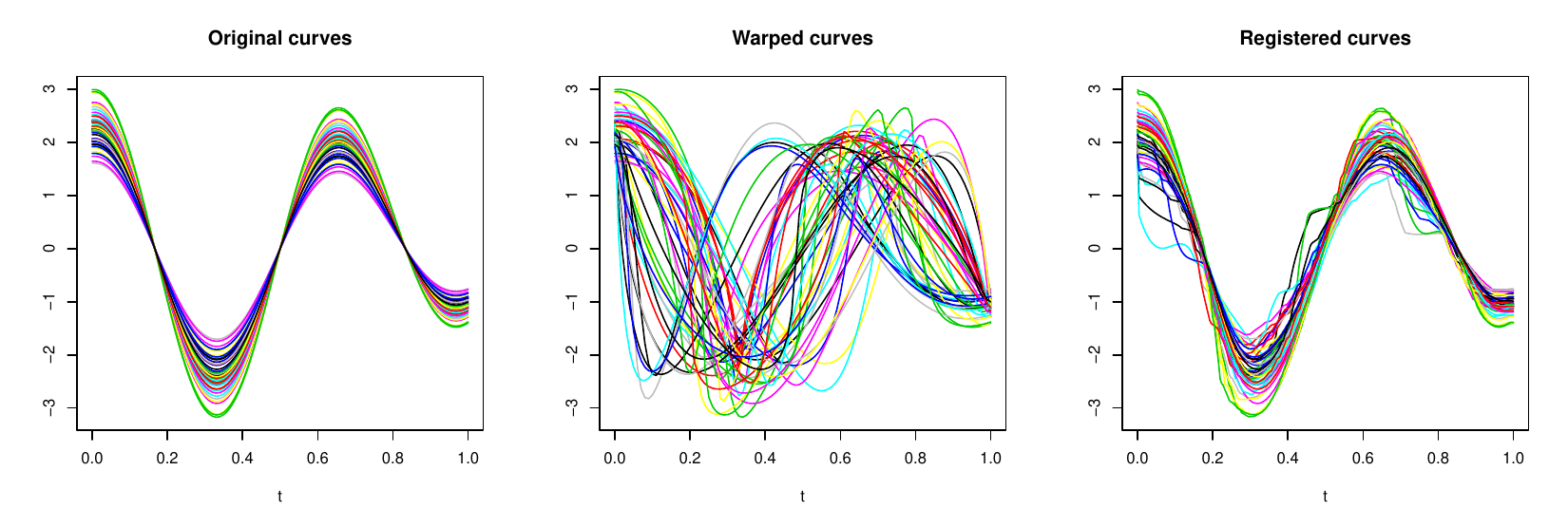}
\includegraphics[scale=0.3]{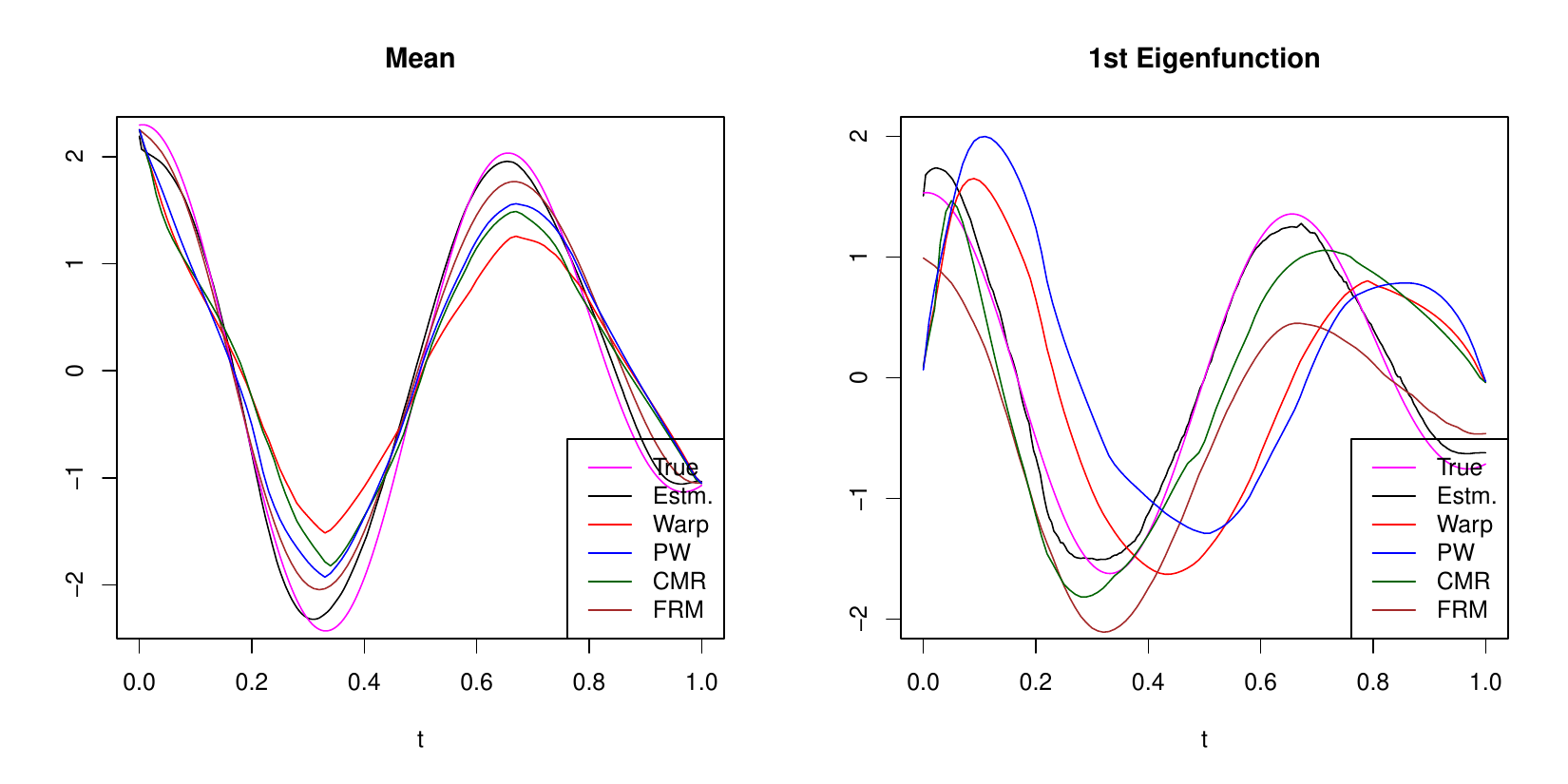}
}
\end{center}
\vspace{-0.3in}
\caption{Plots of the true, warped and registered data curves (using our procedure) along with the estimated mean leading eigenfunction under Model 1 (top two rows) and Model 2 (bottom two rows) without measurement error obtained using our procedure as well as some other methods.}
\label{Fig2}
\end{figure}
\indent Under both the models, it is seen that the estimates of the mean and the leading eigenfunction obtained using the proposed registration procedure is closest to the true functions compared to all the other methods considered. This is more prominent under Model 2 (see the bottom two rows in Fig. \ref{Fig2}), where the estimates of the leading eigenfunction obtained by all of other competing procedures considered are far from the true eigenfunction. Also, the registered functions obtained using the CMR and the PW methods do not resemble the true functions (see Figures \ref{Figsupp-model1} and \ref{Figsupp-model2} in Appendix A (Supplementary Material)). The above facts show that for small sample sizes, even under no measurement error, some of the well-known registration procedures may yield unsatisfactory results, while the proposed procedure works well in these cases. 

\subsection{Identifiable Regime With Measurement Error} \label{errors}

\indent We now consider the situation when the warped observations under an identifable rank one model have been observed with measurement errors. As observed in our theoretical study in Section \ref{specified_asymptotics}, the rate of convergence will be much slower than the case when there is no measurement error. For our simulations, we thus keep the same two models as in Section \ref{well-specified} but increase the sample size to $n = 250$. The measurement errors under Model $1$ are i.i.d. $\mathrm{Unif}(-0.2,0.2)$ while those under Model $2$ are i.i.d. $\mathrm{Unif}(-0.4,0.4)$. The bandwidths for the smoothing steps involved in the registration procedure are chosen using built-in cross-validation bandwidth choice function ``regCVBwSelC'' in the \texttt{locpol} package in the \texttt{R} software. Figures \ref{Fig7} and \ref{Fig8} show the plots of the unobserved true rank one curves, the warped curves that are observed with error and the registered curves. They also contain the plots of the mean function and the leading eigenfunction of the true, warped and registered data under the two models. \\
\begin{figure}[t]
\vspace{-0.2in}
\begin{center}
{
\includegraphics[scale=0.4]{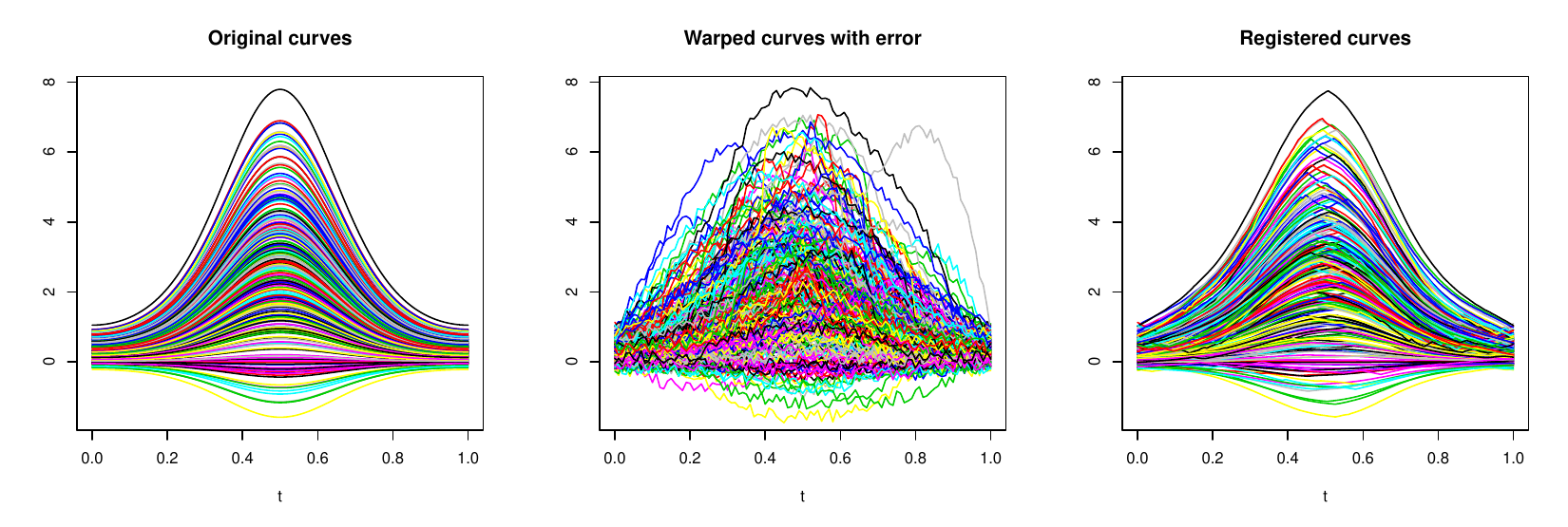}
\includegraphics[scale=0.3]{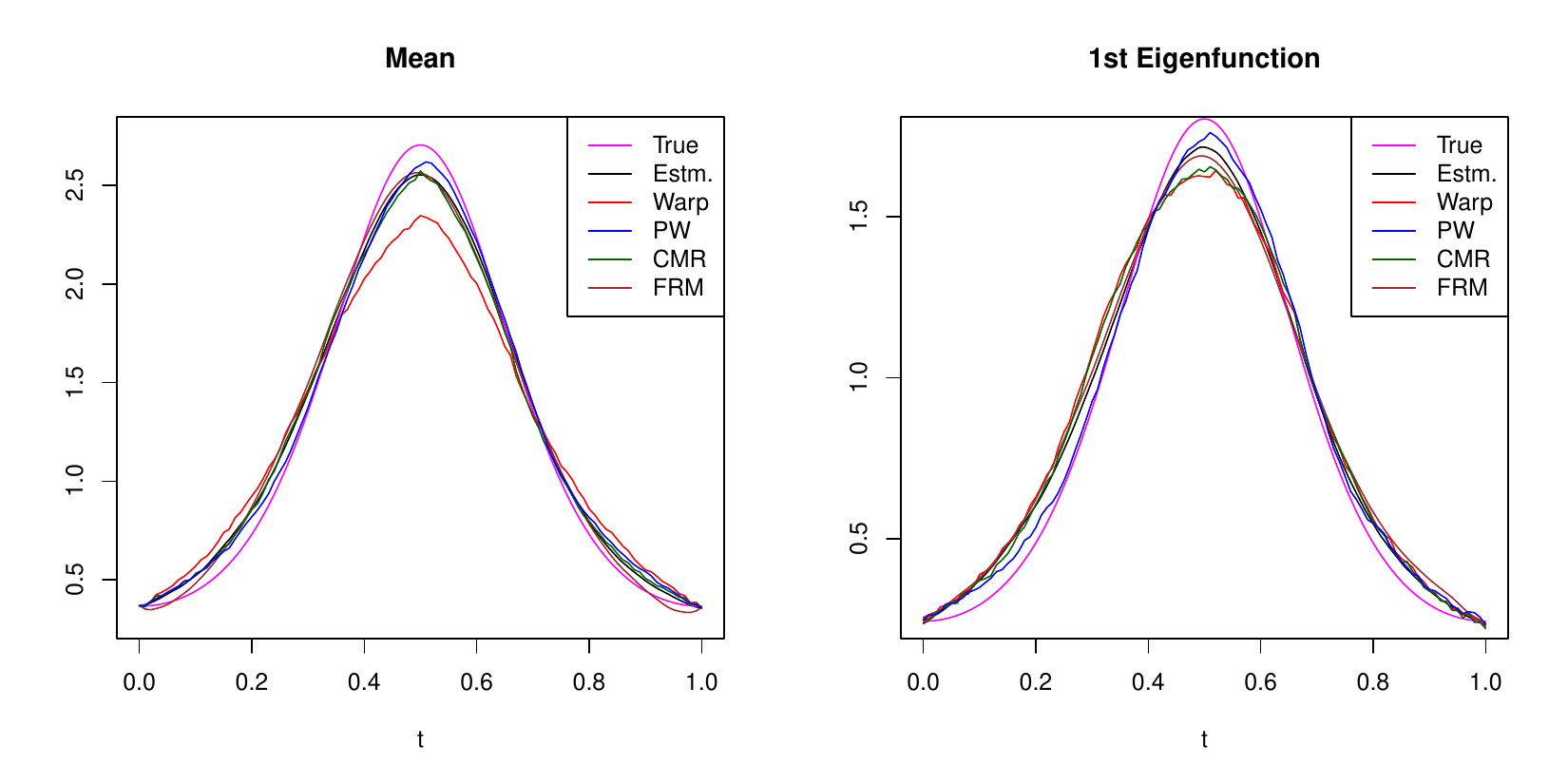}
}
\end{center}
\vspace{-0.3in}
\caption{Plots of the true, warped and registered data curves (using our procedure) along with the estimated mean leading eigenfunction under Model 1 with measurement error obtained using our procedure as well as some other methods.}
\label{Fig7}
\end{figure}
It is observed that even subject to measurement error contamination, the proposed registration procedure is able to adequately register the curves. In particular, under Model 2, the means as well as the leading eigenfunction of the true and the registered curves are quite close. We also performed the registration procedure with a Nadaraya-Watson estimator (without boundary kernels) for obtaining an estimate of the $\widetilde{X}_{i}$'s (see Step 3**). The performance was not that different from the one using a local linear estimator.

Only the FRM procedure fares similarly as the proposed one when estimating the leading eigenfunction under both models. However, the PW method yields quite similar estimates of the mean as the proposed method and the FRM method under each of the two models. Both the CMR and the PW methods fail to produce adequately registered curves as is seen from Figures \ref{Figsupp-model1err} and \ref{Figsupp-model2err} in Appendix A (Supplementary Material). The improvement in the performance of the FRM technique under Model 2 with error compared to the case without error considered in the previous subsection (see the plots in the last row in Fig. \ref{Fig2}) is perhaps due to the increased sample size, which compensates for the measurement error. 
\begin{figure}[t]
\vspace{-0.2in}
\begin{center}
{
\includegraphics[scale=0.4]{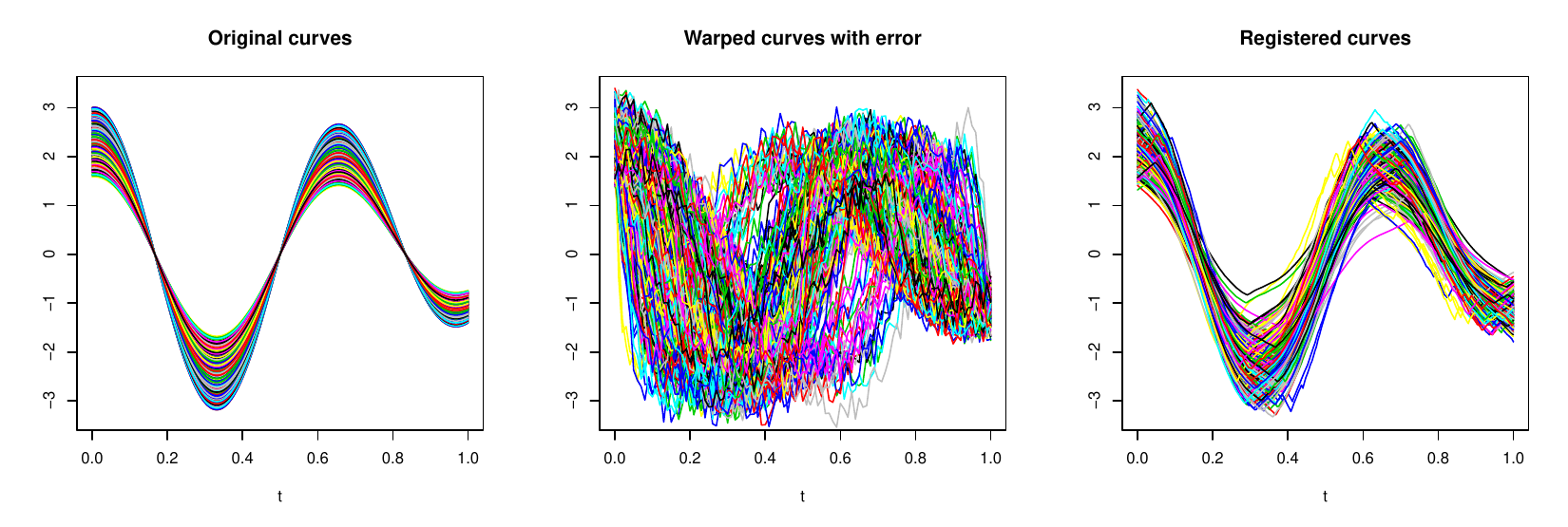}
\includegraphics[scale=0.3]{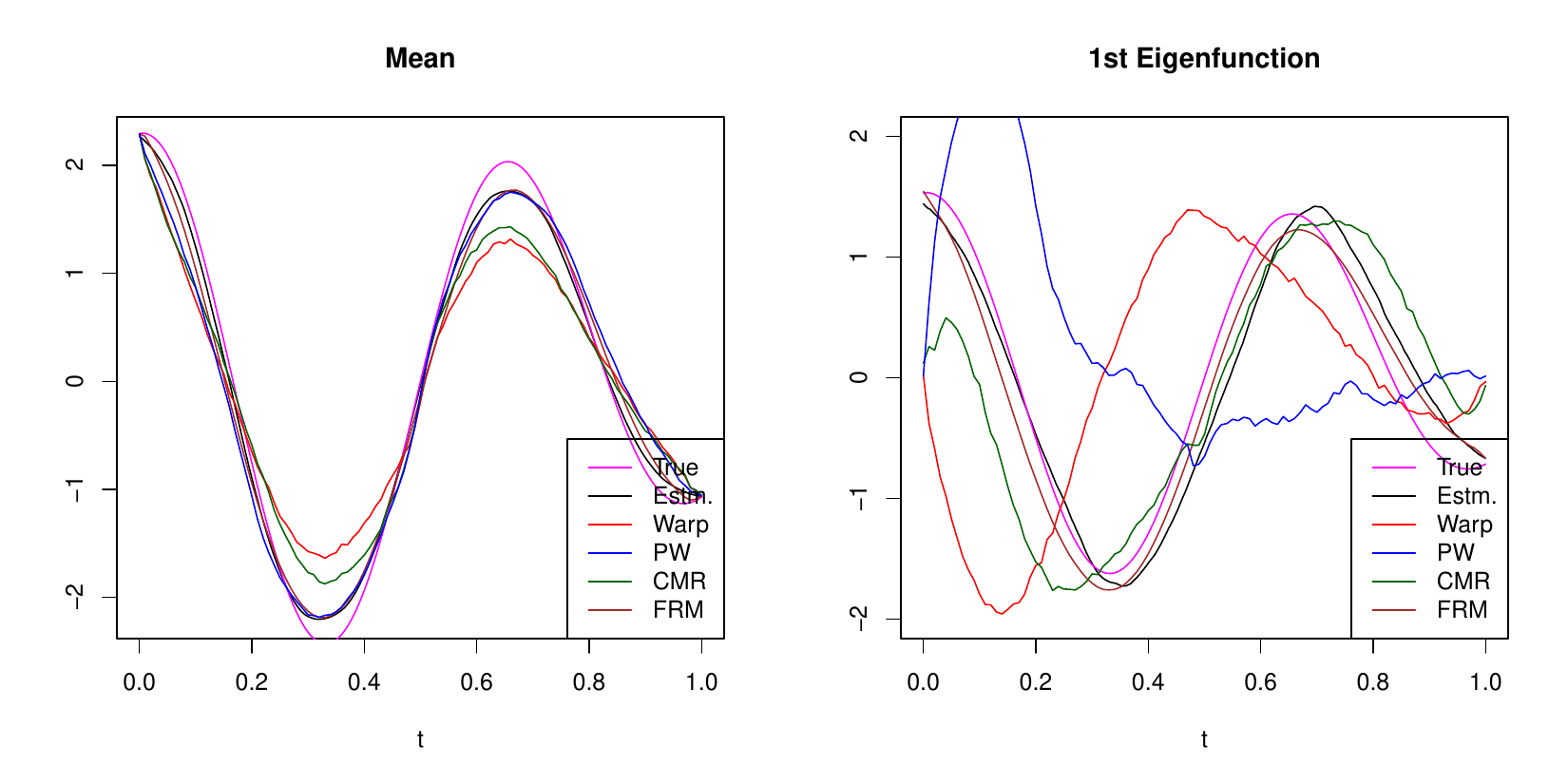}
}
\end{center}
\vspace{-0.3in}
\caption{Plots of the true, warped and registered data curves (using our procedure) along with the estimated mean leading eigenfunction under Model 2 with measurement error obtained using our procedure as well as some other methods.}
\label{Fig8}
\end{figure}

\subsection{{Potentially} Unidentifiable Regime}\label{misspecified_simulations}

\begin{figure}
\vspace{-0.2in}
\begin{center}
{
\includegraphics[scale=0.42]{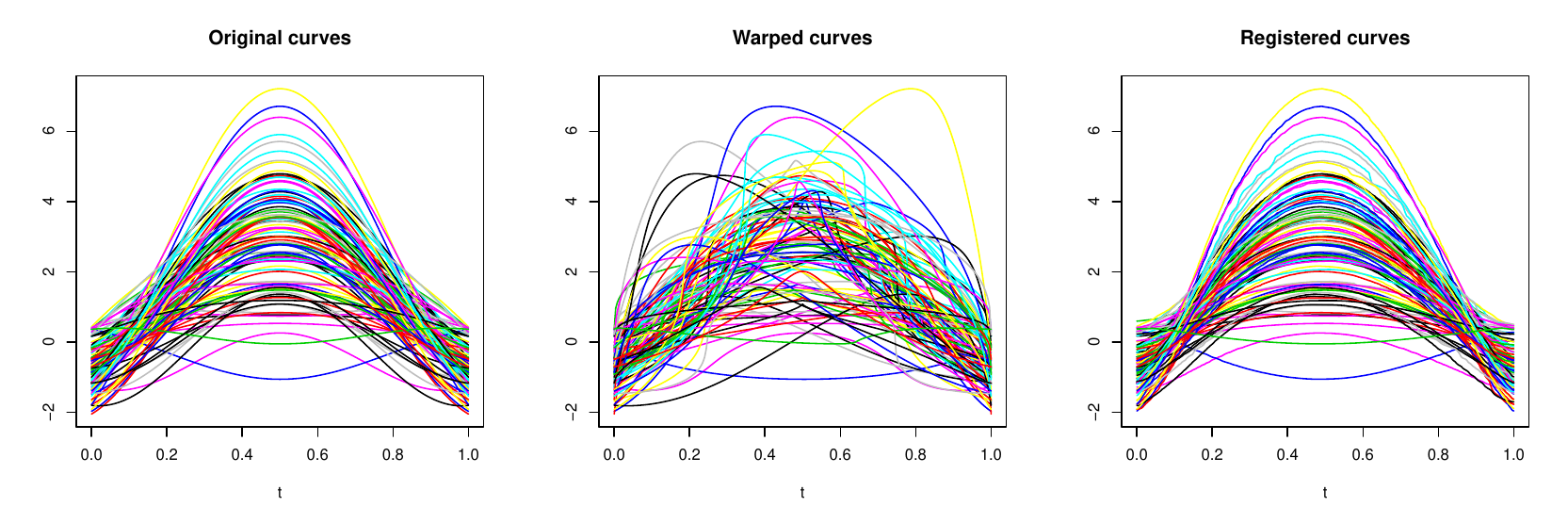}
\includegraphics[scale=0.42]{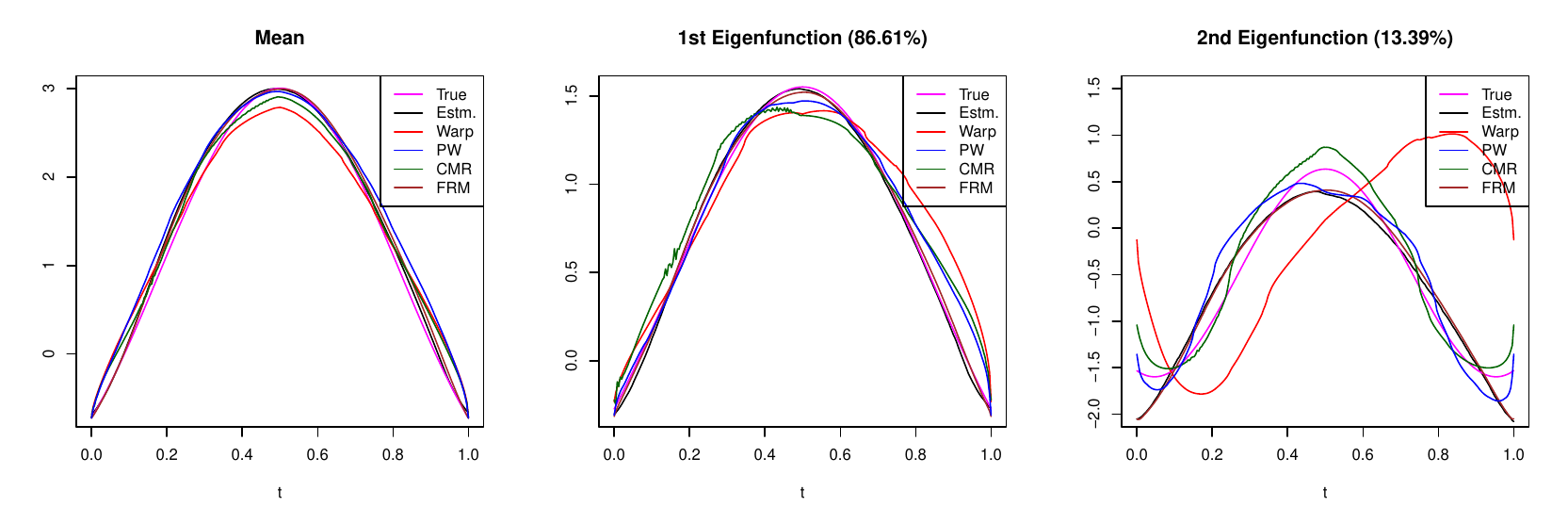}
\includegraphics[scale=0.42]{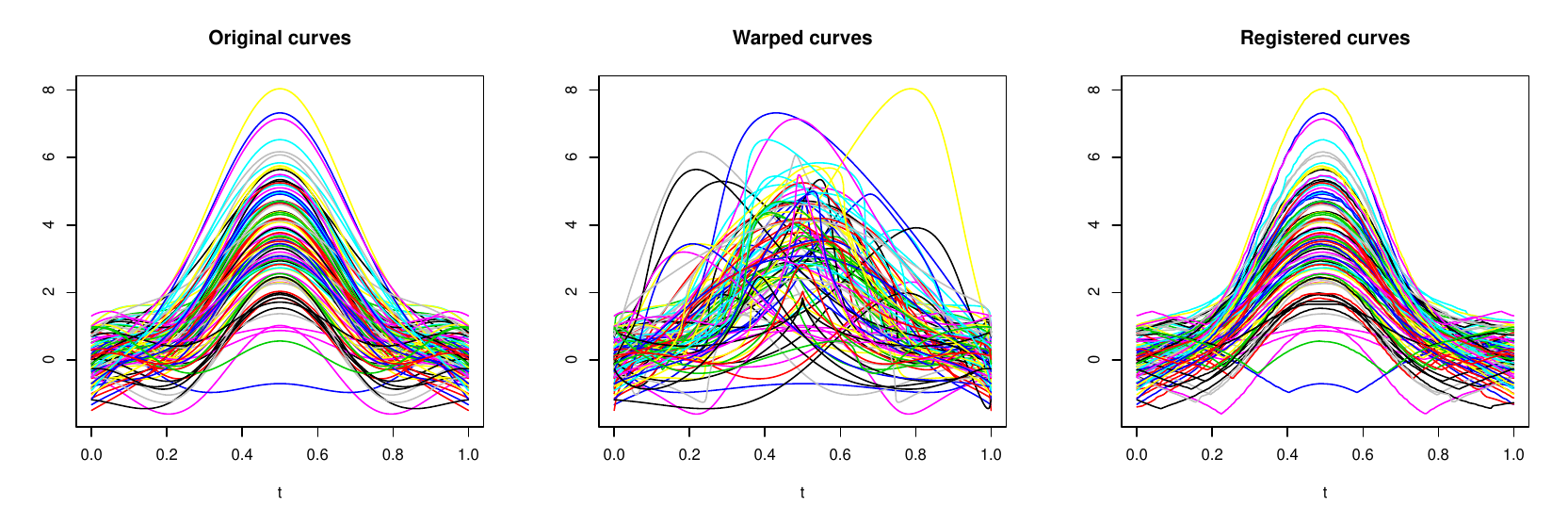}
\includegraphics[scale=0.3]{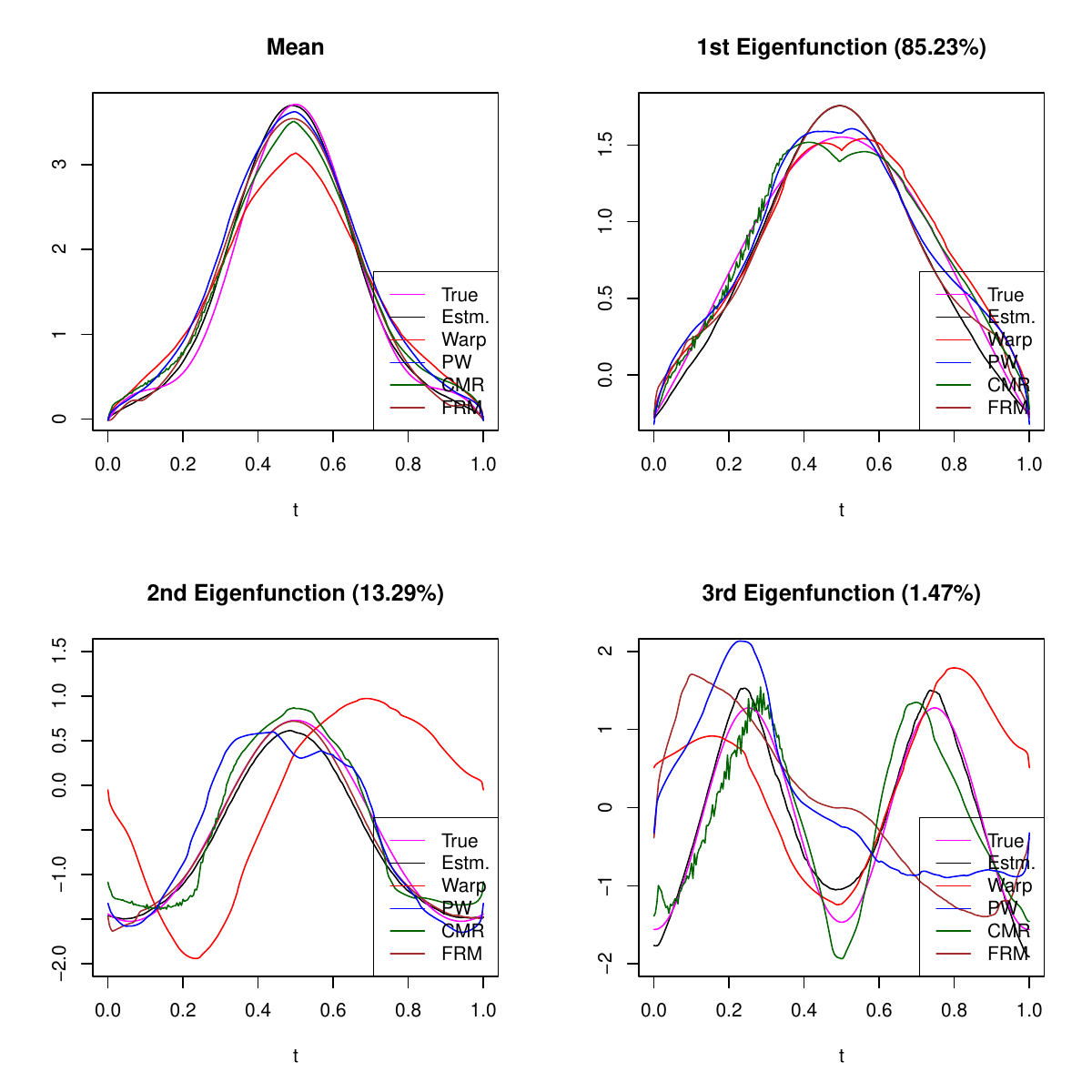}
}
\end{center}
\vspace{-0.2in}
\caption{Plots of the true, warped and registered data curves along with the means and eigenfunctions of the true, warped and the registered data using our method and some other procedures under the rank $2$ (top two rows) and the rank $3$ models (bottom three rows).}
\label{Fig5}
\end{figure}
\indent We next carry out experiments to probe the performance of the registration procedure in a rank $2$ and a rank $3$ setting -- these correspond to a {potentially} unidentifiable regime. The model considered in the rank $2$ case are $X = \xi_{1}\phi_{1} + \xi_{2}\phi_{2}$ with $\xi_{1} \sim N(1.5,1)$, $\xi_{2} \sim N(-0.5,0.15)$, $\phi_{1}(t) = \sqrt{2}\sin({\pi}t)$ and $\phi_{2}(t) = \sqrt{2}\cos(2{\pi}t), t \in [0,1]$. In the rank $3$ case, we consider $X = \xi_{1}\phi_{1} + \xi_{2}\phi_{2} + \xi_{3}\phi_{3}$ with the same choices of $\xi_{j}$ and $\phi_{j}$ as above for $j=1,2$ along with $\xi_{3} \sim N(0.5,(0.15)^2)$ and $\phi_{3}(t) = \sqrt{2}\cos(4{\pi}t)$. The warp maps are the same as those considered in the simulation study in Section \ref{sec3}. The plots of the true curves, the warped curves and the registered curves are provided in Figure \ref{Fig5} for the rank $2$ and the rank $3$ models. 
 The {potentially} unidentifiable setting has to be interpreted as follows: in light of  Theorem \ref{thm1} and the ensuing counter-examples, there may be other models that could have generated the (statistically) same data. Consequently, strictly speaking, we cannot really talk about good or bad performance, as we there may be several equally valid ``ground truths" to compare to. But the way we have constructed the {potentially} unidentifiable simulation setting is by means of a mild departure from an identifiable model. Therefore, we can arbitrarily consider that the latter identifiable model is the truth and investigate whether the registration procedure is stable to the said mild departure. A more detailed investigation of stability is pursued later in this subsection.

It is observed that the registration procedure performs quite well and aligns the peak (present in the true curves) adequately under both models (see Figure \ref{Fig5}). Further, the two smaller troughs near the end-points present in the rank $3$ model are also reasonably aligned (see the plots in the third row in Figure \ref{Fig5}). However, except the FRM procedure, the other two competing methods completely fail in registering the data curves (see Figures \ref{Figsupp-rank2} and \ref{Figsupp-rank3} in Appendix A (Supplementary Material)). Also, unlike our procedure, the registered curves using the FRM procedure seems to lack the two troughs present in the original curves near the boundary points for the rank $3$ model. For each of the two models, the mean seems to be estimated very well based on the registered curves using our procedure. The other procedures follow suit. A similar statement is also true for the first eigenfunction under these two models. However, there is more bias in the estimate of the second eigenfunction under the rank $2$ model for all of the registration procedures. Under the rank $3$ model, the CMR and the PW methods are not fully able to capture the shape of the second eigenfunction, while our procedure and the FRM method does. The third eigenfunction under this model is somewhat reasonably estimated only by our procedure.  \\
\indent In order to probe the breakdown point of the proposed registration procedure in the rank $> 1$ setting, we also considered classes of rank $2$ and rank $3$ models, recorded the relative $L_{2}$-error in estimation of the data curves, i.e, the median of $||\widehat{X}_{i} - X_{i}||/||X_{i}||, i=1,2,\ldots,n$, and consider a threshold of $10\%$ error as a criterion for good performance. The models are generated similar to the earlier simulation. For the rank $2$ case, let $X = \xi_{1,c}\phi_{1} + \xi_{2,c,r}\phi_{2}$, where $\xi_{1} \sim N(3c,1)$, $\xi_{2} \sim N(-c,r)$, where $c \in [0.1,2]$ and $r \in [0.01,0.3]$. The choices of $c$ and $r$ ensure that we include both approximately rank $1$ models ($c$ and $r$ close to zero) as well as proper rank $2$ models (large values of $r$). Similarly, for the rank $3$ case, let $X = \xi_{1,c}\phi_{1} + \xi_{2,c,r}\phi_{2} + \xi_{3,c,r}\phi_{3}$, where $\xi_{3} \sim N(c,r^{2})$. Figure \ref{Fig6} shows a plot of the relative $L_{2}$-errors under these classes of models, for various combinations of the parameters $c$ and $r$. It is seen that when $c$ is large, the performance of the registration procedure is good, which conforms with our theoretical arguments in Theorem \ref{thm5}. In fact, for this class of rank $2$ models, the maximum $L_{2}$ error does not exceed $12.9\%$. On the other hand, when $c$ is small, the allowable range of $r$ values for good performance is much greater in the rank $2$ setup compared to the rank $3$ setup (cf. (c) in Remark \ref{rem3}). In fact, in the rank $3$ setup, the error is more than $10\%$ for all $r$ in the range considered when $c \leq 0.2$. Further, the maximum $L_{2}$ error is now $29.8\%$. 
\begin{figure}[t]
\vspace{-0.1in}
\begin{center}
{
\includegraphics[scale=0.19]{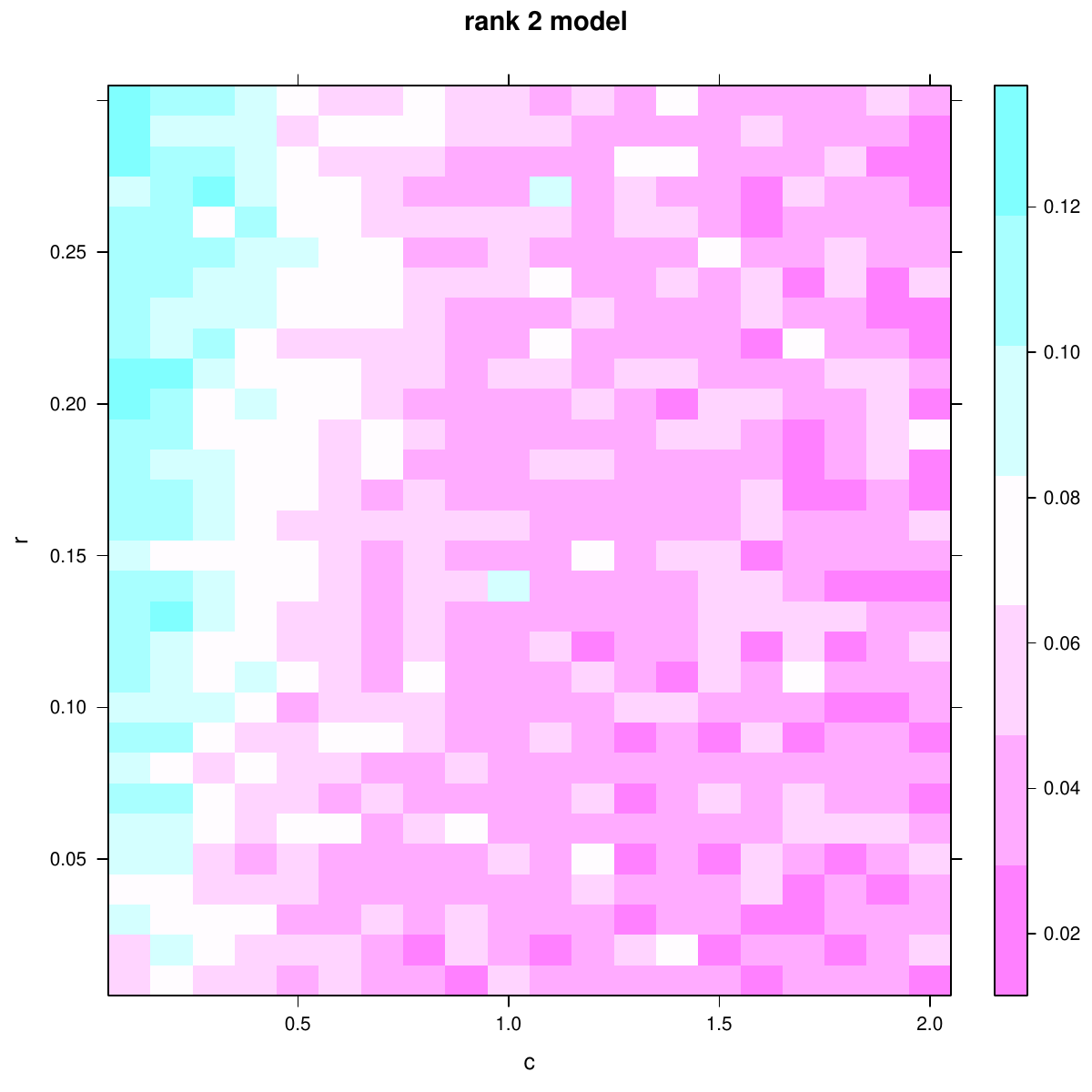}
\qquad
\includegraphics[scale=0.19]{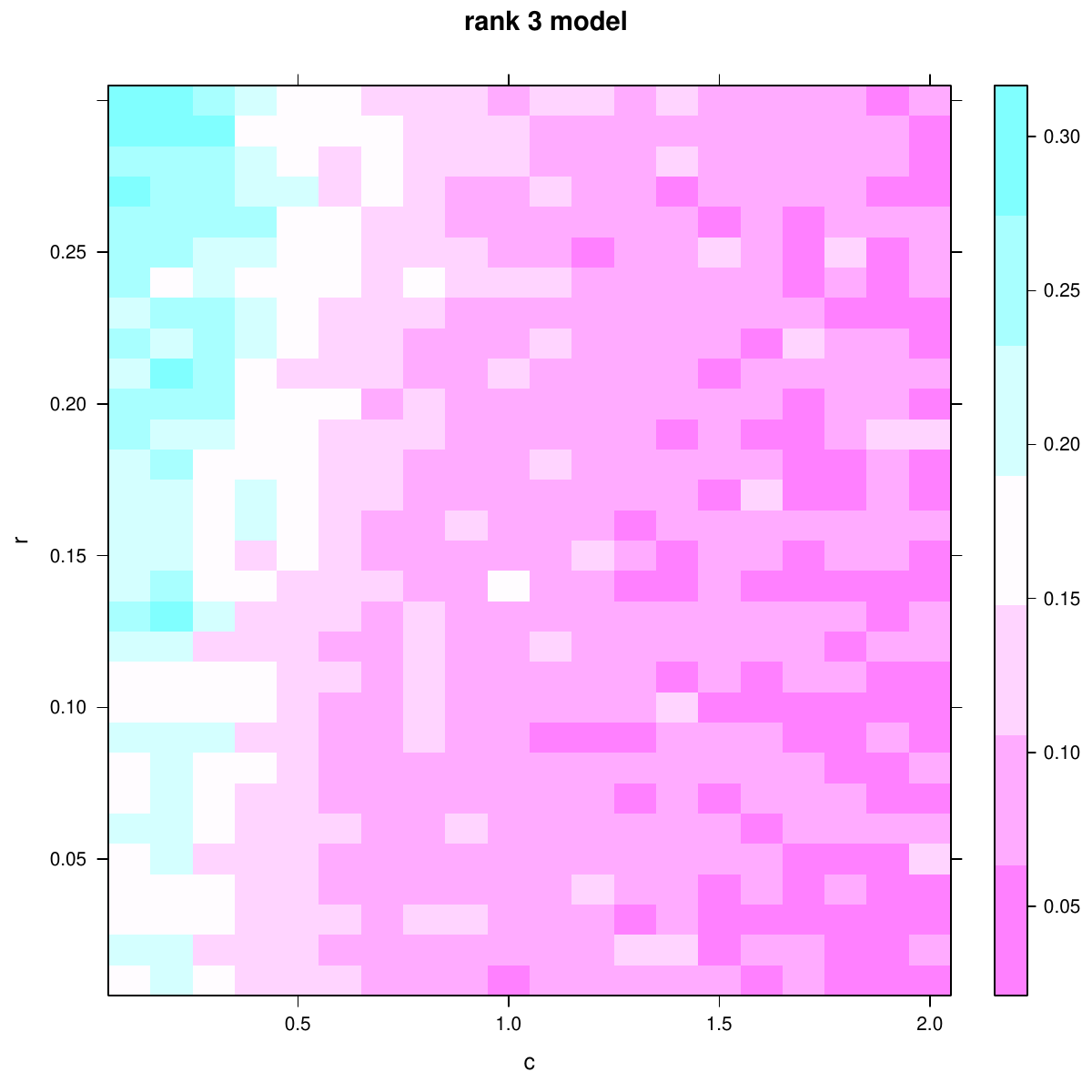}
}
\end{center}
\vspace{-0.2in}
\caption{Level-plots of the relative $L_{2}$ errors under the rank $2$ and the rank $3$ classes of models.}
\label{Fig6}
\end{figure}

\section{Data Analysis} \label{sec4}
\indent In this section, we illustrate the performance of our registration procedure on a data set of growth curves of \textit{Tribolium} beetle larvae, collected and analysed by \citet{irwin2013constraints}. Each curve represents the mass measurement (in milligrams) as a function of the age of the larvae since hatching (in days). Their analysis of \textit{Tribolium} growth suggests that these beetles' growth patterns differ from those of other animals with determinate growth (that is, growth that is contained in certain life stages). Usually, the longer the growth period, the larger the maximal mass attained (see \cite{irwin2014artificial}, and references therein). In \textit{Tribolium}, however, it seems that beetles that tend to grow faster, and thus have a shorter growth period, also tend to attain larger size (e.g. Figure \ref{Fig3}, top left). See \citet{irwin2013constraints} for more details and background. This observation suggests that the \textit{Tribolium} data could be well-suited for a phase-amplitude analysis under a latent rank 1 model that has been warped: one expects that correcting for different ``growth clocks" (phase variation) should yield curves that are roughly of unimodal amplitude variation, due to final mass. Conversely, it suggests a potential latent model that produces rank 1 vertical variation related only to final mass, and horizontal variation due to growth timing (i.e. how this total final mass is accumulated in time). 

For our analysis, we have only considered the part of the dataset where there were at least $10$ discrete measurements per individual curve, which results in a sample size of $159$. Also, not all larvae were recorded on the same day so that the number of observations differed across individuals. Since there are relatively few measurements (maximum $12$) per individual larvae, we smoothed each observation vector as a pre-processing step. This was done using the built-in function \texttt{splinefun} in the \texttt{R} software with the method \texttt{monoH.FC} that uses monotone Hermite spline interpolation proposed by \cite{FC80} (since the curves are expected to be approximately increasing). 

As is typically the case with growth curves, one expects that, if unaccounted for, the lurking phase variation would give the impression of several modes of amplitude variation. The aim our analysis is thus to register the curves, estimate the warp maps, estimate the mean of the registered curves, and carry out an eigenanalysis of the registered data. \\
\begin{figure}[t]
\vspace{-0.2in}
\begin{center}
{
\includegraphics[scale=0.46]{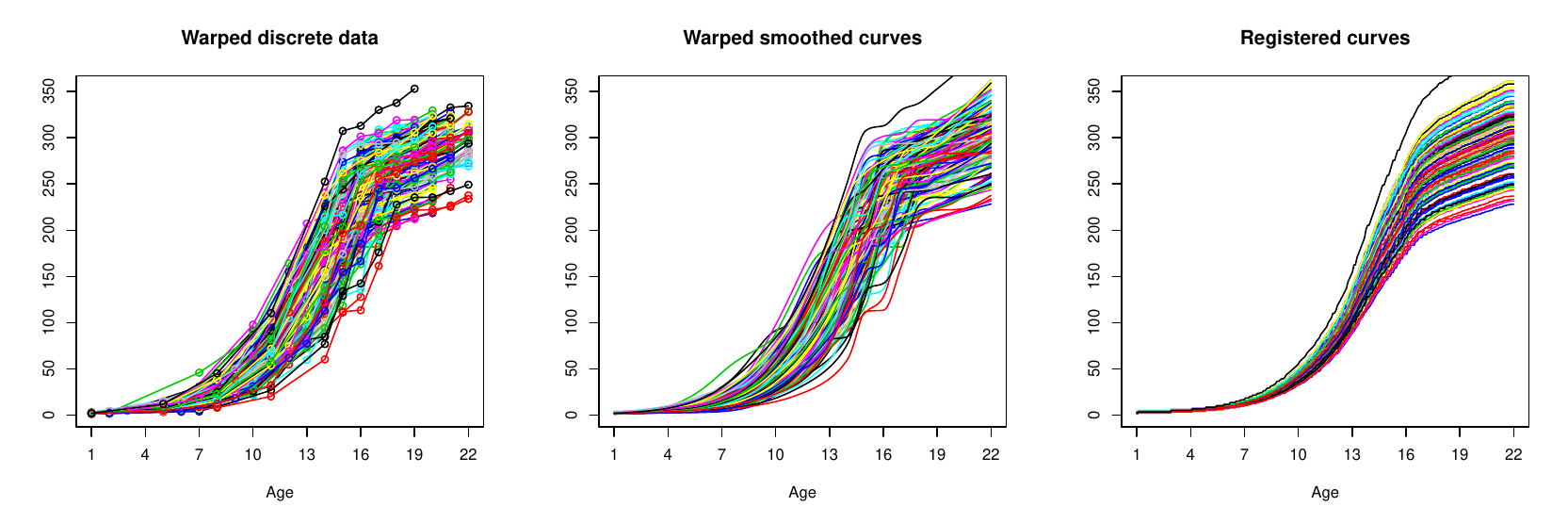}
\includegraphics[scale=0.46]{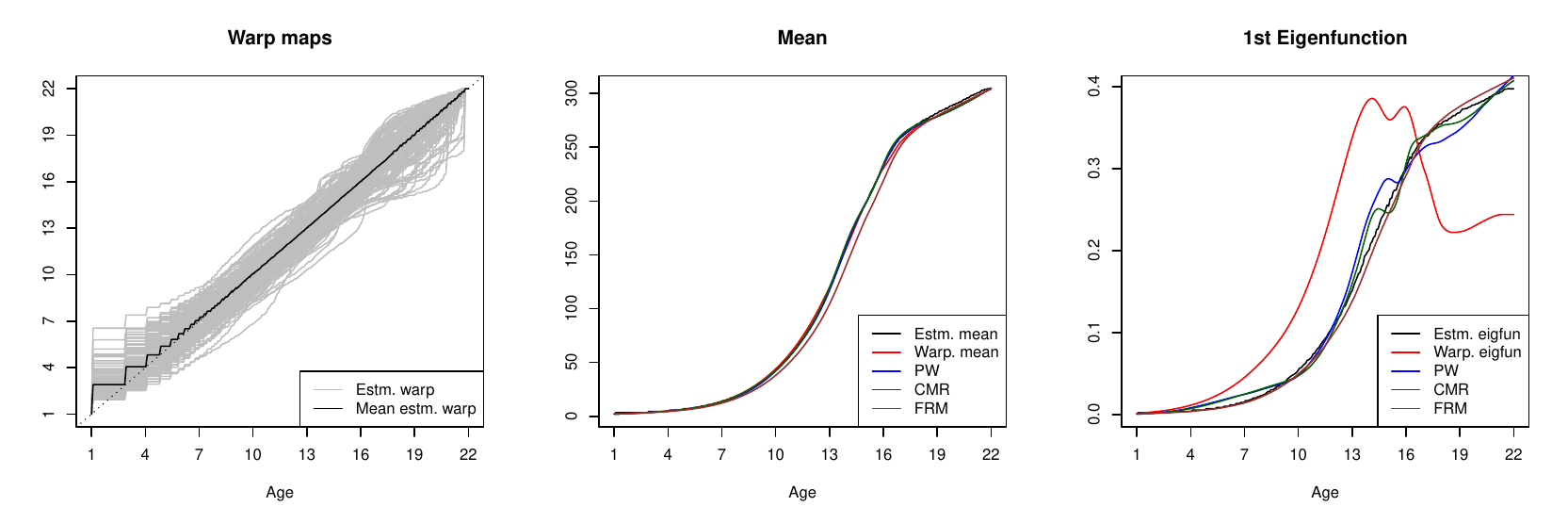}
}
\end{center}
\vspace{-0.2in}
\caption{Plots in the first row are those of the \textit{Tribolium} data, the smoothed curves and the registered curves using our procedure. The first plot in the second row shows the estimated warp maps, where the dotted line is the identity map. The other two plots in the second row show the means and the leading eigenfunctions of the warped and the registered data using our procedure and some other registration methods.}
\label{Fig3}
\end{figure}
\indent It is indeed observed that prior to any registration, the data present at least two substantial modes of amplitude variation, with the first three principal components explaining $78.4\%$, $12\%$ and $3.85\%$ of the total variation, respectively. However, after registration using our method, the empirical covariance operator is almost precisely of rank 1, with the leading principal component explaining $99.72\%$ of the total variation. Interestingly, the mean of the registered data has the same shape as the leading eigenfunction and is in fact roughly equal to $776$ times the leading eigenfunction. This can be seen as a model diagnostic, corroborating the model: if the rank 1 model were correct, then after registration one would expect to have a single mode of amplitude variation and a mean in the span of the corresponding eigenfunction (see the discussion after Counterexample \ref{counter_1}).

Figure \ref{Fig3} show the plots of the actual data, the monotone spline smoothed data and the registered data, as well as the plot of the estimated warp maps and the average warp map, which is very close to the identity. It also shows the plots of the mean and the leading eigenfunction of the warped and the registered data. Although the means of the warped and the registered data are very close, there are substantial qualitative differences between the corresponding eigenfunctions. The eigenfunction of the registered data shows that the variation in growth pattern essentially starts at about the $8$ days after hatching. Between ages $10-16$ days post hatching, there is a notable increase in the growth variation, and it somewhat recedes after that age. These periods are in fact compatible with biologically interpretable phases of growth: the larvae enter an ``instar" (a distinct growth period between exoskeleton moults) characterised by exponential growth at around day 7-8; then, around day 17, they enter the ``wandering phase" and begin losing weight in preparation for pupation.

The performance of the FRM technique is very similar to the proposed procedure and results in an almost rank one registration. However, the CMR and the PW procedures do not yield a rank one registration although the estimated means are very similar to that obtained by our procedure, which is observed by comparing Figure \ref{Fig3} with Figure \ref{Figsupp-beetle} in Appendix A (Supplementary Material). However, the difference lies in the registered curves and the estimate of the leading eigenfunction. The latter shows some artifacts which do not conform to the biological explanation provided earlier, e.g., the presence of flat regions in the estimated  eigenfunction during the ``instar'' phase of exponential growth as well as the growth spurt towards the end where the larvae would actually enter the ``wandering phase''. 

\section*{Acknowledgements}
We are grateful to Dr. Kristen Irwin (EPFL) for kindly sharing and discussing her \emph{Triboleum} data set. We also thank the reviewers for their constructive feedback.

\bibliographystyle{imsart-nameyear}

\bibliography{biblio1}


\newpage

\appendix

\section{Supplementary Material}\label{sec6}

\subsection{Some Practical Issues}\label{sec:implementation}

As mentioned earlier, $\widetilde{F}_{i,d}$ is a step function with jump discontinuities at the grid points. In particular, $\widetilde{F}_{i,d}(t) = 0$ for $t \in [0,t_{2})$ and $\widetilde{F}_{i,d}(t) = 1$ for $t \in [t_{r},1]$. 
Thus, $\widetilde{F}_{i,d}^{-}(0) = 0$ and $\widetilde{F}_{i,d}^{-}(1) = t_{r}$, which is less than $1$ if $t_{r} < 1$, i.e., the grid does not include the right end-point. In this case, $\widehat{F}_{d}(t)$ and thus $\widehat{T}_{i,d}(t)$ is properly defined only for $t \in [0,t_{r}]$. Also, $\widetilde{F}_{i,d}^{-}(u) \leq t_{r}$ and equality holds iff $u \in (\widetilde{F}_{i,d}(t_{r-1}),1]$. Thus,  
$$\widehat{F}_{d}(t_{r}) = \inf\left\{u : n^{-1}\sum_{i=1}^{n} \widetilde{F}_{i,d}^{-}(u) \geq t_{r}\right\} = \inf\{u : \widetilde{F}_{i,d}^{-}(u) = t_{r} ~ \forall ~ i=1,2,\ldots,n\}$$
$$ = \inf\{u : u \in \cap_{i=1}^{n} (\widetilde{F}_{i,d}(t_{r-1}),1]\} = \max_{1 \leq i \leq n} \widetilde{F}_{i,d}(t_{r-1}).$$ 
Then, $\widehat{T}_{i,d}(t_{r}) = \widehat{F}_{i,d}^{-}(\widehat{F}_{d}(t_{r})) = \widehat{F}_{i,d}^{-}(\max_{1 \leq j \leq n} \widetilde{F}_{j,d}(t_{r-1})) = t_{r}$. One can then extend $\widehat{T}_{i,d}(t)$ to the whole of $[0,1]$ by, e.g., linearly interpolating between $(t_{r},\widehat{T}_{i,d}(t_{r})) = (t_{r},t_{r})$ and $(1,1)$.  This practical modification, in case $t_{r} < 1$, enjoys the same asymptotic properties as the originally defined estimator (Section \ref{asymptotics}), since the effect of the modification is asymptotically negligible due to the homogeneity assumptions on the grid. 

Similarly, $\widehat{F}_{d}^{*}(u) = n^{-1}\sum_{i=1}^{n} \widetilde{F}_{i,d}^{-}(u) = t_{r}$ iff $u \in \cap_{i=1}^{n} (\widetilde{F}_{i,d}(t_{r-1}),1] = (\max_{1 \leq i \leq n} \widetilde{F}_{i,d}(t_{r-1}),1]$. So, in case $t_{r} < 1$, we have $\widehat{T}_{i,d}^{*}(1) = \widehat{F}_{d}^{*}(\widetilde{F}_{i,d}(1)) = \widehat{F}_{d}^{*}(1) = t_{r} < 1$. This is not a problem since this estimator is not used in the registration procedure and the problem disappears asymptotically anyway, just as described above. 

We conclude this section by noting that, since the estimates $\widehat{T}_{i,d}$ of the warp maps do not involve any smoothing and are obtained from compositions of step functions, the resulting registered curves will not be very smooth. This will be particularly noticeable if the number of grid points is small. Note that even in that case, the estimated mean function will be smoother if the sample size is moderately large. If one is interested in obtaining a smooth registration of the sample curves, the following procedure may be adopted. First, we produce smooth versions of the $\widehat{T}_{i,d}$ by some non-parametric smoothing procedure, e.g., polynomial splines of a fixed degree $m$, and call these new estimates as $\widehat{T}_{i,s}$, say. Then, we plug-in these smoothed estimates of the warp functions and define the new registered observations as $\widehat{X}_{i}^{*}(t) := X_{i}^{\dagger}(\widehat{T}_{i,s}(t))$. It is well-known that a spline smoothed estimate of a smooth function converges to that function in the $L_{2}[0,1]$ sense provided the oscillations of the function go to zero as the number of knots grows to infinity (see Theorem 6.27 in \cite{Schu07}). The latter holds for the $\widehat{T}_{i,d}$'s since they lie in $L_{2}[0,1]$ (see equation (2.121) in Theorem 2.59 in \cite{Schu07}). Thus, this modified estimator will also provide consistent registration.

\subsection{Proofs of Formal Statements}\label{proofs}

\begin{proof}[Proof of Lemma \ref{expectation_lemma}]
Since $X(t) = \xi\phi(t), t \in [0,1]$, we have 
$$F(t) = \int_{0}^{t} |X'(u)|du/\int_{0}^{1} |X'(u)|du = \int_{0}^{t} |\phi'(u)|du/\int_{0}^{1}|\phi'(u)|du = F_{\phi}(t)$$ 
by Definition \ref{local_variation}. Next, $\widetilde{X}(t) = \xi\phi(T^{-1}(t))$ so that $\widetilde{X}'(t) = \xi\phi'(T^{-1}(t))/T'(T^{-1}(t))$. Thus, using the strict monotonicity of $T$, we have 
\begin{eqnarray*}
\widetilde{F}(t) &=& \int_{0}^{t} |\widetilde{X}'(u)|du/\int_{0}^{1} |\widetilde{X}'(u)|du \\
&=& \left\{\int_{0}^{t} |\phi'(T^{-1}(u))|/T'(T^{-1}(u))du\right\}/\left\{\int_{0}^{t} |\phi'(T^{-1}(u))|/T'(T^{-1}(u))du\right\}.\end{eqnarray*}
A standard change-of-variable argument and the fact that $T$ is a bijection with $T(0) = 0$ and $T(1) = 1$ now yields $\widetilde{F}(t) = \int_{0}^{T^{-1}(t)} |\phi'(u)|du/\int_{0}^{1} |\phi'(u)|du = F_{\phi}(T^{-1}(t))$. So, $\widetilde{F} = F_{\phi} \circ T^{-1}$, equivalently, $T = \widetilde{F}^{-1} \circ F_{\phi} \leftrightarrow T \circ F_{\phi}^{-1} = \widetilde{F}^{-1}$. Using the assumption that $E(T) = Id$, we now have $E(\widetilde{F}^{-1}) = F_{\phi}^{-1}$.
\end{proof}

\begin{proof}[Proof of Theorem \ref{thm1}]
Note that $f : C^{1}[0,1] \mapsto f' \in (C[0,1],||\cdot||_{\infty})$ is a Lipschitz map. Thus, $\widetilde{X}_{1} \stackrel{d}{=} \widetilde{X}_{2}$ implies that $\widetilde{X}_{1}' \stackrel{d}{=} \widetilde{X}_{2}'$. Consider the random probability measure given by $$\Psi_{1}(A) = \int_{A} |\widetilde{X}_{1}'(u)|du/\int_{[0,1]} |\widetilde{X}_{1}'(u)|du$$ for $A$ in the Borel $\sigma$-field of $[0,1]$. 
Similarly, $\Psi_{2}(A) = \int_{A} |\widetilde{X}_{2}'(u)|du/\int_{[0,1]} |\widetilde{X}_{2}'(u)|du$.
We equip the space ${\cal P}$ of diffuse probability measures on $[0,1]$ with the $L_{2}$-Wasserstein metric (see, e.g., \cite{Vill03}) given by $d_{W}(\mu,\nu) = ||F_{\nu}^{-1} - F_{\mu}^{-1}||$, where $F_{\mu}$ and $F_{\nu}$ are the distribution functions associated with the probability measures $\mu$ and $\nu$. Now for any $f_{1}, f_{2} \in C^{1}[0,1]$ satisfying $\int_{0}^{1} |f_{i}'(u)|du > 0$ for $i=1,2$, consider the measure $\mu_{i}$ with density $|f_{i}'(s)|/\int_{0}^{1} |f_{i}'(u)|du$ for $i=1,2$. The condition $\int_{0}^{1} |f'(u)|du > 0$ is equivalent to $f \neq const.$. 
Since $\mu_{1}$ and $\mu_{2}$ are supported on the bounded set $[0,1]$, it follows from Proposition 7.10 in \cite{Vill03} that $d_{W}(\mu_{1},\mu_{2}) \leq c~d_{TV}(\mu_{1},\mu_{2})$ for a constant $c > 0$, where $d_{TV}(\cdot,\cdot)$ is the total variation distance. It now follows that
\begin{eqnarray*}
d_{W}(\mu_{1},\mu_{2}) &\leq& \frac{c}{2}\int_{0}^{1} \left| \frac{|f_{1}'(s)|}{\int_{0}^{1} |f_{1}'(u)|du} - \frac{|f_{2}'(s)|}{\int_{0}^{1} |f_{2}'(u)|du} \right| ds \\
&\leq& \frac{c}{2}\int_{0}^{1} \left| \frac{|f_{1}'(s)|}{\int_{0}^{1} |f_{1}'(u)|du} - \frac{|f_{1}'(s)|}{\int_{0}^{1} |f_{2}'(u)|du} \right| ds + \frac{c}{2}\int_{0}^{1} \left| \frac{|f_{1}'(s)|}{\int_{0}^{1} |f_{2}'(u)|du} - \frac{|f_{2}'(s)|}{\int_{0}^{1} |f_{2}'(u)|du} \right| ds \\
&\leq& \frac{c\int_{0}^{1} |f_{1}'(s) - f_{2}'(s)| ds}{\int_{0}^{1} |f_{1}'(s)| ds} \\
&\leq& \ \frac{c||f_{1}' - f_{2}'||_{\infty}}{\int_{0}^{1} |f_{1}'(s)| ds} \ \leq \ \frac{c|||f_{1} - f_{2}|||_{1}}{\int_{0}^{1} |f_{1}'(s)| ds}
\end{eqnarray*}
Thus, the embedding $H: f \mapsto \mu_{f}$ is continuous when the domain, say, ${\cal A}$ is restricted to the set of all non-constant functions on $C^{1}[0,1]$. But the set ${\cal A}^{c}$ is a one dimensional linear subspace spanned by the constant function $f \equiv 1$, and this implies that ${\cal A}^{c}$ is a Borel measurable subset of $C^{1}[0,1]$. So, ${\cal A}$ is a Borel measurable subset of $C^{1}[0,1]$. Equip ${\cal A}$ with the Borel $\sigma$-field induced from $C^{1}[0,1]$. Since $P(\widetilde{X}_{1} \in {\cal A}^{c}) = 0$, we have that $H(\widetilde{X}_{1})$ is a valid random probability measure on $[0,1]$. Note that for any Borel subset $A$ of $[0,1]$, we have $H(\widetilde{X}_{1})(A) = \Psi_{1}(A)$. Thus, for any Borel subset $B$ of ${\cal P}$, we have
\begin{eqnarray*}
P(H(\widetilde{X}_{1}) \in B) = P(\widetilde{X}_{1} \in H^{-1}(B)) = P(\widetilde{X}_{2} \in H^{-1}(B)) = P(H(\widetilde{X}_{2}) \in B).
\end{eqnarray*}
The first equality follows from the continuity of $H$ on ${\cal A}$ and the fact that $P(\widetilde{X}_{1} \in {\cal A}^{c}) = 0$ discussed above. The second equality follows from the fact that $\widetilde{X}_{1}$ and $\widetilde{X}_{2}$ have the same distributions by assumption. So, $H(\widetilde{X}_{1}) \stackrel{d}{=} H(\widetilde{X}_{2})$ as random probability measures.  \\
\indent Next, note that the random measures $H(\widetilde{X}_{i})$, $i=1,2$, have strictly increasing cdfs almost surely. Proposition $2$ in \cite{PZ16} states that for each $i=1,2$, the map $\gamma \rightarrow E\{d^{2}_{W}(H(\widetilde{X}_{i}),\gamma)\}$ admits a unique minimizer given by $E\{\widetilde{F}_{\Psi_{i}}^{-1}\}$, where $\widetilde{F}_{\Psi_{i}}$ is the random distribution function of the random measure $H(\widetilde{X}_{i})$. Since $\widetilde{X}_{i} = \xi_{i}\phi_{i}(T_{i}^{-1})$ with $T_{i}$ being a strictly increasing homeomorphism on $[0,1]$, it follows from the change-of-variable formula that $H(\widetilde{X}_{i})(A) = \Psi_{i}(A) = \int_{T_{i}^{-1}(A)} |\phi_{i}'(u)|du/\int_{[0,1]} |\phi_{i}'(u)|du$. Thus, $\widetilde{F}_{\Psi_{i}} = F_{\phi_{i}} \circ T_{i}^{-1}$, equivalently, $\widetilde{F}_{\Psi_{i}}^{-1} = T_{i} \circ F_{\phi_{i}}^{-1}$, where $F_{\phi_{i}}$ is the cdf associated with the (deterministic) probability measure $\Phi_{i}(A) = \int_{A} |\phi_{i}'(u)|du/\int_{[0,1]} |\phi_{i}'(u)|du$. \\
\indent Note that $F_{\phi_{i}}$ has a continuous and strictly increasing cdf since $\phi_{i}'$ is zero only on a countable set for $i=1,2$. Since $E(T_{i}) = Id$, it follows that the minimizer $E\{\widetilde{F}_{\Psi_{i}}^{-1}\} = F_{\phi_{i}}$ for $i=1,2$. But since $H(\widetilde{X}_{1}) \stackrel{d}{=} H(\widetilde{X}_{2})$, it now follows that $F_{\phi_{1}} = F_{\phi_{2}}$. Also, $T_{i} = \widetilde{F}_{\Psi_{i}}^{-1} \circ F_{\phi_{i}}$, equivalently, $T_{i}^{-1} = F_{\phi_{i}}^{-1} \circ \widetilde{F}_{\Psi_{i}}$. Using the above facts and the result obtained in the previous paragraph, it now follows that $T_{1} \stackrel{d}{=} T_{2}$. \\
\indent We next claim that the joint distributions of $(\widetilde{X}_{i},T_{i}^{-1})$, $i=1,2$ are the same. To this end, consider the map $H_{1} : f \mapsto (f,H(f))$ defined from ${\cal A}$ to ${\cal A} \otimes {\cal P}$ with the latter being equipped with the induced product topology and the induced product $\sigma$-field. It follows from the same arguments used to prove the continuity of $H$ that $H_{1}$ is continuous. Thus, for Borel subsets $G_{1}$ and $G_{2}$ of $C^{1}[0,1]$, we have
\begin{eqnarray*}
P(\widetilde{X}_{1} \in G_{1}, T_{1}^{-1} \in G_{2}) &=& P(\widetilde{X}_{1} \in G_{1}, F_{\phi_{1}}^{-1} \circ \widetilde{F}_{\Psi_{1}} \in G_{2}) 
= P(\widetilde{X}_{1} \in G_{1}, \widetilde{F}_{\Psi_{1}} \in F_{\phi_{1}}(G_{2})) \\
&=& P(H_{1}(\widetilde{X}_{1}) \in G_{1} \times F_{\phi_{1}}(G_{2})) 
= P(\widetilde{X}_{1} \in H_{1}^{-1}(G_{1} \times F_{\phi_{1}}(G_{2}))) \\
&=& P(\widetilde{X}_{2} \in H_{1}^{-1}(G_{1} \times F_{\phi_{2}}(G_{2}))) \ \ \ \ \mbox{[since $F_{\phi_{1}} = F_{\phi_{2}}$]}\\
&=& P(H_{1}(\widetilde{X}_{2}) \in G_{1} \times F_{\phi_{2}}(G_{2})) 
= P(\widetilde{X}_{2} \in G_{1}, \widetilde{F}_{\Psi_{2}} \in F_{\phi_{2}}(G_{2})) \\
&=& P(\widetilde{X}_{2} \in G_{1}, F_{\phi_{2}}^{-1} \circ \widetilde{F}_{\Psi_{2}} \in G_{2}) 
= P(\widetilde{X}_{2} \in G_{1}, T_{2}^{-1} \in G_{2}).
\end{eqnarray*}
Next, note that $X_{i} = \widetilde{X}_{i} \circ T_{i}$ is the true unobserved process. It is easy to show that the map $(f,g) \mapsto f \circ g$ from $C^{1}[0,1] \otimes C^{1}[0,1]$ into $C^{1}[0,1]$ is continuous. Thus, using the observation in the previous paragraph, we have $X_{1} \stackrel{d}{=} X_{2}$ as random elements in $C^{1}[0,1]$. It follows from the equality of distributions that their covariance operators are equal, and thus the corresponding eigenfunctions are equal. Now, the covariance operator of $X_{i}$ is given by $Var(\xi_{i})\phi_{i} \otimes \phi_{i}$. Since $X_{i} = \xi_{i}\phi_{i}$ is a rank one process, the equality of the covariance operators implies that $\phi_{1} = \pm \phi_{2}$ (since $||\phi_{1}|| = ||\phi_{2}|| = 1$). This equality along with the fact that $X_{1} \stackrel{d}{=} X_{2}$ implies that $\xi_{1}  = \langle X_{1},\phi_{1}\rangle \stackrel{d}{=} \langle X_{2},\phi_{1}\rangle = \langle X_{2},\pm\phi_{2}\rangle = \pm \xi_{2}$.
\end{proof}

We will now state and prove the statement regarding identifiability of the rank two model obtained as the rank 1 model plus a random vertical shift, as discussed in Remark \ref{identifiable-rank-two-model}. 

\smallskip

A general rank two model would be given by the Karhunen-Lo\`eve-like expansion $X(t) = \xi\phi(t) + \zeta\psi(t)$, where $\xi$ and $\zeta$ are uncorrelated random variables, and $\xi$ and $\psi$ are orthonormal eigenfunctions of the covariance operator of $X$ (we assume here that the mean function lies in the range of the covariance operator of $X$). The rank two model discussed in Remark \ref{identifiable-rank-two-model} satisfies $\psi \equiv 1$. Thus, the following conditions are satisfied for the model $X(t) = \zeta + \xi\phi(t)$: 
\begin{itemize}
\item[(a)] $\int_0^1 \phi(t)dt$ = 0, and 
\item[(b)] $\xi$ and $\phi$ are uncorrelated.
\end{itemize}
We call these two conditions \emph{Condition (R)}.

\begin{theorem} \label{identifiable-rank-two-model-thm}
Consider two random elements $\{X_1,X_2\}$ in $C^{1}[0,1]$ defined as $X_i(t) = \zeta_i + \xi_i\phi_i(t)$ for real-valued random variables $\zeta_i$, deterministic functions $\phi_i \in C^{(1)}[0,1]$ with $\int_{0}^{1} \phi_i^{2}(t)dt = 1$ and $\phi'_i$ vanishing on at most a countable set. Assume that both $X_1$ and $X_2$ satisfy Condition (R) above. Further, assume that $\{T_1,T_2\}$ are strictly increasing homeomorphisms in $C^{1}[0,1]$, and such that $E(T_i) = Id$. Write $\widetilde{X}_{i}=X_i(T_i^{-1}(t))$ for $i=1,2$. Then, 
$$\widetilde{X}_{1}\stackrel{d}{=}\widetilde{X}_{2} \iff \Big\{\zeta_1 \stackrel{d}{=} \zeta_{2},\quad T_{1} \stackrel{d}{=} T_{2},\quad \phi_{1} = \pm\phi_{2},\quad  \xi_{1} = \pm\xi_{2}\Big\}.$$ 
\end{theorem}

\begin{proof}
Note that $f : C^{1}[0,1] \mapsto f' \in (C[0,1],||\cdot||_{\infty})$ is a Lipschitz map. Thus, $\widetilde{X}_{1} \stackrel{d}{=} \widetilde{X}_{2}$ implies that $\widetilde{X}_{1}' \stackrel{d}{=} \widetilde{X}_{2}'$. Observe that $\widetilde{X}_i'$ is the same as the derivative of the process $\widetilde{Z}_i$, where $Z_i$ is a rank one process defined by $Z_i = \xi_i\phi_i$, and $\widetilde{Z}_i = Z_i \circ T_i^{-1}$. So, the proof of Theorem \ref{thm1} shows that $T_{1} \stackrel{d}{=} T_{2}$. It also follows from that proof the joint distributions of $(\widetilde{X}_1,T_1)$ and $(\widetilde{X}_2,T_2)$ are the same. Hence, $X_1 \stackrel{d}{=} X_2$, i.e., $\zeta_1\mathbf{1} + \xi_1\phi_1 \stackrel{d}{=} \zeta_2\mathbf{1} + \xi_2\phi_2$, where the constant function one is denoted by $\mathbf{1}$. 

Observe that by Condition (R), $\mathbf{1} \in \{\mbox{span}\{\phi_i\}\}^{\perp}$ for $i=1,2$. So,   $\mathbf{1} \in \{\mbox{span}\{\phi_1,\phi_2\}\}^{\perp}$, and hence there exists a non-zero $\phi \in C^{1}[0,1]$ such that $\langle\phi,\phi_1\rangle = \langle\phi,\phi_2\rangle = 0$ and $\langle\mathbf{1},\phi\rangle \neq 0$. Thus,
$$ {\zeta_1}\langle\mathbf{1},\phi\rangle = \langle{\zeta_1}\mathbf{1} + \xi_1\phi_1,\phi\rangle \stackrel{d}{=} \langle{\zeta_2}\mathbf{1} + \xi_2\phi_2,\phi\rangle = {\zeta_2}\langle\mathbf{1},\phi\rangle, $$
which implies that $\zeta_1 \stackrel{d}{=} \zeta_2$. 

Next observe that
\begin{eqnarray*}
\mathrm{Cov}(\zeta_i\mathbf{1} + \xi_i\phi_i) = \mathrm{Var}(\zeta_i)(\mathbf{1} \otimes \mathbf{1}) + \mathrm{Var}(\xi_i)(\phi_i \otimes \phi_i) + \mathrm{Cov}(\zeta_i,\xi_i)(\mathbf{1} \otimes \phi_i + \phi_i \otimes \mathbf{1})
\end{eqnarray*}
for each $i=1,2$. Using Condition (R) along with the fact that $X_1 \stackrel{d}{=} X_2$ (shown previously), it now follows that
\begin{eqnarray*}
\mathrm{Var}(\zeta_1)(\mathbf{1} \otimes \mathbf{1}) + \mathrm{Var}(\xi_1)(\phi_1 \otimes \phi_1) = \mathrm{Var}(\zeta_2)(\mathbf{1} \otimes \mathbf{1}) + \mathrm{Var}(\xi_2)(\phi_2 \otimes \phi_2).
\end{eqnarray*}
Using the fact that $\zeta_1 \stackrel{d}{=} \zeta_2$, the above equality implies that 
$$\mathrm{Var}(\xi_1)(\phi_1 \otimes \phi_1) = \mathrm{Var}(\xi_2)(\phi_2 \otimes \phi_2).$$
Since $||\phi_1|| = ||\phi_2|| = 1$, the above equality implies that $\phi_1 = \pm\phi_2$. Finally, there exists a non-zero $\phi_*$ such that $\langle\mathbf{1},\phi_*\rangle = 0$ and $\langle\phi_1,\phi_*\rangle \neq 0$. So,
$$ {\xi_1}\langle\phi_1,\phi_*\rangle = \langle{Y_1}\mathbf{1} + \xi_1\phi_1,\phi_*\rangle \stackrel{d}{=} \langle{Y_2}\mathbf{1} + \xi_2\phi_2,\phi_*\rangle = {\xi_2}\langle\phi_2,\phi_*\rangle = \pm\xi_2\langle\phi_1,\phi_*\rangle, $$
which implies that $\xi_1 \stackrel{d}{=} \pm\xi_2$.
\end{proof}

\begin{proof}[Proof of Theorem \ref{thm2}]
First observe that the $T_{i}$'s are also i.i.d. random elements in $C[0,1]$. Moreover, since $T_{1}$ is strictly increasing and positive, we have $E(||T_{1}||_{\infty}) = E(T_{1}(1)) = 1 < \infty$. Thus, by the strong law for Banach space valued random elements (see, e.g, Theorem 2.4 in \cite{Bosq00}), it follows that $\overline{T} \rightarrow E(T_{1}) = Id$ as $n \rightarrow \infty$ almost surely. In addition, if $E(||T_{1}'||_{\infty}) < \infty$ implying that $E(|||T_{1}|||_{1}) < \infty$, then the almost sure convergence $\overline{T} \rightarrow E(T_{1}) = Id$ holds in $C^{1}[0,1]$. \\
(a) Since $\widehat{F}^{-1} = \overline{T} \circ F_{\phi}^{-1}$, using Theorem 2.18 in \cite{Vill03}, we get that
\begin{eqnarray*}
d^{2}_{W}(\widehat{F},F_{\phi}) &=& ||\widehat{F}^{-1} - F_{\phi}^{-1}||^{2} \\
&=& \int_{0}^{1} \left|\widehat{F}^{-1}(F_{\phi}(t)) - t\right|^{2}F_{\phi}(dt) \\
&=& \int_{0}^{1} \left|\overline{T}(t) - t\right|^{2}F_{\phi}dt \ \leq \ ||\overline{T} - Id||_{\infty}^{2} \ \rightarrow \ 0 \ \ \mbox{as $n \rightarrow \infty$}.
\end{eqnarray*}
(b) Since each $T_{i}$ is a strictly increasing bijection on $[0,1]$, we have
\begin{eqnarray*}
||\widehat{T}_{i}^{-1} - T_{i}^{-1}||_{\infty} &=& \sup_{t \in [0,1]} \left|\overline{T}(T_{i}^{-1}(t)) - T_{i}^{-1}(t)\right| \ = \ ||\overline{T} - Id||_{\infty} \ \rightarrow \ 0 \ \ \mbox{as $n \rightarrow \infty$}.
\end{eqnarray*}
Since both $\widehat{T}_{i}^{-1}$ and $T_{i}^{-1}$ are strictly increasing homeomorphisms, the uniform convergence of $\widehat{T}_{i}$ to $T_{i}$ follows as a consequence of the above uniform convergence. \\
\indent Suppose now that Condition 1 holds. We have discussed towards the beginning of the proof that in this case $|||\overline{T} - Id|||_{1} \rightarrow 0$ as $n \rightarrow \infty$ almost surely. In view of the first half of part (b) of the theorem along with the definition of the $|||\cdot|||_{1}$ norm, it is enough to show the uniform convergence of the derivatives. Since each $T_{i}$ is a strictly increasing bijection on $[0,1]$, so is $\overline{T}$ for every $n \geq 1$. First note that
\begin{eqnarray*}
||(\widehat{T}_{i}^{-1})' - (T_{i}^{-1})'||_{\infty} &=& \sup_{t \in [0,1]} |(\overline{T} \circ T_{i}^{-1})'(t) - (T_{i}^{-1})'(t)| 
= \sup_{t \in [0,1]} \left|\frac{\overline{T}'(T_{i}^{-1}(t))}{T_{i}'(T_{i}^{-1}(t))} - \frac{1}{T_{i}'(T_{i}^{-1}(t))} \right| \\
&=& \sup_{t \in [0,1]} \left|\frac{\overline{T}'(t) - 1}{T_{i}'(t)}\right| \leq \delta^{-1}||\overline{T}' - \mathbf{1}||_{\infty},
\end{eqnarray*}
where $\mathbf{1}$ is the constant function taking value $1$. It thus follows from an earlier bound that
\begin{eqnarray*}
|||(\widehat{T}_{i}^{-1})' - (T_{i}^{-1})'|||_{1} \leq ||\overline{T} - Id||_{\infty} + \delta^{-1}||\overline{T}' - \mathbf{1}||_{\infty} \leq \max(1,\delta^{-1}) |||\overline{T} - Id|||_{1} \ \rightarrow \ 0 \ \ \mbox{as $n \rightarrow \infty$}.
\end{eqnarray*}
\indent Next note that $\overline{T}'(t) = n^{-1} \sum_{i=1}^{n} T_{i}'(t) \geq n^{-1} \sum_{i=1}^{n} \inf_{s \in [0,1]} T_{i}'(t) = \delta$ so that $\inf_{t \in [0,1]} \overline{T}'(t) \geq \delta > 0$. Now,
\begin{eqnarray*}
||\widehat{T}_{i}' - T_{i}'||_{\infty} &=& \sup_{t \in [0,1]} |(T_{i} \circ \overline{T}^{-1})'(t) - T_{i}'(t)| 
= \sup_{t \in [0,1]} \left|\frac{T_{i}'(\overline{T}^{-1}(t))}{\overline{T}'(\overline{T}^{-1}(t))} - T_{i}'(t) \right| 
= \sup_{t \in [0,1]} \left|\frac{T_{i}'(t)}{\overline{T}'(t)} - T_{i}'(\overline{T}(t))\right| \\
&\leq& \sup_{t \in [0,1]} \left|\frac{T_{i}'(t)}{\overline{T}'(t)} - \frac{T_{i}'(\overline{T}(t))}{\overline{T}'(t)} \right| + \sup_{t \in [0,1]} \left|\frac{T_{i}'(\overline{T}(t))}{\overline{T}'(t)} - T_{i}'(\overline{T}(t)) \right|  \\
&\leq& \delta^{-1} \sup_{t \in [0,1]} \left|T_{i}'(t) - T_{i}'(\overline{T}(t)) \right| + \delta^{-1}||T_{i}'||_{\infty}||\overline{T}' - \mathbf{1}||_{\infty}.
\end{eqnarray*}
Since $T_{i}'$ is continuous on $[0,1]$, it is uniformly continuous. This and the fact that $||\overline{T} - Id||_{\infty} \rightarrow 0$ as $n \rightarrow \infty$ almost surely implies that $\sup_{t \in [0,1]} \left|T_{i}'(t) - T_{i}'(\overline{T}(t)) \right| \rightarrow 0$ as $n \rightarrow \infty$ almost surely. Combining this fact with the uniform convergence of $\overline{T}'$ to $\mathbf{1}$, we get that $|||\widehat{T}_{i} - T_{i}|||_{1} \rightarrow 0$ as $n \rightarrow \infty$ almost surely. \\
(c) Note that 
\begin{eqnarray*}
||\widehat{X}_{i} - X_{i}||_{\infty} = |\xi_{i}|\sup_{t \in [0,1]} |\phi(\overline{T}^{-1}(t)) - \phi(t)| = |\xi_{i}|\sup_{t \in [0,1]} |\phi(\overline{T}(t)) - \phi(t)| \ \rightarrow \ 0 \ \ \mbox{as $n \rightarrow \infty$},
\end{eqnarray*}
since $||\overline{T} - Id||_{\infty} \rightarrow 0$ as $n \rightarrow \infty$ almost surely, and $\phi$ is continuous on $[0,1]$ and hence uniformly continuous. \\
\indent Suppose now that Condition 1 holds. Then, as before,
\begin{eqnarray*}
||\widehat{X}_{i}' - X_{i}'||_{\infty} &=& |\xi_{i}| \sup_{t \in [0,1]} \left|\frac{\phi'(\overline{T}^{-1}(t))}{\overline{T}'(\overline{T}^{-1}(t))} - \phi'(t) \right| 
= |\xi_{i}| \sup_{t \in [0,1]} \left|\frac{\phi'(t)}{\overline{T}'(t)} - \phi'(\overline{T}(t))\right| \\
&\leq& |\xi_{i}| \sup_{t \in [0,1]} \left|\frac{\phi'(t)}{\overline{T}'(t)} - \frac{\phi'(\overline{T}'(t))}{\overline{T}'(t)}\right| + |\xi_{i}| \sup_{t \in [0,1]} \left|\frac{\phi'(\overline{T}'(t))}{\overline{T}'(t)} - \phi'(\overline{T}(t))\right| \\
&\leq& |\xi_{i}|\delta^{-1} \sup_{t \in [0,1]} |\phi'(t) - \phi'(\overline{T}'(t))| + |\xi_{i}|~||\phi'||_{\infty}\delta^{-1}||\overline{T}' - \mathbf{1}||_{\infty}.
\end{eqnarray*}
Using similar arguments as earlier, we conclude that $||\widehat{X}_{i}' - X_{i}'||_{\infty}  \rightarrow 0$ and hence $|||\widehat{X}_{i} - X_{i}|||_{1} \rightarrow 0$ as $n \rightarrow \infty$ almost surely. \\
(d) Observe that since $\widehat{X}_{i} = \xi_{i}\phi \circ \overline{T}^{-1} = X_{i} \circ \overline{T}^{-1}$, it follows from the change-of-variable formula that $\widehat{F}_{i} = F_{\phi} \circ \overline{T}^{-1}$. Thus,
\begin{eqnarray*}
d^{2}_{W}(\widehat{F}_{i},F_{\phi}) \ = \ ||\widehat{F}_{i}^{-1} - F_{\phi}^{-1}||^{2} 
\ = \ ||\overline{T} \circ F_{\phi}^{-1} - F_{\phi}^{-1}||^{2} 
&=& \int_{0}^{1} \left|\overline{T}(t) - t\right|^{2}F_{\phi}(dt) \\
&\leq& \ ||\overline{T} - Id||_{\infty}^{2} \ \rightarrow \ 0 \ \ \mbox{as $n \rightarrow \infty$}.
\end{eqnarray*}
(e) Observe that
\begin{eqnarray*}
||\overline{X}_{r} - \mu||_{\infty} &=& ||n^{-1}\sum_{i=1}^{n} (\widehat{X}_{i} - X_{i}) + n^{-1} \sum_{i=1}^{n} X_{i} - \mu||_{\infty} 
\leq n^{-1} \sum_{i=1}^{n} ||\widehat{X}_{i} - X_{i}||_{\infty} + ||n^{-1}\sum_{i=1}^{n} X_{i} - \mu||_{\infty}.
\end{eqnarray*}
Since the $X_{i}$'s are i.i.d. random elements in $C[0,1]$ with $E(||X_{1}||_{\infty}) = E(|\xi_{1}|)||\phi||_{\infty} < \infty$, we conclude from the strong law for Banach space valued random elements that $||n^{-1}\sum_{i=1}^{n} X_{i} - \mu||_{\infty} \rightarrow 0$ as $n \rightarrow \infty$ almost surely. Also, from the proof of part (c), we have that
\begin{eqnarray*}
n^{-1} \sum_{i=1}^{n} ||\widehat{X}_{i} - X_{i}||_{\infty} = \sup_{t \in [0,1]} |\phi(\overline{T}(t)) - \phi(t)| \times n^{-1} \sum_{i=1}^{n} |\xi_{i}| = \sup_{t \in [0,1]} |\phi(\overline{T}(t)) - \phi(t)| \times \{E(|\xi_{1}|) + o(1)\}
\end{eqnarray*}
as $n \rightarrow \infty$ almost surely. Thus, using similar arguments as in part (c) of the theorem, we obtain $n^{-1} \sum_{i=1}^{n} ||\widehat{X}_{i} - X_{i}||_{\infty} \rightarrow 0$ as $n \rightarrow \infty$ almost surely. Combining the above facts, we conclude $||\overline{X}_{r} - \mu||_{\infty} \rightarrow 0$ as $n \rightarrow \infty$ almost surely. \\
\indent Note that since $X_{i} = \xi_{i}\phi$, it follows that $||n^{-1}\sum_{i=1}^{n} X_{i}' - \mu'||_{\infty} \rightarrow 0$ as $n \rightarrow \infty$ almost surely. Now, suppose that Condition 1 holds. A similar decomposition as above yields
\begin{eqnarray*}
||\overline{X}_{r}' - \mu'||_{\infty} &\leq& n^{-1} \sum_{i=1}^{n} ||\widehat{X}_{i}' - X_{i}'||_{\infty} + ||n^{-1}\sum_{i=1}^{n} X_{i}' - \mu'||_{\infty}.
\end{eqnarray*}
The proof of part (c) implies that
\begin{eqnarray*}
n^{-1} \sum_{i=1}^{n} ||\widehat{X}_{i}' - X_{i}'||_{\infty} \leq \delta^{-1}\left(n^{-1} \sum_{i=1}^{n} |\xi_{i}|\right) \left\{\sup_{t \in [0,1]} |\phi'(t) - \phi'(\overline{T}'(t))| + ||\phi'||_{\infty}||\overline{T}' - \mathbf{1}||_{\infty}\right\}.
\end{eqnarray*}
The right-hand term above converges to zero as $n \rightarrow \infty$ almost surely. The result is now established upon combining the above facts. \\
(f) Straightforward algebraic manipulations yield
\begin{eqnarray*}
\widehat{\mathscr{K}}_{r} &=& n^{-1} \sum_{i=1}^{n} (\widehat{X}_{i} - \overline{X}_{r}) \otimes (\widehat{X}_{i} -  \overline{X}_{r}) \\
&=& n^{-1} \sum_{i=1}^{n} (X_{i} - \overline{X}) \otimes (X_{i} - \overline{X}) + n^{-1} \sum_{i=1}^{n} (\widehat{X}_{i} - X_{i}) \otimes (\widehat{X}_{i} - X_{i}) - (\overline{X} - \overline{X}_{r}) \otimes (\overline{X} - \overline{X}_{r}) \\
&& + \ n^{-1} \sum_{i=1}^{n} \{(\widehat{X}_{i} - X_{i}) \otimes (X_{i} - \overline{X}) + (X_{i} - \overline{X}) \otimes (\widehat{X}_{i} - X_{i})\}.
\end{eqnarray*}
Denote $\widehat{\mathscr{K}} =  n^{-1} \sum_{i=1}^{n} (X_{i} - \overline{X}) \otimes (X_{i} - \overline{X})$. Then,
\begin{eqnarray*}
|||\widehat{\mathscr{K}}_{r} - \widehat{\mathscr{K}}||| \leq \frac{2}{n}\sum_{i=1}^{n} ||\widehat{X}_{i} - X_{i}||~||X_{i} - \overline{X}|| + \frac{1}{n} \sum_{i=1}^{n} ||\widehat{X}_{i} - X_{i}||^{2} + ||\overline{X} - \overline{X}_{r}||^{2}.
\end{eqnarray*}
Using the Cauchy-Schwarz inequality, we have $n^{-1}\sum_{i=1}^{n} ||\widehat{X}_{i} - X_{i}||~||X_{i} - \overline{X}|| \leq \{n^{-1}\sum_{i=1}^{n} ||\widehat{X}_{i} - X_{i}||^{2}\}^{1/2}\{n^{-1} \sum_{i=1}^{n} ||X_{i} - \overline{X}||^{2}\}^{1/2}$, and $n^{-1} \sum_{i=1}^{n} ||X_{i} - \overline{X}||^{2} = O(1)$  as $n \rightarrow \infty$ almost surely. It follows from the arguments in the proof of part (c) of the theorem that
\begin{eqnarray*}
n^{-1}\sum_{i=1}^{n} ||\widehat{X}_{i} - X_{i}||^{2} \leq n^{-1} \sum_{i=1}^{n} ||\widehat{X}_{i} - X_{i}||_{\infty}^{2} \leq \sup_{t \in [0,1]} |\phi(\overline{T}(t)) - \phi(t)|^{2} \left(n^{-1} \sum_{i=1}^{n} |\xi_{i}|^{2}\right),
\end{eqnarray*}
and the right hand side is $o(1)$ as $n \rightarrow \infty$ almost surely since $E(|\xi_{1}|^{2}) < \infty$. Further, $||\overline{X} - \overline{X}_{r}||^{2} = o(1)$ as $n \rightarrow \infty$ almost surely. Thus, $|||\widehat{\mathscr{K}}_{r} - \widehat{\mathscr{K}}||| = o(1)$ as $n \rightarrow \infty$ almost surely. \\ 
\indent The proof of the uniform convergence of $\widehat{K}_{r}(s,t)$ to $K(s,t)$ is obtained by use of a decomposition of $\widehat{K}_{r}(s,t)$ similar to the one used above, noting that $\widehat{K}(s,t)$ converges uniformly to $K(s,t)$ (by the strong law of large numbers in $C([0,1]^{2})$), and the fact that all the other bounds hold in the supremum norm. \\
\indent Next, note that $\widehat{\phi}(t) = \widehat{\lambda}^{-1} \int_{0}^{1} \widehat{K}_{r}(s,t)\widehat{\phi}(s)ds$ and $\phi(t) = \lambda^{-1} \int_{0}^{1} K(s,t)\phi(s)ds$ for all $t \in [0,1]$, where $|\widehat{\lambda} - \lambda| \leq |||\widehat{\mathscr{K}}_{r} - \mathscr{K}||| \rightarrow 0$ as $n \rightarrow \infty$ almost surely. Also, $||\widehat{\phi} - \phi|| \leq 2\sqrt{2}\lambda^{-1}|||\widehat{\mathscr{K}}_{r} - \mathscr{K}||| \rightarrow 0$ as $n \rightarrow \infty$ almost surely. So,
\begin{eqnarray*}
|\widehat{\phi}(t) - \phi(t)| &\leq& \left|\widehat{\lambda}^{-1} \int_{0}^{1} \widehat{K}_{r}(s,t)\widehat{\phi}(s)ds - \widehat{\lambda}^{-1} \int_{0}^{1} K(s,t)\widehat{\phi}(s)ds\right| \\
&& + \ \left|\widehat{\lambda}^{-1} \int_{0}^{1} K(s,t)\widehat{\phi}(s)ds - \widehat{\lambda}^{-1} \int_{0}^{1} K(s,t)\phi(s)ds\right| \\
&& + \ \left|\widehat{\lambda}^{-1} \int_{0}^{1} K(s,t)\phi(s)ds - \lambda^{-1} \int_{0}^{1} K(s,t)\phi(s)ds\right| \\
&\leq& \widehat{\lambda}^{-1} ||\widehat{K}_{r} - K||_{\infty} + \widehat{\lambda}^{-1}||K||_{\infty}||\widehat{\phi} - \phi|| + \left|(\widehat{\lambda}^{-1} - \lambda^{-1})\lambda\phi(t)\right| \\
&\leq& (\lambda^{-1} + o(1))\{||\widehat{K}_{r} - K||_{\infty} + ||K||_{\infty}||\widehat{\phi} - \phi||\} + |\lambda - \widehat{\lambda}|(\lambda^{-1} + o(1))^{-1}||\phi||_{\infty}
\end{eqnarray*}
as $n \rightarrow \infty$ almost surely. Thus, $||\widehat{\phi} - \phi||_{\infty} \rightarrow 0$ as $n \rightarrow \infty$ almost surely. \\
\indent Finally, $|\widehat{\xi}_{i} - \xi_{i}| = |\langle\widehat{X}_{i},\widehat{\phi}\rangle - \langle X_{i},\phi\rangle| \leq |\langle \widehat{X}_{i} - X_{i},\widehat{\phi}\rangle| + |\langle X_{i}, \widehat{\phi} - \phi\rangle| \leq ||\widehat{X}_{i} - X_{i}||_{\infty} + ||\widehat{\phi} - \phi|| \rightarrow 0$ as $n \rightarrow \infty$ almost surely.
\end{proof}
\begin{proof}[Proof of Theorem \ref{thm3}]
We have $|T_{1}(t) - T_{1}(s)| \leq ||T_{1}'||_{\infty}|s - t|$ and by assumption $E(||T_{1}'||_{\infty}^{2}) < \infty$. So, by the CLT for i.i.d. $C[0,1]$ valued random elements (see, e.g, Theorem 2.4 \cite{Bosq00}), we have $\sqrt{n}(\overline{T} - Id) \stackrel{d}{\rightarrow} Y$ for a zero mean Gaussian random element $Y$ in $C[0,1]$. \\
(a) From the proof of part (a) of Theorem \ref{thm2}, one has that $d^{2}_{W}(\widehat{F},F_{\phi}) = \int_{0}^{1} |\overline{T}(t) - t|^{2}F_{\phi}(dt)$. Now, it is easy to check that the map $C[0,1] \ni f \rightarrow \int_{0}^{1} |f(t)|^{2}F_{\phi}(dt)$ is continuous. The result follows from the continuous mapping theorem. \\
(b) Note that for each fixed $i \geq 1$, we have $\sqrt{n}(\widehat{T}_{i}^{-1} - T_{i}^{-1}) = U_{n} \circ V_{n}$, where $U_{n} = \sqrt{n}(\overline{T} - Id)$ and $V_{n} = T_{i}^{-1}$. We will first derive the weak limit conditional on $T_{i} = t_{i}$. From the previous paragraph, it follows that conditional on $T_{i} = t_{i}$, $U_{n} = \sqrt{n}(n^{-1}t_{i} + n^{-1}\sum_{j \neq i} T_{j} - Id) \stackrel{d}{\rightarrow} Y$, and $V_{n}$, being a constant sequence, converges conditionally in probability to $t_{i}^{-1}$ as $n \rightarrow \infty$. So, by Theorem 4.4 in \cite{Bill68}, conditional on $T_{i} = t_{i}$, we have $(U_{n},V_{n}) \stackrel{d}{\rightarrow} (Y,t_{i}^{-1})$ in the $C[0,1]$ topology. Using the fact that the map $(f,g) \mapsto f \circ g$ is continuous in $C([0,1]^{2})$ (see, e.g., p. 155 in \cite{Bill68}), it follows from the continuous mapping theorem that conditional on $T_{i} = t_{i}$,  $\sqrt{n}(\widehat{T}_{i}^{-1} - T_{i}^{-1}) \stackrel{d}{\rightarrow} Y \circ t_{i}^{-1}$ as $n \rightarrow \infty$ for each fixed $i \geq 1$. Thus, by the Dominated Convergence Theorem, the unconditional distribution of $\sqrt{n}(\widehat{T}_{i}^{-1} - T_{i}^{-1})$ converges weakly as $n \rightarrow \infty$ for each fixed $i \geq 1$.\\
\indent To prove the weak convergence of $\sqrt{n}(\widehat{T}_{i} - T_{i}) = \sqrt{n}(T_{i} \circ \overline{T}^{-1} - T_{i})$, we will as earlier first derive its weak limit conditional on $T_{i} = t_{i}$. Now, using the fact that $T_{i}' \in C[0,1]$ almost surely, we have
\begin{eqnarray*}
\widehat{T}_{i}(s) - t_{i}(s) &=& t_{i}(\overline{T}^{-1}(s)) - t_{i}(s) \ = \ t_{i}(s + \overline{T}^{-1}(s) - s) - t_{i}(s) \\
&=& (\overline{T}^{-1}(s) - s) \times t_{i}'(s+ \beta(\overline{T}^{-1}(s) - s))
\end{eqnarray*}
for some $\beta_{1} \in [0,1]$ (possibly depending on $s$ and $i$). Thus,
\begin{eqnarray*}
\sqrt{n}(\widehat{T}_{i} - t_{i}) = \{\sqrt{n}(\overline{T}^{-1} - Id)\} \times t_{i}'(\cdot + o_{P}(1)) = \{\sqrt{n}(Id - \overline{T}) \circ \overline{T}^{-1}\} \times t_{i}'(\cdot + o_{P}(1))
\end{eqnarray*}
where the $o_{P}(1)$ term is uniform in $s$ since $||\overline{T}^{-1} - Id||_{\infty} \rightarrow 0$ as $n \rightarrow \infty$ almost surely. Using similar arguments as in the above proof and noting that $||\overline{T} - Id||_{\infty}$ as $n \rightarrow \infty$ almost surely, we deduce that $\sqrt{n}(\widehat{T}_{i} - t_{i}) \stackrel{d}{\rightarrow} Y \times t_{i}'$ as $n \rightarrow \infty$. Thus, by the Dominated Convergence Theorem, the unconditional distribution of $\sqrt{n}(\widehat{T}_{i} - T_{i})$ converges weakly as $n \rightarrow \infty$ for each fixed $i \geq 1$. \\
(c) Note that for each fixed $i \geq 1$,
\begin{eqnarray*}
&& \widehat{X}_{i}(s) - X_{i}(s) = \xi_{i}\{\phi(\overline{T}^{-1}(s)) - \phi(s)\} = \xi_{i}\{(\overline{T}^{-1}(s) - s)\phi'(s + \beta_{2}(\overline{T}^{-1}(s) - s))\} \\
\Rightarrow && \sqrt{n}(\widehat{X}_{i} - X_{i}) = \xi_{i}\{\sqrt{n}(Id - \overline{T}) \circ \overline{T}^{-1}\} \times \phi'(\cdot + o_{P}(1)),
\end{eqnarray*}
where $\beta_{2} \in [0,1]$, and the $o_{P}(1)$ term is uniform in $s$ as earlier. Similar arguments as in part (b) above yield $\sqrt{n}(\widehat{X}_{i} - X_{i}) \stackrel{d}{\rightarrow} \xi_{i}Y \times \phi'$ as $n \rightarrow \infty$ for each fixed $i \geq 1$. \\
(d) The proof is similar to that of part (a) and is omitted. \\
(e) Note that
\begin{eqnarray*}
\sqrt{n}(\overline{X}_{r} - \mu) &=& \sqrt{n}\left\{n^{-1}\sum_{i=1}^{n} \xi_{i}\phi \circ \overline{T}^{-1} - E(\xi_{1})\phi\right\} \\
&=& \sqrt{n}\left\{n^{-1}\sum_{i=1}^{n} (\xi_{i} - E(\xi_{1}))\right\}\phi \circ \overline{T}^{-1} + E(\xi_{1})\sqrt{n}\left\{\phi \circ \overline{T}^{-1} - \phi\right\} \\
&\stackrel{d}{\rightarrow}& N(0,Var(\xi_{1}))\phi + E(\xi_{1})Y \times \phi',
\end{eqnarray*} 
which follows from similar arguments as in part (c) and the independence of the $\xi_{i}$'s and the $T_{i}$'s. \\
(f) For the first part, note that
\begin{eqnarray*}
\widehat{\mathscr{K}}_{r} &=& n^{-1} \sum_{i=1}^{n} (\widehat{X}_{i} - \overline{X}_{r}) \otimes (\widehat{X}_{i} - \overline{X}_{r}) \\
&=& n^{-1} \sum_{i=1}^{n} (\widehat{X}_{i} - \mu) \otimes (\widehat{X}_{i} - \mu) - (\overline{X}_{r} - \mu) \otimes (\overline{X}_{r} - \mu) \\
&=& S_{1} + S_{2}, \ \ \ \mbox{say}.
\end{eqnarray*}
Now, some straightforward manipulations yield
\begin{eqnarray*}
S_{1} &=& n^{-1} \sum_{i=1}^{n} \{\xi_{i} \phi \circ \overline{T}^{-1} - E(\xi_{1})\phi\} \otimes \{\xi_{i} \phi \circ \overline{T}^{-1} - E(\xi_{1})\phi\} \\
&=& n^{-1} \sum_{i=1}^{n} \{\xi_{i} - E(\xi_{1})\}^{2} (\phi \circ \overline{T}^{-1}) \otimes (\phi \circ \overline{T}^{-1}) + E^{2}(\xi_{1}) (\phi \circ \overline{T}^{-1} - \phi) \otimes (\phi \circ \overline{T}^{-1} - \phi) \\
&& + \ n^{-1}E(\xi_{1}) \sum_{i=1}^{n} \{\xi_{i} - E(\xi_{1})\} \left[(\phi \circ \overline{T}^{-1}) \otimes (\phi \circ \overline{T}^{-1} - \phi) + (\phi \circ \overline{T}^{-1} - \phi) \otimes (\phi \circ \overline{T}^{-1}) \right].
\end{eqnarray*}
So,
\begin{eqnarray*}
&& \sqrt{n}(S_{1} - \mathscr{K}) \\
&=& \sqrt{n}\left\{n^{-1} \sum_{i=1}^{n} \{\xi_{i} - E(\xi_{1})\}^{2} (\phi \circ \overline{T}^{-1}) \otimes (\phi \circ \overline{T}^{-1}) - \mathscr{K}\right\} \\
&=& \sqrt{n} \left\{n^{-1} \sum_{i=1}^{n} \{\xi_{i} - E(\xi_{1})\}^{2} (\phi \circ \overline{T}^{-1}) \otimes (\phi \circ \overline{T}^{-1}) - Var(\xi_{1}) \phi \otimes \phi\right\} \\
&=& \sqrt{n} \left\{n^{-1} \sum_{i=1}^{n} \left[\{\xi_{i} - E(\xi_{1})\}^{2} - Var(\xi_{1})\right] (\phi \circ \overline{T}^{-1}) \otimes (\phi \circ \overline{T}^{-1}) \right. \\
&& \hspace{1cm}+ \ \left. Var(\xi_{1}) \left[(\phi \circ \overline{T}^{-1}) \otimes (\phi \circ \overline{T}^{-1}) - \phi \otimes \phi\right] \right. \\
&& \hspace{1cm}+ \ \left. E^{2}(\xi_{1})(\phi \circ \overline{T}^{-1} - \phi) \otimes (\phi \circ \overline{T}^{-1} - \phi) \right. \\
&& \hspace{1cm}+ \ \left. n^{-1}E(\xi_{1}) \sum_{i=1}^{n} \{\xi_{i} - E(\xi_{1})\} \left[(\phi \circ \overline{T}^{-1}) \otimes (\phi \circ \overline{T}^{-1} - \phi) + (\phi \circ \overline{T}^{-1} - \phi) \otimes (\phi \circ \overline{T}^{-1}) \right] \right\}
\end{eqnarray*}
The first term on the right hand side of the above equality converges in distribution to $N(0,E\{\xi_{1} - E(\xi_{1})\}^{4})\phi \otimes \phi$ since $\overline{T} \rightarrow Id$ as $n \rightarrow \infty$ almost surely. For the latter reason, the third and the fourth terms converge to zero in probability as $n \rightarrow \infty$. For the second term, note that
\begin{eqnarray*}
&& (\phi \circ \overline{T}^{-1}) \otimes (\phi \circ \overline{T}^{-1}) - \phi \otimes \phi \\
&=& (\phi \circ \overline{T}^{-1} - \phi) \otimes \phi + (\phi \circ \overline{T}^{-1}) \otimes (\phi \circ \overline{T}^{-1} - \phi).
\end{eqnarray*}
Thus, by similar arguments as in part (c) earlier, and the continuity of the mapping $(f,g) \mapsto f \otimes g$ from $L_{2}([0,1]^{2})$ to the space of Hilbert Schmidt operators, we have that the second term converges in distribution to $Var(\xi_{1})\{(Y \times \phi') \otimes \phi + \phi \otimes (Y \times \phi')\}$. Combining the above observations and the fact that $\sqrt{n}S_{2} \rightarrow 0$ in probability (follows from part (e)), we deduce that
\begin{eqnarray*}
\sqrt{n}(\widehat{\mathscr{K}}_{r} - \mathscr{K}) \stackrel{d}{\rightarrow} N(0,E\{\xi_{1} - E(\xi_{1})\}^{4})\phi \otimes \phi + Var(\xi_{1})\{(Y \times \phi') \otimes \phi + \phi \otimes (Y \times \phi')\}
\end{eqnarray*}
as $n \rightarrow \infty$. \\
\indent In order to prove the weak convergence of the empirical process $\{\sqrt{n}(\widehat{K}_{r}(s,t) - K(s,t)) : s,t \in [0,1]\}$ in $C([0,1]^{2})$, we follow the same decomposition as in the proof of the weak convergence of the operators in the Hilbert Schmidt topology. Now, note that the proof of part (c) of the theorem implies that the empirical process $\{\sqrt{n}(\phi(\overline{T}^{-1}(t)) - \phi(t)) : t \in [0,1]\}$ in $C[0,1]$ converges in distribution to the process $\{Y(t)\phi'(t) : t \in [0,1]\}$ in $C[0,1]$. This fact and the same arguments as in part (f) yield
\begin{eqnarray*}
&& \{\sqrt{n}(\widehat{K}_{r}(s,t) - K(s,t)) : s,t \in [0,1]\}  \\
&\stackrel{d}{\rightarrow}& \{Z\phi(s)\phi(t) + Var(\xi_{1})[Y(s)\phi'(s)\phi(t) + Y(t)\phi'(t)\phi(s)] : s,t \in [0,1]\}
\end{eqnarray*}
as $n \rightarrow \infty$, where $Z \sim N(0,E\{\xi_{1} - E(\xi_{1})\}^{4})$ does not depend on $s,t$. \\
\indent For the weak convergence of $\widehat{\phi}$, first note that $\widehat{\mathscr{K}}_{r} = n^{-1}\sum_{i=1}^{n} (\xi_{i} - \overline{\xi})^{2} (\phi \circ \overline{T}^{-1}) \otimes (\phi \circ \overline{T}^{-1})$. Thus, $\widehat{\phi} =  (\phi \circ \overline{T}^{-1})/||\phi \circ \overline{T}^{-1}||$. Now,
\begin{eqnarray*}
\widehat{\phi} - \phi &=& \frac{\phi \circ \overline{T}^{-1}}{||\phi \circ \overline{T}^{-1}||} - \phi 
= \frac{\phi \circ \overline{T}^{-1} - \phi}{||\phi \circ \overline{T}^{-1}||} - \frac{\phi(||\phi \circ \overline{T}^{-1}|| - 1)}{||\phi \circ \overline{T}^{-1}||} \\
&=& \frac{\phi \circ \overline{T}^{-1} - \phi}{||\phi \circ \overline{T}^{-1}||} - \frac{\phi(||\phi \circ \overline{T}^{-1}||^{2} - 1)}{||\phi \circ \overline{T}^{-1}||(||\phi \circ \overline{T}^{-1}|| + 1)} \\
&=& \frac{\phi \circ \overline{T}^{-1} - \phi}{||\phi \circ \overline{T}^{-1}||} - \frac{\phi(||\phi \circ \overline{T}^{-1} - \phi||^{2} + 2\langle\phi \circ \overline{T}^{-1} - \phi,\phi\rangle)}{||\phi \circ \overline{T}^{-1}||(||\phi \circ \overline{T}^{-1}|| + 1)}.
\end{eqnarray*}
Using the weak convergence of $\sqrt{n}(\phi \circ \overline{T}^{-1} - \phi)$ to $Y \times \phi'$ in the $C[0,1]$ topology, we have that
\begin{eqnarray*}
\sqrt{n}(\widehat{\phi} - \phi) \stackrel{d}{\rightarrow} Y \times \phi' - \frac{1}{2} \times 2\langle Y \times \phi',\phi\rangle\phi = Y \times \phi' - \langle Y \times \phi',\phi\rangle\phi
\end{eqnarray*}
as $n \rightarrow \infty$ in the $C[0,1]$ topology. \\
\indent Finally, for the weak convergence of the $\widehat{\xi}_{i}$'s, observe that 
\begin{eqnarray*}
\sqrt{n}(\widehat{\xi}_{i} - \xi_{i}) &=& \sqrt{n}\{\langle\widehat{X}_{i} - X_{i},\widehat{\phi} - \phi\rangle + \langle\widehat{X}_{i} - X_{i},\phi\rangle + \langle X_{i},\widehat{\phi} - \phi\rangle\} \\
&=& \sqrt{n}\{\xi_{i}\langle(\phi \circ \overline{T}^{-1} - \phi),(\widehat{\phi} - \phi)\rangle + \xi_{i}\langle(\phi \circ \overline{T}^{-1} - \phi),\phi\rangle + \xi_{i}\langle\phi,(\widehat{\phi} - \phi)\rangle\}.
\end{eqnarray*}
Using the independence of  $\xi_{i}$ and the $T_{j}$'s, and using the asymptotic distributions obtained above and in part (c), it follows that
\begin{eqnarray*}
\sqrt{n}(\widehat{\xi}_{i} - \xi_{i}) \stackrel{d}{\rightarrow} \xi_{i}\{\langle Y \times \phi',\phi\rangle + \langle \phi, (Y \times \phi' - 2^{-1}\{||Y \times \phi' + \phi||^{2} - 1\}\phi)\rangle\}
\end{eqnarray*}
as $n \rightarrow \infty$.
\end{proof}
\indent In order to prove Theorem \ref{thm4}, we will first prove a few crucial results.
\begin{proposition} \label{prop1}
Assume that $\phi \in C^{2}[0,1]$ and $\inf_{t \in [0,1]} T'(u) \geq \delta > 0$ almost surely for a deterministic constant $\delta$. Then, for each $i \geq 1$, we have $\sum_{j=1}^{r-1} |\phi(s_{i,j+1}) - \phi(s_{i,j})| = \int_{0}^{1} |\phi'(u)|du + B_{1,r}$ almost surely, where $B_{1,r} = O(r^{-1})$ almost surely with the $O(1)$ term being uniform in $i$. Further, $\sum_{j \in \mathscr{I}_{t}} |\phi(s_{i,j+1}) - \phi(s_{i,j})| = \int_{0}^{T_{i}^{-1}(t)} |\phi'(u)|du + B_{2,r}(t)$ for all $t \in [0,1]$ almost surely, where $||B_{2,r}||_{\infty} = O(r^{-1})$ almost surely with the $O(1)$ term being uniform in $i$. Consequently, we have $\sum_{j=1}^{r-1} |\phi(t_{j+1}) - \phi(t_{j})| = \int_{0}^{1} |\phi'(u)|du + B_{3,r}$ and $\sum_{j \in \mathscr{I}_{t}} |\phi(t_{j+1}) - \phi(t_{j})| = \int_{0}^{t} |\phi'(u)|du + B_{4,r}(t)$ for all $t \in [0,1]$ almost surely, where $B_{3,r} = O(r^{-1})$ and $||B_{4,r}||_{\infty} = O(r^{-1})$ almost surely.
\end{proposition}
\begin{proof}[Proof of Proposition \ref{prop1}]
First, let us define $t_{0} = 0$ and $t_{r+1} = 1$ in case $t_{1} > 0$ and $t_{r} < 1$. Then, $\{t_{j} : 0 \leq j \leq r+1\}$ is a partition of $[0,1]$. Consider the sum $S_{i} = \sum_{j=0}^{r} |\phi(s_{i,j+1}) - \phi(s_{i,j})|$ and note that by a Taylor expansion, $S_{i} = \sum_{j=0}^{r} (s_{i,j+1} - s_{i,j})|\phi'(\widetilde{s}_{i,j})|$, where $\widetilde{s}_{i,j} \in [s_{i,j},s_{i,j+1}]$. The right hand side is a Riemann sum approximation of $\int_{0}^{1} |\phi'(u)|du$ with $\{s_{i,j} = T_{i}^{-1}(t_{j}) : 0 \leq j \leq r+1\}$ as the partition of $[0,1]$, since $T_{i}$ is a strictly increasing bijection. Thus, writing $\Delta = \max_{0 \leq j \leq r} (s_{i,j+1} - s_{i,j})$, we have
\begin{eqnarray*}
|S_{i} - \int_{0}^{1} |\phi'(u)|du| &\leq& \sup\{|~|\phi'(t)| - |\phi'(s)|~| : s,t \in [0,1] \ \mbox{and} \ |t-s| \leq \Delta\} \\
&\leq& \sup\{|\phi'(t) - \phi'(s)| : s,t \in [0,1] \ \mbox{and} \ |t-s| \leq \Delta\} \\
&\leq& ||\phi''||_{\infty}\Delta.
\end{eqnarray*}
Now for any $0 \leq j \leq r$, we have
\begin{eqnarray*}
s_{i,j+1} - s_{i,j} = T_{i}^{-1}(t_{j+1}) - T_{i}^{-1}(t_{j}) = (t_{j+1} - t_{j})/T_{i}'(T_{i}^{-1}(\widetilde{t}_{j})),
\end{eqnarray*}
for some $\widetilde{t}_{j} \in [t_{j},t_{j+1}]$. Using the assumption in the theorem and that on the grid, it now follows that $\Delta = \max_{0 \leq j \leq r} (s_{i,j+1} - s_{i,j}) \leq \delta^{-1}O(r^{-1})$ uniformly on $i$. Thus, $|S_{i} - \int_{0}^{1} |\phi'(u)|du| \leq ||\phi''||_{\infty}\delta^{-1}O(r^{-1})$. To complete the first part of the proof, note that $\sum_{j=1}^{r-1} |\phi(s_{i,j+1}) - \phi(s_{i,j})|$ differs from $S_{i}$ by at most two terms, and both of these terms are $O(r^{-1})$ uniformly over $i$ by the same arguments as those for $S_{i}$. \\
\indent For the second part, fix any $t \in [0,1]$. Defining $B_{2,r}(0) = 0$, there is nothing to prove when $t = 0$. For $t > 0$, define $t_{0} = 0$. If $j^{*}$ is the largest $j$ for which $t_{j+1} \leq t$, define $t_{j^{*}+1} = t$ if $t_{j^{*}+1} < t$. Note that $j^{*}$ depends on $t$. Then, $\{t_{j} : 0 \leq j \leq j^{*}+1\}$ is a partition of $[0,t]$, and hence $\{s_{i,j} = T_{i}^{-1}(t_{j}) : 0 \leq j \leq j^{*}+1\}$ is a partition of $[0,T_{i}^{-1}(t)]$. Define $R_{i}(t) = \sum_{j=0}^{j^{*}} |\phi(s_{i,j+1}) - \phi(s_{i,j})|$. Then, by similar arguments as earlier, we have
\begin{eqnarray*}
\left|R_{i}(t) - \int_{0}^{T_{i}^{-1}(t)} |\phi'(u)|du\right| \leq ||\phi''||_{\infty}\delta^{-1}\max_{0 \leq j \leq j^{*}} (s_{i,j+1} - s_{i,j}) = B_{2,r}(t), \ \ \mbox{say}.
\end{eqnarray*}
Thus, $||B_{2,r}||_{\infty} \leq O(r^{-1})$ uniformly over $i$. The proof is completed upon noting that $R_{i}(t)$ differs from $\sum_{j \in \mathscr{I}_{t}} |\phi(s_{i,j+1}) - \phi(s_{i,j})|$ by at most two terms, and both of them are $O(r^{-1})$ uniformly over $i$ by the same argument as before. \\
The last statement of the proposition is an immediate corollary for the case $T = Id$ almost surely.
\end{proof}
Note that the $B_{l,r}$'s are not continuous functions, but we can still define their $||\cdot||_{\infty}$ norms as all of them are uniformly bounded functions on $[0,1]$. The following corollary is a consequence of Proposition \ref{prop1} and the fact that $\int_{0}^{1} |\phi'(u)|du \in (0,\infty)$.
\begin{corollary} \label{cor1}
Under the assumptions of Proposition \ref{prop1}, we have $\widetilde{F}_{i,d}(t) = \widetilde{F}_{i}(t) + C_{1,r}(t)$ for all $t \in [0,1]$ almost surely for each $i \geq 1$, where $||C_{1,r}||_{\infty} = O(r^{-1})$ almost surely uniformly over $i$. Further, $F_{d}(t) = F_{\phi}(t) + C_{2,r}(t)$ for all $t \in [0,1]$, where $||C_{2,r}||_{\infty} = O(r^{-1})$.
\end{corollary}
\begin{lemma} \label{lem1}
Assume that $\int_{0}^{1} |\phi'(u)|^{-\epsilon}du < \infty$ for some $\epsilon > 0$. Then, $|F_{\phi}^{-1}(s) - F_{\phi}^{-1}(t)| \leq C_{\phi}|t-s|^{\epsilon/(1+\epsilon)}$, where $C_{\phi}^{1+\epsilon} = \int_{0}^{1} |\phi'(u)|^{-\epsilon}du$. In other words, $F_{\phi}^{-1}$ is $\alpha$-H\"older continuous for $\alpha = \epsilon/(1+\epsilon)$.
\end{lemma}
\begin{proof}[Proof of Lemma \ref{lem1}]
Note that the assumption in the statement of the lemma implies that $\phi' > 0$ almost everywhere with respect to the Lebesgue measure on $[0,1]$. This fact along with Zarecki's theorem on the inverse of an absolutely continuous function (see, e.g., p. 271 in \citet{Nata55}) applied to the function $F_{\phi}$ yields that $F_{\phi}^{-1}$ is absolutely continuous on $[0,1]$. Thus, $F_{\phi}^{-1}(t) = \int_{0}^{t} [F_{\phi}'(F_{\phi}^{-1}(u))]^{-1}du$. Now, using H\"older's inequality and some algebraic manipulations, we obtain
\begin{eqnarray*}
|F_{\phi}^{-1}(s) - F_{\phi}^{-1}(t)| \leq ||\phi'||_{\infty}|t-s|^{1/p}\left(\int_{0}^{1} |\phi'(u)|^{-q+1}du\right)^{1/q}.
\end{eqnarray*}
To complete the proof, choose $q = 1+\epsilon$, which implies that $p = (1+\epsilon)/\epsilon$.
\end{proof}
\begin{proposition} \label{prop2}
Assume that the conditions of Proposition \ref{prop1} and Lemma \ref{lem1} hold. Let $\alpha = \epsilon/(1+\epsilon)$ as in Lemma \ref{lem1}. Then, for each $i \geq 1$,  \\
(a) $\widetilde{F}_{i}^{-1}$ is $\alpha$-H\"older continuous almost surely. \\
(b) $\widetilde{F}_{i,d}^{-}(t) = \widetilde{F}_{i}^{-1}(t) + ||T_{i}'||_{\infty}D_{1,r}(t)$ for all $t \in [0,1]$ almost surely, where $||D_{1,r}||_{\infty} = O(r^{-\alpha})$ almost surely uniformly over $i$.
\end{proposition}
\begin{proof}[Proof of Proposition \ref{prop2}]
(a) Using the definition of $\widetilde{F}_{i}$, it follows that 
\begin{eqnarray*}
|\widetilde{F}_{i}^{-1}(s) - \widetilde{F}_{i}^{-1}(t)| = |T_{i}(F_{\phi}^{-1}(s)) - T_{i}(F_{\phi}^{-1}(t))| \leq ||T_{i}'||_{\infty}|F_{\phi}^{-1}(s) - F_{\phi}^{-1}(t)| \leq  ||T_{i}'||_{\infty}C_{\phi}|s-t|^{\alpha},
\end{eqnarray*}
where the last inequality follows from Lemma \ref{lem1}. This completes the proof of part (a). \\
(b) As mentioned earlier, $\widetilde{F}_{i,d}$ is a c\`adl\`ag step function with maximum jump discontinuities given by $A_{i,r}$. Thus, if $t \in (\widetilde{F}_{i,d}(t_{j}),\widetilde{F}_{i,d}(t_{j+1})]$ for any $1 \leq j \leq r-1$, it follows that $\widetilde{F}_{i,d}(\widetilde{F}_{i,d}^{-}(t)) = \widetilde{F}_{i,d}(t_{j+1}) = t + q_{i,j,r}(t)$, where $q_{i,j,r}(t) = \widetilde{F}_{i,d}(t_{j+1}) - t$. So, $|q_{i,j,r}(t)| \leq \widetilde{F}_{i,d}(t_{j+1}) - \widetilde{F}_{i,d}(t_{j}) \leq A_{i,r}$, where $A_{i,r}$ is the maximum step size of $\widetilde{F}_{i,d}$ defined earlier. Now, from arguments similar to those used in Proposition \ref{prop1}, it follows that $A_{i,r} = O(r^{-1})$ uniformly in $i$. Thus, $\widetilde{F}_{i,d}(\widetilde{F}_{i,d}^{-}(t)) = t + Q_{i,r}(t)$ for all $t \in [0,1]$ almost surely, where $||Q_{r}||_{\infty} = O(r^{-1})$ almost surely uniformly over $i$. \\
\indent From Proposition \ref{prop1}, we know that $\widetilde{F}_{i,d}(s) = \widetilde{F}_{i}(s) + C_{1,r}(s)$ for all $s \in [0,1]$ almost surely, where $||C_{1,r}||_{\infty} = O(r^{-1})$ almost surely uniformly over $i$. Letting $s = \widetilde{F}_{i,d}^{-}(t)$, we now have $t + Q_{r}(t) = \widetilde{F}_{i}(\widetilde{F}_{i,d}^{-}(t)) + C_{1,r}(\widetilde{F}_{i,d}^{-}(t))$ for all $t$ almost surely. Re-arranging terms, we obtain $\widetilde{F}_{i,d}^{-}(t) = \widetilde{F}_{i}^{-1}(t + Q_{1,r}(t))$ for all $t \in [0,1]$ almost surely, where $Q_{1,r}(t) = Q_{r}(t) - C_{1,r}(\widetilde{F}_{i,d}^{-}(t))$. Thus, $||Q_{1,r}||_{\infty} = O(r^{-1})$ almost surely uniformly over $i$. Now, using part (a), we can conclude that $\widetilde{F}_{i,d}^{-}(t) = \widetilde{F}_{i}^{-1}(t) + ||T_{i}'||_{\infty}D_{1,r}(t)$ for all $t \in [0,1]$ almost surely, where $D_{1,r}(t) = C_{\phi}|Q_{1,r}(t)|^{\alpha}$ satisfies $||D_{1,r}||_{\infty} = O(r^{-\alpha})$ almost surely uniformly over $i$.
\end{proof}
\begin{proof}[Proof of Theorem \ref{thm4}]
(a) Note that 
\begin{eqnarray*}
\widehat{F}_{d}^{*}(t) \ = \ n^{-1} \sum_{i=1}^{n} \widetilde{F}_{i,d}^{-}(t) &=& n^{-1} \sum_{i=1}^{n} \{\widetilde{F}_{i}^{-1}(t) + ||T_{i}'||_{\infty}D_{1,r}(t)\}  
= \widehat{F}^{-1}(t) + \left(n^{-1}\sum_{i=1}^{n} ||T_{i}'||_{\infty}D_{1,r}(t)\right) \\
&=& \widehat{F}^{-1}(t) + D_{2,r}(t)
\end{eqnarray*}
for all $t \in [0,1]$ almost surely, where $||D_{2,r}||_{\infty} = O(r^{-\alpha})$ almost surely since $||D_{1,r}||_{\infty} = O(r^{-\alpha})$ almost surely and $n^{-1}\sum_{i=1}^{n} ||T_{i}'||_{\infty} = E(||T_{1}'||_{\infty}) + o(1)$ almost surely. Thus, it follows from Theorem 2.18 in \cite{Vill03} that
\begin{eqnarray*}
d^{2}_{W}(\widehat{F}_{d},F_{\phi}) = ||\widehat{F}_{d}^{*} - F_{\phi}^{-1}||^{2} \leq 2||\widehat{F}^{-1} - F_{\phi}^{-1}||^{2} + 2||D_{2,r}||^{2} \leq 2d^{2}_{W}(\widehat{F},F_{\phi}) + O(r^{-2\alpha})
\end{eqnarray*} 
almost surely. Combining the above statement with part (a) of Theorem \ref{thm2} and \ref{thm3} completes the proof of part (a) of Theorem \ref{thm4}. \\
(b) Next, note that 
\begin{eqnarray*}
\widehat{T}_{i,d}^{*}(t) &=& n^{-1} \sum_{l=1} \widetilde{F}_{l,d}^{-}(\widetilde{F}_{i,d}(t)) 
= n^{-1} \sum_{l=1}^{n} \left\{\widetilde{F}_{l}^{-1}(\widetilde{F}_{i,d}(t)) + ||T_{i}'||_{\infty}D_{1,r}(\widetilde{F}_{i,d}(t))\right\} \\
&=& n^{-1} \sum_{l=1}^{n} \widetilde{F}_{l}^{-1}(\widetilde{F}_{i}(t) + C_{1,r}(t)) + n^{-1}\sum_{i=1}^{n} ||T_{i}'||_{\infty}D_{1,r}(\widetilde{F}_{i,d}(t)) \\
&=& n^{-1} \sum_{l=1}^{n} \left[\widetilde{F}_{l}^{-1}(\widetilde{F}_{i}(t)) + \left\{\widetilde{F}_{l}^{-1}(\widetilde{F}_{i}(t) + C_{1,r}(t)) - \widetilde{F}_{l}^{-1}(\widetilde{F}_{i}(t))\right\}\right] + n^{-1}\sum_{i=1}^{n} ||T_{i}'||_{\infty}D_{1,r}(\widetilde{F}_{i,d}(t)) \\
&=& \widehat{T}_{i}^{-1}(t) + n^{-1} \sum_{l=1}^{n} \left\{\widetilde{F}_{l}^{-1}(\widetilde{F}_{i}(t) + C_{1,r}(t)) - \widetilde{F}_{l}^{-1}(\widetilde{F}_{i}(t))\right\} + n^{-1}\sum_{i=1}^{n} ||T_{i}'||_{\infty}D_{1,r}(\widetilde{F}_{i,d}(t)),
\end{eqnarray*}
for all $t \in [0,1]$ almost surely. By part (a) of Proposition \ref{prop2}, we have $|\{\widetilde{F}_{l}^{-1}(\widetilde{F}_{i}(t) + C_{1,r}(t)) - \widetilde{F}_{l}^{-1}(\widetilde{F}_{i}(t))\}| \leq ||T_{i}'||_{\infty}D_{3,r}(t)$ for all $t \in [0,1]$ almost surely, where $||D_{3,r}||_{\infty} = O(r^{-\alpha})$ almost surely uniformly over $i$. Thus, $\sup_{t \in [0,1]} n^{-1} \sum_{i=1}^{n} |\{\widetilde{F}_{l}^{-1}(\widetilde{F}_{i}(t) + C_{1,r}(t)) - \widetilde{F}_{l}^{-1}(\widetilde{F}_{i}(t))\}| \leq \{E(||T_{1}'||_{\infty}) + o(1)\}O(r^{-\alpha})$ almost surely. Similar arguments yield 
$$\sup_{t \in [0,1]} n^{-1}\sum_{i=1}^{n} ||T_{i}'||_{\infty}|D_{1,r}(\widetilde{F}_{i,d}(t))| \leq \{E(||T_{1}'||_{\infty}) + o(1)\}O(r^{-\alpha})$$ 
almost surely. Thus, 
\begin{eqnarray}
\widehat{T}_{i,d}^{*}(t) = \widehat{T}_{i}^{-1}(t) + D_{4,r}(t),  \label{eqn1}
\end{eqnarray}
for all $t \in [0,1]$ almost surely, where $||D_{4,r}||_{\infty} = O(r^{-\alpha})$ almost surely uniformly over $i$. Consequently,
\begin{eqnarray*}
||\widehat{T}_{i,d}^{*} - T_{i}^{-1}||_{\infty} \leq ||\widehat{T}_{i}^{-1} - T_{i}^{-1}||_{\infty} + O(r^{-\alpha})
\end{eqnarray*}
almost surely, where the $O(1)$ term is uniform over $i$. This along with part (b) of Theorem \ref{thm2} shows that $||\widehat{T}_{i,d}^{*} - T_{i}^{-1}||_{\infty} \rightarrow 0$ as $n \rightarrow \infty$ almost surely for all $i \geq 1$. Equation \eqref{eqn1} implies that $\sqrt{n}(\widehat{T}_{i,d}^{*} - T_{i}^{-1}) = \sqrt{n}(\widehat{T}_{i}^{-1} - T_{i}^{-1}) + O(\sqrt{n}r^{-\alpha})$ in $L_{2}[0,1]$. This in conjunction with part (b) of Theorem \ref{thm3} proves that $\sqrt{n}(\widehat{T}_{i,d}^{*} - T_{i}^{-1})$ has the same asymptotic distribution as  $\sqrt{n}(\widehat{T}_{i}^{-1} - T_{i}^{-1})$ in the $L_{2}[0,1]$ topology. \\
\indent Next we consider $\widehat{T}_{i,d}(t) = \widetilde{F}_{i,d}^{-}(\widehat{F}_{d}(t)) = \widetilde{F}_{i}^{-1}(\widehat{F}_{d}(t)) + ||T_{i}'||_{\infty}D_{1,r}(\widehat{F}_{d}(t))$ for all $t \in [0,1]$ almost surely (from part (b) of Proposition \ref{prop2}). Note that $\widehat{F}_{d}(t) = \{n^{-1} \sum_{l=1}^{n} \widetilde{F}_{l,d}^{-}\}^{-}(t) = \{G_{n} + D_{5,r}\}^{-}(t)$, where $G_{n}(s) = n^{-1}\sum_{l=1}^{n} \widetilde{F}_{l}^{-1}(s)$ and $D_{5,r}(s) = n^{-1}\sum_{l=1}^{n} ||T_{l}'||_{\infty}D_{1,r}(s)$. Thus, $||D_{5,r}||_{\infty} = O(r^{-\alpha})$. Also note that $G_{n}$ is a strictly increasing homeomorphism on $[0,1]$. Define $\widetilde{G}_{n,r} = G_{n} + D_{5,r} = n^{-1} \sum_{l=1}^{n} \widetilde{F}_{l,d}^{-}$ so that $\widetilde{G}_{n,r}$ is an increasing function (not necessarily strictly increasing) from $[0,1]$ onto $[0,1]$. In fact, since each $\widetilde{F}_{l,d}^{-}$ is left continuous and has right limits (being the generalized inverse of the c\`adl\`ag function $\widetilde{F}_{l,d}$), $\widetilde{G}_{n,r}$ is also left continuous and has right limits. \\
\indent If $t \in (\widetilde{G}_{n,r}(v),\widetilde{G}_{n,r}(v+)]$ for some $v \in [0,1]$ with $\widetilde{G}_{n,r}(v+) > \widetilde{G}_{n,r}(v)$, then $\widetilde{G}_{n,r}(\widehat{F}_{d}(t)) = \widetilde{G}_{n,r}(\widetilde{G}_{n,r}^{-1}(t)) = \widetilde{G}_{n,r}(v) = t + (\widetilde{G}_{n,r}(v) - t)$. Now, $|\widetilde{G}_{n,r}(v) - t| \leq |\widetilde{G}_{n,r}(v+) - \widetilde{G}_{n,r}(v)| = |G_{n}(v+) - G_{n}(v) + D_{5,r}(v+) - D_{5,r}(v)| = |D_{5,r}(v+) - D_{5,r}(v)| = O(r^{-\alpha})$ uniformly in $t$ almost surely, where the penultimate equality follows from the continuity of $G_{n}$. So, in these cases, $G_{n}(\widehat{F}_{d}(t)) = \widetilde{G}_{n,r}(\widehat{F}_{d}(t)) - D_{5,r}(\widehat{F}_{d}(t)) = t + O(r^{-\alpha})$ uniformly in $t$ almost surely, i.e., $t = G_{n}(\widehat{F}_{d}(t)) + O(r^{-\alpha}))$ uniformly in $t$ almost surely. \\
\indent Next, suppose that for some $v_{1} < v_{2}$, we have $\widetilde{G}_{n,r}(v_{1}) = \widetilde{G}_{n,r}(v_{2})$, $\widetilde{G}_{n,r}(v) < \widetilde{G}_{n,r}(v_{1})$ for $v < v_{1}$ and $\widetilde{G}_{n,r}(v) > \widetilde{G}_{n,r}(v_{2})$ for $v > v_{2}$. If $t = \widetilde{G}_{n,r}(v_{1}) = \widetilde{G}_{n,r}(v_{2})$, then $\widetilde{G}_{n,r}(\widehat{F}_{d}(t)) = t$ if $v_{1}$ is a continuity point of $\widetilde{G}_{n,r}$. If not, then this is already taken care of in the previous paragraph. In the former case, we have $t = G_{n}(\widehat{F}_{d}(t)) + O(r^{-\alpha})$ uniformly over $t$ almost surely. \\
\indent Finally, if $t$ is a point of both continuity and strict increment of $\widetilde{G}_{n,r}$, then $\widetilde{G}_{n,r}(\widehat{F}_{d}(t)) = t$ as well, which implies that $t = G_{n}(\widehat{F}_{d}(t)) + O(r^{-\alpha})$ uniformly over $t$ almost surely. Thus, all possibilities are exhausted. Let us denote the $O(r^{-\alpha})$ term by $D_{6,r}(\cdot)$. \\
\indent Now note that $G_{n}^{-1} = (n^{-1}\sum_{l=1}^{n} \widetilde{F}_{l}^{-1})^{-1} = (n^{-1}\sum_{l=1}^{n} T_{l} \circ F_{\phi}^{-1})^{-1} = F_{\phi} \circ \overline{T}^{-1}$. Thus, it follows from our work above that $\widehat{F}_{d}(t) = F_{\phi}\{\overline{T}^{-1}(t - D_{6,r}(t))\}$. Recall that $\widetilde{F}_{i}^{-1} = T_{i} \circ F_{\phi}^{-1}$ and that $\widehat{T}_{i,d}(t) = \widetilde{F}_{i}^{-1}(\widehat{F}_{d}(t)) + ||T_{i}'||_{\infty}D_{1,r}(\widehat{F}_{d}(t))$ for all $t \in [0,1]$ almost surely as obtained earlier. Since $\widehat{F}_{d}(t) = F_{\phi}\{\overline{T}^{-1}(t - D_{6,r}(t))\}$, it follows from the decomposition of $\widehat{T}_{i,d}(t)$ that $\widehat{T}_{i,d}(t) = T_{i}\{\overline{T}^{-1}(t - D_{6,r}(t))\} + ||T_{i}'||_{\infty}D_{1,r}(\widehat{F}_{d}(t))$ for all $t \in [0,1]$ almost surely. Since $\inf_{t \in [0,1]} T'(t) \geq \delta > 0$, it follows that $\inf_{t \in [0,1]} \overline{T}'(t) \geq n^{-1} \sum_{l=1}^{n} \inf_{t \in [0,1]} T'_{l}(t) \geq \delta > 0$. So, by Taylor expansion, we have $T_{i}\{\overline{T}^{-1}(t - D_{6,r}(t))\} = T_{i}(\overline{T}^{-1}(t)) + ||T_{i}'||_{\infty}D_{7,r}(t)$ for all $t \in [0,1]$ almost surely, where $||D_{7,r}||_{\infty} = O(r^{-\alpha})$ almost surely, where the $O(1)$ term is uniform over $i$. \\
\indent Combining the above findings, we arrive at 
\begin{eqnarray*}
\widehat{T}_{i,d}(t) &=& \widetilde{F}_{i}^{-1}(\widehat{F}_{d}(t)) + ||T_{i}'||_{\infty}D_{1,r}(\widehat{F}_{d}(t)) 
= \widetilde{F}_{i}^{-1}(G_{n}^{-1}(t) + D_{7,r}(t)) + ||T_{i}'||_{\infty}D_{1,r}(\widehat{F}_{d}(t)) \\
&=& T_{i}(\overline{T}^{-1}(t)) + ||T_{i}'||_{\infty}D_{7,r}(t) + ||T_{i}'||_{\infty}D_{1,r}(\widehat{F}_{d}(t)),
\end{eqnarray*}
where the last equality follows from the discussion in the previous paragraph. Since $||D_{1,r}||_{\infty} = O(r^{-\alpha})$ almost surely uniformly over $i$, we obtain
\begin{eqnarray*}
\widehat{T}_{i,d}(t) = \widehat{T}_{i}(t) + ||T_{i}'||_{\infty}D_{8,r}(t)  \label{eqn2}
\end{eqnarray*}
for all $t \in [0,1]$ almost surely, where $||D_{r,8}||_{\infty} = O(r^{-\alpha})$ almost surely uniformly over $i$. Consequently,
\begin{eqnarray*}
||\widehat{T}_{i,d} - T_{i}||_{\infty} \leq ||\widehat{T}_{i} - T_{i}||_{\infty} + O(1)r^{-\alpha},
\end{eqnarray*}
almost surely. Combined with part (b) of Theorem \ref{thm2}, this shows that $||\widehat{T}_{i,d} - T_{i}||_{\infty} \rightarrow 0$ as $n \rightarrow \infty$ almost surely for all $i \geq 1$. Equation \eqref{eqn2} implies that $\sqrt{n}(\widehat{T}_{i,d} - T_{i}) = \sqrt{n}(\widehat{T}_{i} - T_{i}) + O(\sqrt{n}r^{-\alpha})$ in $L_{2}[0,1]$. This in conjunction with part (b) of Theorem \ref{thm3} proves that $\sqrt{n}(\widehat{T}_{i,d} - T_{i})$ has the same asymptotic distribution as  $\sqrt{n}(\widehat{T}_{i} - T_{i})$ in the $L_{2}[0,1]$ topology. This completes the proof of part (b) of Theorem \ref{thm4}. \\
(c) Next we register the warped functional observations. As mentioned earlier, since the warped observations are only recorded over a discrete grid, the registration algorithm in the fully observed case will not work. So, as a pre-processing step, we need to first smooth the warped discrete observations. We do this by using the Nadaraya-Watson kernel regression estimator as follows. Let $k(\cdot)$ be any kernel supported on $[-1,1]$ and choose a bandwidth parameter $h > 0$. Then, the smooth version of $\widehat{X}_{i,d}$ is given by 
$$X^{\dagger}_{i}(t) = \frac{\sum_{j=1}^{r} k\left(\frac{t - t_{j}}{h}\right)\widetilde{X}_{i}(t_{j})}{\sum_{j=1}^{r} k\left(\frac{t - t_{j}}{h}\right)} = \xi_{i}\frac{\sum_{j=1}^{r} k\left(\frac{t - t_{j}}{h}\right)\phi(T_{i}^{-1}(t_{j}))}{\sum_{j=1}^{r} k\left(\frac{t - t_{j}}{h}\right)}, \ \ t \in [0,1].$$
Now, note that 
\begin{eqnarray*}
|X^{\dagger}_{i}(t) - \widetilde{X}_{i}(t)| &=& \left|\xi_{i}\frac{\sum_{j=1}^{r} k\left(\frac{t - t_{j}}{h}\right)\{\phi(T_{i}^{-1}(t_{j})) - \phi(T_{i}^{-1}(t))\}}{\sum_{j=1}^{r} k\left(\frac{t - t_{j}}{h}\right)} \right| \\
&\leq& ||\phi'||_{\infty}\delta^{-1}|\xi_{i}|\frac{\sum_{j=1}^{r} k\left(\frac{t - t_{j}}{h}\right)|t_{j} - t|}{\sum_{j=1}^{r} k\left(\frac{t - t_{j}}{h}\right)} 
\leq c|\xi_{i}|h,
\end{eqnarray*}
for all $t \in [0,1]$ almost surely, where $c$ is a constant not depending on $i$ and $t$. The first inequality above follows from arguments similar to those used in the proof of Theorem \ref{thm1}. The second inequality follows form the fact that $k(\cdot)$ is supported on $[-1,1]$ so that only those $j$'s in the numerator for which $|t_{j} - t| \leq h$ will contribute to the sum. Thus, $||X^{\dagger}_{i} - \widetilde{X}_{i}||_{\infty} \leq c|\xi_{i}|h$ almost surely. \\
\indent We register the warped discrete observation $\widetilde{X}_{i,d}$ by defining $\widehat{X}^{*}_{i} = X^{\dagger}_{i} \circ \widehat{T}_{i,d}$ for each $1 \leq i \leq n$. Observe that
\begin{eqnarray}
|\widehat{X}^{*}_{i}(t) - \widehat{X}_{i}(t)| &\leq& |\widehat{X}^{*}_{i}(t) - \widetilde{X}_{i}(\widehat{T}_{i,d}(t))| + |\widetilde{X}_{i}(\widehat{T}_{i,d}(t)) - \widehat{X}_{i}(t)| \nonumber \\
&\leq& ||X^{\dagger}_{i} - \widetilde{X}_{i}||_{\infty} + |\xi_{i}|~|\phi(T_{i}^{-1}(\widehat{T}_{i,d}(t))) - \phi(T_{i}^{-1}(\widehat{T}_{i}(t)))| \nonumber \\
&\leq& c|\xi_{i}|h + |\xi_{i}|~|\phi(T_{i}^{-1}(\widehat{T}_{i}(t) + ||T_{i}'||_{\infty}D_{8,r}(t))) - \phi(T_{i}^{-1}(\widehat{T}_{i}(t)))| \nonumber \\
&\leq& c|\xi_{i}|h + |\xi_{i}|~||T_{i}'||_{\infty}|D_{8,r}(t)|~||\phi'||_{\infty}\delta^{-1} \ \leq \ O(1)|\xi_{i}|(h + ||T_{i}'||_{\infty}r^{-\alpha}) \label{eqn3}
\end{eqnarray}
for all $t \in [0,1]$ almost surely, where the $O(1)$ term is uniform in $i$ and $t$. The last two inequalities above follow from a first order Taylor expansion and the fact that $||D_{8,r}||_{\infty} = O(r^{-\alpha})$ almost surely uniformly over $i$. Hence, 
$$||\widehat{X}^{*}_{i} - \widehat{X}_{i}||_{\infty} = O(1)|\xi_{i}||(h + ||T_{i}'||_{\infty}r^{-\alpha})$$
almost surely. In conjunction with part (c) of Theorem \ref{thm2}, this shows that $||\widehat{X}^{*}_{i} - X_{i}||_{\infty} \rightarrow 0$ as $n \rightarrow \infty$ almost surely for all $i \geq 1$. Equation \eqref{eqn3} implies that $\sqrt{n}(\widehat{X}^{*}_{i} - X_{i}) = \sqrt{n}(\widehat{X}_{i} - X_{i}) + O(\sqrt{n}(h+r^{-\alpha}))$ in $L_{2}[0,1]$. Invoking part (c) of Theorem \ref{thm3} thus establishes that $\sqrt{n}(\widehat{X}^{*}_{i} - X_{i})$ has the same asymptotic distribution as  $\sqrt{n}(\widehat{X}_{i} - X_{i})$ in the $L_{2}[0,1]$ topology. This completes the proof of part (c) of Theorem \ref{thm4}. \\
(d) Next, define the random measure induced by $\widehat{X}^{*}_{i}$ as 
\begin{eqnarray*}
\widehat{F}^{*}_{i}(t) &=& \sum_{j \in \mathscr{I}_{t}} |\widehat{X}^{*}_{i}(t_{j+1}) - \widehat{X}^{*}_{i}(t_{j})| \bigg/ \sum_{j=1}^{r-1} |\widehat{X}^{*}_{i}(t_{j+1}) - \widehat{X}^{*}_{i}(t_{j})| \\
&=& \sum_{j \in \mathscr{I}_{t}} |\widehat{X}^{\dagger}_{i}(\widehat{T}_{i,d}(t_{j+1})) - \widehat{X}^{\dagger}_{i}(\widehat{T}_{i,d}(t_{j}))| \bigg/ \sum_{j=1}^{r-1} |\widehat{X}^{\dagger}_{i}(\widehat{T}_{i,d}(t_{j+1})) - \widehat{X}^{\dagger}_{i}(\widehat{T}_{i,d}(t_{j}))| \\
&=& \frac{\left\{\sum_{j \in \mathscr{I}_{t}} |\widetilde{X}_{i}(\widehat{T}_{i,d}(t_{j+1})) - \widetilde{X}_{i}(\widehat{T}_{i,d}(t_{j}))| + O(h)|\xi_{i}|\right\} }{ \left\{\sum_{j=1}^{r-1} |\widetilde{X}_{i}(\widehat{T}_{i,d}(t_{j+1})) - \widetilde{X}_{i}(\widehat{T}_{i,d}(t_{j}))| + O(h)|\xi_{i}| \right\} }
\end{eqnarray*}
for all $t \in [0,1]$ almost surely, where the $O(1)$ term is uniform in $i$ and $t$, and the last equality follows from the fact that $||X^{\dagger}_{i} - \widetilde{X}_{i}||_{\infty} \leq c|\xi_{i}|h$ almost surely. Also note that by definition of $\widetilde{X}_{i}$, the term $|\xi_{i}|$ cancels from the numerator and the denominator. \\
\indent Using the fact that $\widehat{T}_{i,d}(t) = \widehat{T}_{i}(t) + ||T_{i}'||_{\infty}D_{8,r}(t)$ with $||D_{8,r}||_{\infty} = O(r^{-\alpha})$ almost surely, and arguments similar to those used in the proof of Proposition \ref{prop1}, one obtains
\begin{eqnarray*}
\widehat{F}^{*}_{i}(t) = \widehat{F}(t) + O(1)(h + ||T_{i}'||_{\infty}r^{-\alpha})
\end{eqnarray*}  
for all $t \in [0,1]$ almost surely, where the $O(1)$ term is uniform in $i$ and $t$ almost surely. Now, using Lemma \ref{lem1} and arguments similar to those used in the proof of part (b) of Proposition \ref{prop2}, we have
\begin{eqnarray*}
(\widehat{F}^{*}_{i})^{-}(t) = \widehat{F}^{-1}(t) + O(1)r^{-\alpha}(h + ||T_{i}'||_{\infty}r^{-\alpha})
\end{eqnarray*}  
for all $t \in [0,1]$ almost surely, where the $O(1)$ term is uniform in $i$ and $t$ almost surely. Thus,
\begin{eqnarray*}
d^{2}_{W}(\widehat{F}^{*}_{i},F_{\phi}) \ = \ ||(\widehat{F}^{*}_{i})^{-} - F_{\phi}^{-1}||^{2} &\leq& 2||\widehat{F}^{-1} - F_{\phi}^{-1}||^{2} + O(1)r^{-2\alpha}(h^{2} + r^{-2\alpha}) \\
&=& 2d^{2}_{W}(\widehat{F},F_{\phi}) + O(1)r^{-2\alpha}(h^{2} + r^{-2\alpha})
\end{eqnarray*}
almost surely. Combining the above statement with part (d) of Theorems \ref{thm2} and \ref{thm3} completes the proof of part (d) of Theorem \ref{thm4}. \\
(e) Next, define $\overline{X}_{r*} = n^{-1}\sum_{i=1}^{n} \widehat{X}^{*}_{i}$. Since $||\widehat{X}^{*}_{i} - \widehat{X}_{i}||_{\infty} = O(1)|\xi_{i}||(h + ||T_{i}'||_{\infty}r^{-\alpha})$
almost surely, it follows that
\begin{eqnarray}
||(\overline{X}_{r*} - \mu) - (\overline{X}_{r} - \mu)||_{\infty} &\leq& n^{-1}\sum_{i=1}^{n} ||\widehat{X}^{*}_{i} - \widehat{X}_{i}||_{\infty} 
\leq O(1)\{h + r^{-\alpha}n^{-1}\sum_{i=1}^{n}||T_{i}'||_{\infty}\} \nonumber \\
&\leq& O(1)(h + r^{-\alpha})  \label{eqn4}
\end{eqnarray}
almost surely since $E(||T_{1}'||_{\infty}) < \infty$. Along with part (e) of Theorem \ref{thm2}, this shows that $||\overline{X}_{r*} - \mu||_{\infty} \rightarrow 0$ as $n \rightarrow \infty$ almost surely. Equation \eqref{eqn4} implies that $\sqrt{n}(\overline{X}_{r*} - \mu) = \sqrt{n}(\overline{X}_{r} - \mu) + O(\sqrt{n}(h+r^{-\alpha}))$ in $L_{2}[0,1]$. So by part (e) of Theorem \ref{thm3} we see that $\sqrt{n}(\overline{X}_{r*} - \mu)$ has the same asymptotic distribution as  $\sqrt{n}(\overline{X}_{r} - \mu)$ in the $L_{2}[0,1]$ topology, and the proof of part (e) of Theorem \ref{thm4} is complete. \\
(f) Next, we consider the empirical covariance operator of the $\widehat{X}^{*}_{i}$'s which we will denote by $\widehat{\mathscr{K}}_{r*} = n^{-1} \sum_{i=1}^{n} (\widehat{X}^{*}_{i} - \overline{X}_{r*}) \otimes (\widehat{X}^{*}_{i} - \overline{X}_{r*})$. Recall $S_{1} = n^{-1} \sum_{i=1}^{n} (\widehat{X}_{i} - \mu) \otimes (\widehat{X}_{i} - \mu)$ from the proof of part (f) of Theorem \ref{thm3}. Now, some straightforward manipulations yield
\begin{eqnarray*}
\widehat{\mathscr{K}}_{r*} &=& S_{1} + n^{-1} \sum_{i=1}^{n} (\widehat{X}^{*}_{i} - \widehat{X}_{i}) \otimes (\widehat{X}^{*}_{i} - \widehat{X}_{i}) - (\overline{X}_{r*} - \mu) \otimes (\overline{X}_{r*} - \mu) \\
&& + \ n^{-1}\sum_{i=1}^{n} \{(\widehat{X}^{*}_{i} - \widehat{X}_{i}) \otimes (\widehat{X}_{i} - \mu) + (\widehat{X}_{i} - \mu) \otimes (\widehat{X}^{*}_{i} - \widehat{X}_{i})\} \\
&=& S_{1} + W_{1} - W_{2} + W_{3}, \ \ \ \mbox{say}.
\end{eqnarray*}
Note that $|||W_{1}||| \leq n^{-1} \sum_{i=1}^{n} ||\widehat{X}^{*}_{i} - \widehat{X}_{i}||^{2} \leq O(1)\{h^{2}n^{-1} \sum_{i=1}^{n} |\xi_{i}|^{2} + r^{-2\alpha}n^{-1}\sum_{i=1}^{n}||T_{i}'||_{\infty}^{2}\} = O(1)(h^{2} + r^{-2\alpha})$ almost surely. Next, from the previous paragraph, it follows that $|||W_{2}||| \leq ||\overline{X}_{r*} - \mu||^{2} \leq O(1)(h^{2} + r^{-2\alpha}) + 2||\overline{X}_{r} - \mu||_{\infty}^{2}$.  Moreover, $|||W_{3}||| \leq 2n^{-1}\sum_{i=1}^{n} ||\widehat{X}^{*}_{i} - \widehat{X}_{i}||||\widehat{X}_{i} - \mu|| \leq O(1)n^{-1}\sum_{i=1}^{n} \{h|\xi_{i}| + ||T_{i}'||_{\infty}r^{-\alpha}\}||\widehat{X}_{i} - \mu||$ almost surely. Observe that 
\begin{eqnarray*}
n^{-1}\sum_{i=1}^{n} |\xi_{i}|~||\widehat{X}_{i} - \mu|| &=& n^{-1}\sum_{i=1}^{n} |\xi_{i}|~||\xi_{i} \phi \circ \overline{T}^{-1} - E(\xi_{1})\phi|| \\
&\leq& n^{-1}\sum_{i=1}^{n} |\xi_{i}|~|\xi_{i} - E(\xi_{1})|~||\phi \circ \overline{T}^{-1}|| + n^{-1}\sum_{i=1}^{n} |\xi_{i}|~|E(\xi_{1})|~||\phi \circ \overline{T}^{-1} - \phi||.
\end{eqnarray*}
Since $||\phi \circ \overline{T}^{-1} - \phi||_{\infty} \rightarrow 0$ almost surely, it follows that the first term above is $O(1)$ almost surely, and the second term is $o(1)$ almost surely. Similar arguments show that $n^{-1}\sum_{i=1}^{n} ||T_{i}'||_{\infty}~||\widehat{X}_{i} - \mu|| = O(1)$ almost surely. Thus, $|||W_{3}||| \leq O(1)(h + r^{-\alpha})$ almost surely. Also, $S_{2}$ in the proof of part (f) of Theorem \ref{thm3} satsifies $|||S_{2}||| = O_{P}(n^{-1})$. Combining the above facts and using the decomposition of $\widehat{\mathscr{K}}_{r}$ in the proof of part (f) of Theorem \ref{thm3}, it follows that
\begin{eqnarray}
\widehat{\mathscr{K}}_{r*} &=& S_{1} + O(1)(h + r^{-\alpha} + ||\overline{X}_{r} - \mu||_{\infty}^{2}) \ = \ \widehat{\mathscr{K}}_{r} + O(1)(h + r^{-\alpha} + ||\overline{X}_{r} - \mu||_{\infty}^{2}) \label{eqn5}
\end{eqnarray}
almost surely. This along with part (f) of Theorem \ref{thm2} shows that $|||\widehat{\mathscr{K}}_{r*} - \mathscr{K}||| \rightarrow 0$ as $n \rightarrow \infty$ almost surely. By part (e) of Theorem \ref{thm3}, it follows that $\sqrt{n}||\overline{X}_{r} - \mu||_{\infty} = O_{P}(1)$ as $n \rightarrow \infty$. So, equation \eqref{eqn5} implies that $\sqrt{n}(\widehat{\mathscr{K}}_{r*} - \mathscr{K}) = \sqrt{n}(\widehat{\mathscr{K}}_{r} - \mathscr{K}) + O(\sqrt{n}(h+r^{-\alpha}))$ in $L_{2}[0,1]$. This in conjunction with part (f) of Theorem \ref{thm3} proves that $\sqrt{n}(\widehat{\mathscr{K}}_{r*} - \mathscr{K})$ has the same asymptotic distribution as  $\sqrt{n}(\widehat{\mathscr{K}}_{r} - \mathscr{K})$ in the Hilbert-Schmidt topology. \\
\indent For the convergence of the empirical covariance kernel $\widehat{K}_{r*}(s,t) = n^{-1} \sum_{i=1}^{n} [\widehat{X}^{*}_{i}(s) - \overline{X}_{r*}(s)][\widehat{X}^{*}_{i}(t) - \overline{X}_{r*}(t)]$, we follow the same decomposition as above for the case of the operator. Noting the all the bounds used for that proof remain valid in the sup-norm and using the same arguments, we arrive that
\begin{eqnarray}
\widehat{K}_{r*}(s,t) = \widehat{K}_{r}(s,t) + O(1)(h + r^{-\alpha} + ||\overline{X}_{r} - \mu||_{\infty}^{2})  \label{eqn6}
\end{eqnarray}
for all $s,t \in [0,1]$ almost surely, where the $O(1)$ term is uniform in $s,t$ almost surely. This along with part (f) of Theorem \ref{thm2} shows that $||\widehat{K}_{r*} - K||_{\infty} \rightarrow 0$ as $n \rightarrow \infty$ almost surely. Equation \eqref{eqn6} implies that $\{\sqrt{n}(\widehat{K}_{r*}(s,t) - K(s,t)) : s,t \in [0,1]\} = \{\sqrt{n}(\widehat{K}_{r}(s,t) - K(s,t)) : s,t \in [0,1]\} + O(\sqrt{n}(h+r^{-\alpha}))$ in $L_{2}[0,1]$ with the $O(1)$ term being uniform in $s,t$. This in conjunction with part (f) of Theorem \ref{thm3} proves that $\{\sqrt{n}(\widehat{K}_{r*}(s,t) - K(s,t)) : s,t \in [0,1]\}$ has the same asymptotic distribution as $\{\sqrt{n}(\widehat{K}_{r}(s,t) - K(s,t)) : s,t \in [0,1]\}$ in the $L_{2}([0,1]^{2})$ topology. \\
\indent To prove the strong consistency and the weak convergence of the estimated eigenfunction, we will use perturbation bounds for compact operators (see, e.g., Ch. 5 of \cite{HE15}). The leading eigenfunction $\widehat{\phi}_{*}$ of $\widehat{\mathscr{K}}_{r*}$ satisfies the inequality $||\widehat{\phi}_{*} - \phi|| \leq 2\sqrt{2}\lambda^{-1}|||\widehat{\mathscr{K}}_{r*} - \mathscr{K}||| \rightarrow 0$ as $n \rightarrow \infty$ almost surely. Further, Theorem 5.1.8 of \cite{HE15}, specifically equation (5.27), implies that $\sqrt{n}(\widehat{\phi}_{*} - \phi)$ has the same asymptotic distribution (in $L_{2}[0,1]$) as that of $\mathscr{S}\sqrt{n}(\widehat{\mathscr{K}}_{r*} - \mathscr{K})\phi$, where, in our setup, $\mathscr{S} = -\lambda^{-1}(\mathscr{I} - \phi \otimes \phi)$ with $\lambda = Var(\xi_{1})$ being the leading eigenvalue of $\mathscr{K}$, and $\mathscr{I}$ being the identity operator. Thus, from the results already establishes, it follows that the asymptotic distribution of $\sqrt{n}(\widehat{\phi}_{*} - \phi)$ is that of $-\lambda^{-1}(\mathscr{I} - \phi \otimes \phi)\sqrt{n}(\widehat{\mathscr{K}}_{r} - \mathscr{K})\phi$. Using the expression of the asymptotic distribution of $\sqrt{n}(\widehat{\mathscr{K}}_{r} - \mathscr{K})$ obtained in part (f) of Theorem \ref{thm3} and some simple calculations, it follows that the asymptotic distribution of $\sqrt{n}(\widehat{\phi}_{*} - \phi)$ is that of $Y \times \phi' - \langle Y \times \phi',\phi\rangle\phi$, which is the same as in Theorem \ref{thm3}. \\
\indent The proof of the strong consistency and the weak convergence of $\widehat{\xi}_{i*}$ follows in direct analogy to that of $\widehat{\xi}_{i}$ upon using part (c) and the above facts. The proof of part (f) of Theorem \ref{thm4} is now complete.
\end{proof}

\begin{proof}[Proof of Theorem \ref{thm-error}]
First observe that 
\begin{eqnarray}
|\widetilde{F}_{i,w}(t) - \widetilde{F}_{i}(t)| &\leq& \left|\frac{\int_{0}^{t} |\widehat{X}_{i,w}^{(1)}(u)|du}{\int_{0}^{1} |\widehat{X}_{i,w}^{(1)}(u)|du} - \frac{\int_{0}^{t} |\widehat{X}_{i,w}^{(1)}(u)|du}{\int_{0}^{1} |\widetilde{X}_{i}'(u)|du} \right| + \left|\frac{\int_{0}^{t} |\widehat{X}_{i,w}^{(1)}(u)|du}{\int_{0}^{1} |\widetilde{X}_{i}'(u)|du} - \frac{\int_{0}^{t} |\widetilde{X}_{i}'(u)|du}{\int_{0}^{1} |\widetilde{X}_{i}'(u)|du} \right|  \nonumber \\
&\leq& \frac{2\int_{0}^{1} |\widehat{X}_{i,w}^{(1)}(u) - \widetilde{X}_{i}'(u)|du}{\int_{0}^{1}|\widetilde{X}_{i}'(u)|du} \nonumber 
\leq \frac{2||\widehat{X}_{i,w}^{(1)} - \widetilde{X}_{i}'||}{|\xi_{i}|\int_{0}^{1} |\phi'(u)|du} = d_{\phi}|\xi_{i}|^{-1}A_{i,r}, \ \ \ \mbox{say}. \nonumber \\
\Rightarrow \ \ ||\widetilde{F}_{i,w} - \widetilde{F}_{i}||_{\infty} &\leq& d_{\phi}|\xi_{i}|^{-1}A_{i,r}.  \label{err1}
\end{eqnarray}
\indent Since the term $A_{i,r}$ will be key for our proof, we will first bound $E\{A_{i,r}^{2}\}$. To achieve this, we will first provide bounds on $E\{A_{i,r}^{2} | \xi_{i}, T_{i}\}$ using standard tools from non-parametric regression. So, we will have to estimate the MSE for the regression problem $Y_{ij} = \xi_{i}\phi(T_{i}^{-1}(t_{j})) + \epsilon_{ij}$ and integrate this MSE over $u \in [0,1]$, when $\xi_{i}$ and $T_{i}$ are fixed. 
The expression for the MSE in the deterministic design case is the same as the conditional MSE (given design points) in the random design case with the design distribution being uniform on $[0,1]$. Next, observe that $Var(\widehat{X}_{i,w}(u) | \xi_{i}, T_{i})$ does not depend on $\xi_{i}$ and $T_{i}$ and is thus uniform over $i$ (since the $\epsilon_{ij}$'s are i.i.d.). For $u \in [h_{1},1-h_{1}]$, the expression of this variance is given in p. 137 in \cite{WJ95} and equals $O((rh_{1})^{-3})$, where the $O(1)$ term depends on $k_{1}$, is bounded and is uniform over $u \in [h_{1},1-h_{1}]$. Next, we have to take into account the boundary points. Let $u = {\alpha}h_{1}$ for some $\alpha \in [0,1)$. It follows from a similar analysis that even in this case, $Var(\widehat{X}_{i,w}(u) | \xi_{i}, T_{i}) = O((rh_{1})^{-3})$, where the $O(1)$ term is integrable over $\alpha \in [0,1)$ (see, e.g. pp. 244-247 in \cite{Smooth00}). Similar estimates also hold for $t \in [1-h_{1},1]$, say $t = 1 - {\alpha}h_{1}$. Hence, we get that $Var(\widehat{X}_{i,w}(u) | \xi_{i}, T_{i}) = O((rh_{1})^{-3})$ for all $u \in [0,1]$ with the $O(1)$ term being integrable over $u \in [0,1]$. \\
\indent Next we consider the bias. In our case the degree of the fitted polynomial is one more than the degree of derivative estimated. Thus, applying Taylor's formula and using the expressions in Thm. 9.1 and pp. 244-247 in \cite{Smooth00}, we have $|Bias(\widehat{X}_{i,w}(u) | \xi_{i}, T_{i})| = ||\widetilde{X}_{i}'''||_{\infty}O(h_{1}^{2}) + ||\widetilde{X}_{i}^{(4)}||_{\infty}o(h_{1}^{2})$ for all $u \in [0,1]$. Here, the $O(1)$ and $o(1)$ terms are non-random and are integrable in $u \in [0,1]$. So, using the moment assumptions on the sup-norm of the derivatives of $T$, the independence of the $\xi_{i}$'s and the $T_{i}$'s along with the assumption that $\inf_{t \in [0,1]} T'(t) \geq \delta > 0$, it follows that
\begin{eqnarray}
E\{A_{i,r}^{2}\} = O(h_{1}^{4}) + O((rh_{1}^{3})^{-1})  \label{err2}
\end{eqnarray}
where the $O(1)$ terms are bounded and do not depend on $i$ (the $\widetilde{X}_{i}$'s are i.i.d). This also implies (using Markov's inequality) that 
\begin{eqnarray}
n^{-1} \sum_{i=1}^{n} A_{i,r}^{2} = O_{P}(h_{1}^{4} + (rh_{1}^{3})^{-1})  \label{err3}
\end{eqnarray}
\indent We will now proceed with the rest of the proof. First, let $u_{i,t} = \widetilde{F}_{i,w}^{-1}(t)$. From \eqref{err1}, it follows that $\widetilde{F}_{i}(u_{i,t}) = t - \widetilde{A}_{i,r}(t)$, where $||\widetilde{A}_{i,r}||_{\infty} \leq d_{\phi}|\xi_{i}|^{-1}A_{i,r}$. Thus, using part (a) of Proposition \ref{prop2}, it follows that $|\widetilde{F}_{i,w}^{-1}(t) - \widetilde{F}_{i}^{-1}(t)| = |u_{i,t} - \widetilde{F}_{i}^{-1}(t)| = \widetilde{F}_{i}^{-1}(t - \widetilde{A}_{i,r}(t)) - \widetilde{F}_{i}^{-1}(t)| \leq ||T_{i}'||_{\infty}c'_{\phi}|\xi_{i}|^{-\alpha}A_{i,r}^{\alpha}$ for a constant $c'_{\phi}$. So, $||\widetilde{F}_{i,w}^{-1} - \widetilde{F}_{i}^{-1}||_{\infty} \leq ||T_{i}'||_{\infty}c'_{\phi}|\xi_{i}|^{-\alpha}A_{i,r}^{\alpha}$. Thus, $\widehat{F}_{e}^{-1} = n^{-1} \sum_{i=1}^{n} \widetilde{F}_{i,w}^{-1} = n^{-1} \sum_{i=1}^{n} \widetilde{F}_{i}^{-1} + \widetilde{B}_{r} = \widehat{F}^{-1} + \widetilde{B}_{r}$, where $||\widetilde{B}_{r}||_{\infty} \leq c'_{\phi}n^{-1}\sum_{i=1}^{n} ||T_{i}'||_{\infty}|\xi_{i}|^{-\alpha}A_{i,r}^{\alpha}$. Define $R_{r} = n^{-1}\sum_{i=1}^{n} ||T_{i}'||_{\infty}|\xi_{i}|^{-\alpha}A_{i,r}^{\alpha}$. By H\"older's inequality, the law of large numbers, independence of $T_{i}$'s and $\xi_{i}$'s, and \eqref{err3}, we get that 
\begin{eqnarray}
R_{r} &\leq& \left[n^{-1} \sum_{i=1}^{n} ||T_{i}'||_{\infty}^{2/(2-\alpha)}|\xi_{i}|^{-2\alpha/(2-\alpha)}\right]^{1 - \alpha/2}\left[n^{-1} \sum_{i=1}^{n} A_{i,r}^{2}\right]^{\alpha/2} \nonumber \\ 
\Rightarrow \ \ R_{r} &=& O_{P}(h_{1}^{2\alpha} + (rh_{1}^{3})^{-\alpha/2})  \label{err4}
\end{eqnarray}
(a) Since $d^{2}_{W}(\widehat{F}_{e},F_{\phi}) = ||\widehat{F}_{e}^{-1} - F_{\phi}^{-1}||^{2} \leq 2||\widehat{F}_{e}^{-1} - \widehat{F}^{-1}||^{2} + 2||\widehat{F}^{-1} - F_{\phi}^{-1}||^{2} \leq 2R_{r}^{2} + 2d^{2}_{W}(\widehat{F},F_{\phi})$, the proof follows using part (a) of Theorem \ref{thm3} and \eqref{err4}. \\
(b) Note that $\widehat{T}_{i,e}^{-1}(t) = \widehat{F}_{e}^{-1}(\widetilde{F}_{i,w}(t)) = \widehat{F}^{-1}(\widetilde{F}_{i,w}(t)) + \widetilde{B}_{r}(\widetilde{F}_{i,w}(t))$ using statements proved earlier. Now, arguments in the proof of part (b) of Theorem \ref{thm3} along with \eqref{err1} yield $\widehat{F}^{-1}(\widetilde{F}_{i,w}(t)) = \widehat{T}_{i}^{-1}(t) + \widetilde{C}_{r}(t)$, where $||C_{r}||_{\infty} \leq const.R_{r}$. Thus, $\widetilde{T}_{i,e}^{-1} = \widetilde{T}_{i}^{-1} + \widetilde{C}_{1,r}$, where $||\widetilde{C}_{1,r}||_{\infty} \leq const. R_{r}$. The proof of the first statement in part (b) of this theorem now follows using part (b) of Theorem \ref{thm3} and \eqref{err4}. \\
\indent Next consider $\widehat{T}_{i,e}(t) = \widetilde{F}_{i,w}^{-1}(\widehat{F}_{e}(t)) = \widetilde{F}_{i}^{-1}(\widehat{F}_{e}(t)) + \widetilde{C}_{2,r,i}(t)$, where $||\widetilde{C}_{2,r,i}||_{\infty} \leq ||T_{i}'||_{\infty}c'_{\phi}|\xi_{i}|^{-\alpha}A_{i,r}^{\alpha}$ from statements proved earlier. Note that if $\widehat{F}_{e}(t) = v$ then $t = \widehat{F}_{e}^{-1}(v) = \widehat{F}^{-1}(v) + \widetilde{C}_{3,r}(v)$, where $||\widetilde{C}_{3,r}||_{\infty} \leq R_{r}$. So, $\widehat{F}_{e}(t) = v = \widehat{F}(t - \widetilde{C}_{3,r}(v)) = F_{\phi}(\overline{T}^{-1}(t - \widetilde{C}_{3,r}(v)))$. Noting that $\widetilde{F}_{i}^{-1} = T_{i} \circ F_{\phi}^{-1}$, we get that $\widetilde{F}_{i}^{-1}(\widehat{F}_{e}(t)) = T_{i}(\overline{T}^{-1}(t - \widetilde{C}_{3,r}(v))) = T_{i}(\overline{T}^{-1}(t)) + ||T_{i}'||_{\infty}\widetilde{C}_{4,r}(v) = \widetilde{F}_{i}^{-1}(\widehat{F}(t)) + ||T_{i}'||_{\infty}\widetilde{C}_{4,r}(v) = \widehat{T}_{i}(t) + ||T_{i}'||_{\infty}\widetilde{C}_{4,r}(v)$, where $||\widetilde{C}_{4,r}||_{\infty} \leq R_{r}$. This follows from arguments similar to those used earlier using the smoothness of $T$ and the assumption that $\inf_{t \in [0,1]} T'(t) \geq \delta > 0$. Thus, we finally have 
\begin{eqnarray}
||\widehat{T}_{i,e} - \widehat{T}_{i}||_{\infty} \leq const. \{||T_{i}'||_{\infty}R_{r} + ||T_{i}'||_{\infty}|\xi_{i}|^{-\alpha}A_{i,r}^{\alpha}\}.  \label{err4-1}
\end{eqnarray}
The proof of the second statement of part (b) of this theorem is now completed via part (b) of Theorem \ref{thm3}, \eqref{err2} and \eqref{err4}. \\
For proving part (c) of the theorem we will first have to control $E\{||\widehat{X}_{i,w} - \widetilde{X}_{i}||_{\infty}^{2} | \xi_{i}, T_{i}\}$ for each $i$. Recall that
\begin{eqnarray*}
\widehat{X}_{i,w}(t) = \frac{1}{r} \sum_{j=1}^{r} \frac{\{\widehat{s}_{2}(t;h_{2}) - \widehat{s}_{1}(t;h_{2})(t_{j} - t)\}k_{2,h_{2}}(t_{j}-t)Y_{ij}}{\widehat{s}_{2}(t;h_{2})\widehat{s}_{0}(t;h_{2}) - \widehat{s}_{1}^{2}(t;h_{2})},
\end{eqnarray*}
where $k_{2,h_{2}}(u) = h_{2}^{-1}k_{2}(u/h_{2})$ and $\widehat{s}_{l}(t;h_{2}) = r^{-1}\sum_{j=1}^{r} (t_{j} - t)^{l}k_{2,h_{2}}(t_{j}-t)$ for $l=0,1,2$. Call the denominator $\widehat{f}(t)$, which is deterministic. We will first analyse the term $\widecheck{Y}_{i,w}(t)$ which is defined like $\widehat{X}_{i,w}(t)$ but with $\widetilde{X}_{i}(t_{j})$ in place of $Y_{ij}$. Define $\widecheck{Z}_{i,w}(t) = \widehat{X}_{i,w}(t) - \widecheck{Y}_{i,w}(t)$. \\
\indent Using Taylor's formula, we get that $\widetilde{X}_{i}(t_{j}) = \widetilde{X}_{i}(t) + (t_{j} - t)\widetilde{X}_{i}'(t) + 2^{-1}(t_{j} - t)^{2}\widetilde{X}_{i}''(t) + 6^{-1}(t_{j} - t)^{3}\widetilde{X}_{i}'''(\widetilde{t}_{i,j})$, where $\widetilde{t}_{i,j}$ lies between $t$ and $t_{j}$. Plugging-in this expansion in the definition of $\widecheck{Y}_{i,w}(t)$, we have
\begin{eqnarray*}
\widecheck{Y}_{i,w}(t) &=& \widetilde{X}_{i}(t) + \frac{\widetilde{X}_{i}''(t)}{2}\frac{\widehat{s}_{2}^{2}(t;h_{2}) - \widehat{s}_{1}(t;h_{2})\widehat{s}_{3}(t;h_{2})}{\widehat{s}_{2}(t;h_{2})\widehat{s}_{0}(t;h_{2}) - \widehat{s}_{1}^{2}(t;h_{2})} \\
&& \ + \ \frac{1}{6r}\sum_{j=1}^{r} \frac{\{\widehat{s}_{2}(t;h_{2}) - \widehat{s}_{1}(t;h_{2})(t_{j} - t)\}k_{2,h_{2}}(t_{j}-t)(t_{j} - t)^{3}\widetilde{X}_{i}'''(\widetilde{t}_{i,j})}{\widehat{s}_{2}(t;h_{2})\widehat{s}_{0}(t;h_{2}) - \widehat{s}_{1}^{2}(t;h_{2})} \\
&=& \widetilde{X}_{i}(t) + Q_{i,1}(t;h_{2}) + Q_{i,2}(t;h_{2}), \ \ \ \mbox{say}
\end{eqnarray*}
for all $t \in [0,1]$. Note that the term involving $\widetilde{X}_{i}'(t)$ vanishes, which plays a crucial role in putting the local linear estimator at an advantage over other standard non-parametric regression estimators near the boundary of the data set. Thus, $|\widehat{X}_{i,w}(t) - \widetilde{X}_{i}(t)| \leq |\widecheck{Y}_{i,w}(t) - \widetilde{X}_{i}(t)| + |\widecheck{Z}_{i,w}(t)| \leq |Q_{i,1}(t;h_{2})| + |Q_{i,2}(t;h_{2})| + |\widecheck{Z}_{i,w}(t)|$. \\
\indent By approximations of Riemann sums, we have $\widehat{s}_{l}(t;h_{2}) = h_{2}^{l}\int_{-1}^{1} u^{l}k_{2}(u)du + O((rh_{2})^{-1})$ uniformly for $t \in [h_{2},1-h_{2}]$. Also, for $t \in [0,h_{2})$, say, $t = {\alpha}h_{2}$ with $\alpha \in [0,1)$, we have $\widehat{s}_{l}(t;h_{2}) = h_{2}^{l}\int_{-\alpha}^{1} u^{l}k_{2}(u)du + O((rh_{2})^{-1})$ uniformly for $\alpha \in [0,1)$. The same estimate also holds for $t \in (1-h_{2},1]$, say, $t = 1 - {\alpha}h_{2}$. Define $\mu_{l,\alpha} = \int_{-\alpha}^{1} u^{l}k_{2}(u)du$ for $l=0,1,2$. These estimates imply that for $t \in [h_{2},1-h_{2}]$, we have $|Q_{i,2}(t;h_{2})| \leq 2^{-1}||\widetilde{X}_{i}''||_{\infty}\{h_{2}^{2}\int_{-1}^{1}u^{2}k_{2}(u)du + O((rh_{2})^{-1})\}$. Further, for boundary points, we have $|Q_{i,2}(t;h_{2})| \leq 2^{-1}||\widetilde{X}_{i}''||_{\infty}\{h_{2}^{2}|B_{\alpha}| + O((rh_{2})^{-1})\}$ for $\alpha \in [0,1)$, where $B_{\alpha} = [\mu_{2,\alpha}^{2} - \mu_{1,\alpha}\mu_{3,\alpha}]/[\mu_{2,\alpha}\mu_{0,\alpha} - \mu_{1,\alpha}^{2}]$. In both case, the $O(1)$ terms are non-random (hence does not depend on $i$) and uniform over choices of $t$. Note that the leading term in the squared bias term obtainable from the previous bias expression is an upper bound for the coefficient of the squared bias term in the general result obtained in Thm. 3.3 in \cite{FG96}. It can be shown using similar arguments that $|Q_{i,3}(t;h_{2})| \leq ||\widetilde{X}_{i}'''||_{\infty}o(h_{2}^{2})$, where the $o(1)$ term is non-random and uniform over $t \in [0,1]$. Note that for $\alpha = 1$, which correspond to $t \in [h_{2},1-h_{2}]$, we have $B_{\alpha} = \int_{-1}^{1} u^{2}k_{2}(u)du$ by the symmetry of the kernel. Further, it can be shown that the denominator (which is positive by the Cauchy-Schwarz inequality) in the definition of $B_{\alpha}$ is a strictly increasing function of $\alpha \in [0,1]$ and hence its infimum is achieved at $\alpha = 0$, where it takes the value $\int_{0}^{1} u^{2}k_{2}(u)du\int_{0}^{1}k_{2}(u)du - (\int_{0}^{1}uk_{2}(u)du)^{2} =: a_{0} > 0$ (again by the Cauchy-Schwarz inequality) for any non-degenerate $k_{2}$. Thus $\sup_{\alpha \in [0,1]} |B_{\alpha}| \leq \sup_{\alpha \in [0,1]} |\mu_{2,\alpha}^{2} - \mu_{1,\alpha}\mu_{3,\alpha}|/a_{0} < \infty$ as the numerator is uniformly bounded in $\alpha$. Hence, $||\widecheck{Y}_{i,w} - \widetilde{X}_{i}||_{\infty} \leq 2^{-1}||\widetilde{X}_{i}''||_{\infty}\{h_{2}^{2}\sup_{\alpha \in [0,1]} |B_{\alpha}| + O((rh_{2})^{-1})\} + ||\widetilde{X}_{i}'''||_{\infty}o(h_{2}^{2})  \leq ||\widetilde{X}_{i}''||_{\infty}\{O(h_{2}^{2}) + O((rh_{2})^{-1})\} + ||\widetilde{X}_{i}'''||_{\infty}o(h_{2}^{2})$, where the $O(1)$ and the $o(1)$ terms are non-random (and hence do not depend on $i$). \\
\indent We next control $E\{||\widecheck{Z}_{i,w}||_{\infty}^{2}\}$. Observe that this does not depend on $\widetilde{X}_{i}$ and hence does not depend on $i$ (the errors are i.i.d.). 
Now, 
\begin{eqnarray}
&& E\left\{\sup_{t \in [0,1]} \left|\frac{1}{r} \sum_{j=1}^{r} \frac{\{\widehat{s}_{2}(t;h_{2}) - \widehat{s}_{1}(t;h_{2})(t_{j} - t)\}k_{2,h_{2}}(t_{j}-t)\epsilon_{ij}}{\widehat{s}_{2}(t;h_{2})\widehat{s}_{0}(t;h_{2}) - \widehat{s}_{1}^{2}(t;h_{2})}\right|^{2}\right\} \nonumber \\
&\leq& E\left\{\sup_{t \in [0,1]} \frac{1}{r^{2}} \sum_{j=1}^{r} \frac{\{\widehat{s}_{2}(t;h_{2}) - \widehat{s}_{1}(t;h_{2})(t_{j} - t)\}^{2}k_{2,h_{2}}^{2}(t_{j}-t)\epsilon_{ij}^{2}}{[\widehat{s}_{2}(t;h_{2})\widehat{s}_{0}(t;h_{2}) - \widehat{s}_{1}^{2}(t;h_{2})]^{2}}\right\} + \nonumber \\
&& E\left\{\frac{1}{r^{2}} \sum_{j \neq j'} \epsilon_{ij}\epsilon_{ij'} \sup_{t \in [0,1]} \left[\frac{\{\widehat{s}_{2}(t;h_{2}) - \widehat{s}_{1}(t;h_{2})(t_{j} - t)\}\{\widehat{s}_{2}(t;h_{2}) - \widehat{s}_{1}(t;h_{2})(t_{j'} - t)\}}{[\widehat{s}_{2}(t;h_{2})\widehat{s}_{0}(t;h_{2}) - \widehat{s}_{1}^{2}(t;h_{2})]^{2}} \right. \right. \nonumber\\
&& \hspace{5cm}\left. \left. \times \ k_{2,h_{2}}(t_{j}-t)k_{2,h_{2}}(t_{j'}-t) \right] \right\} \nonumber  \\
&\leq& M^{2}r^{-1} \sup_{t \in [0,1]} \frac{1}{r} \sum_{j=1}^{r} \frac{\{\widehat{s}_{2}(t;h_{2}) - \widehat{s}_{1}(t;h_{2})(t_{j} - t)\}^{2}k_{2,h_{2}}^{2}(t_{j}-t)}{[\widehat{s}_{2}(t;h_{2})\widehat{s}_{0}(t;h_{2}) - \widehat{s}_{1}^{2}(t;h_{2})]^{2}} \nonumber \\
&=& M^{2}(rh_{2})^{-1} \sup_{t \in [0,1]} \frac{\widehat{s}_{2}^{2}(t;h_{2})\widetilde{s}_{0}(t;h_{2}) + \widehat{s}_{1}^{2}(t;h_{2})\widetilde{s}_{2}(t;h_{2}) - 2\widehat{s}_{1}(t;h_{2})\widehat{s}_{2}(t;h_{2})\widetilde{s}_{1}(t;h_{2})}{[\widehat{s}_{2}(t;h_{2})\widehat{s}_{0}(t;h_{2}) - \widehat{s}_{1}^{2}(t;h_{2})]^{2}}.  \label{err5}
\end{eqnarray}
The second term on the right hand side of the first inequality vanishes due to the uncorrelatedness of the errors and the fact that the $t_{j}$'s are non-random. The bound for the first term follows from the a.s. boundedness of the errors, say with bound $M$. Here, $\widetilde{s}_{l}(t;h_{2}) = r^{-1} \sum_{j=1}^{r} (t_{j} - t)^{l}h_{2}^{-1}k^{2}_{2}\{(t_{j} - t)/h_{2}\}$, which is a definition similar to $\widehat{s}_{l}(t;h_{2})$ but with a new ``kernel" $k_{2}^{2}$. As earlier, by Riemann sum approximations, we have $\widetilde{s}_{l}(t;h_{2}) = h_{2}^{l}\int_{-\alpha}^{1} u^{l}k_{2}^{2}(u)du + O((rh_{2})^{-1})$ for $\alpha \in [0,1]$ with the $O(1)$ term being uniform on $t \in [0,1]$. Define $\nu_{l,\alpha} = \int_{-\alpha}^{1} u^{l}k_{2}^{2}(u)du$. Then, 
\begin{eqnarray*}
&& \frac{\widehat{s}_{2}^{2}(t;h_{2})\widetilde{s}_{0}(t;h_{2}) + \widehat{s}_{1}^{2}(t;h_{2})\widetilde{s}_{2}(t;h_{2}) - 2\widehat{s}_{1}(t;h_{2})\widehat{s}_{2}(t;h_{2})\widetilde{s}_{1}(t;h_{2})}{[\widehat{s}_{2}(t;h_{2})\widehat{s}_{0}(t;h_{2}) - \widehat{s}_{1}^{2}(t;h_{2})]^{2}} \\
&=& \frac{\mu_{2,\alpha}\nu_{0,\alpha} + \mu_{1,\alpha}^{2}\nu_{2,\alpha} - 2\mu_{1,\alpha}\mu_{2,\alpha}\nu_{1,\alpha}}{[\mu_{2,\alpha}\mu_{0,\alpha} - \mu_{1,\alpha}^{2}]^{2}} + O((rh_{2})^{-1}) = C_{\alpha} + O((rh_{2})^{-1}), \ \ \ \mbox{say},
\end{eqnarray*}
for all $\alpha \in [0,1]$, where the $O(1)$ term is uniform over $t \in [0,1]$. Note that the expression of $C_{\alpha}$ is the same as the coefficient of the variance term in the general result obtained in Thm. 3.3 in \cite{FG96} (with necessary adaptations). Using \eqref{err5}, it now follows that $E\{||\widecheck{Z}_{i,w}||_{\infty}^{2}\} \leq M\{\sup_{\alpha \in [0,1]}C_{\alpha}\}(rh_{2})^{-1} + o((rh_{2})^{-1}) = O((rh_{2})^{-1})$.
Hence, using the assumptions in the theorem and the bounds on $||\widecheck{Y}_{i,w} - \widetilde{X}_{i}||_{\infty}$ obtained earlier as well as the previous bound, it follows that 
\begin{eqnarray}
E\{||\widehat{X}_{i,w} - \widetilde{X}_{i}||_{\infty}^{2}\} = O(h_{2}^{4}) + O((rh_{2})^{-1}),  \label{err8}
\end{eqnarray}
where the $O(1)$ terms are bounded and do not depend in $i$. Thus, using Markov's inequality, we have
\begin{eqnarray}
n^{-1} \sum_{i=1}^{n} ||\widehat{X}_{i,w} - \widetilde{X}_{i}||_{\infty} = O_{P}\{h_{2}^{2} + (rh_{2})^{-1/2}\}. \label{err9}
\end{eqnarray}
(c) Recall that $\widehat{X}_{i,e}^{*}(t) = \widehat{X}_{i,w}(\widehat{T}_{i,e}(t))$. Thus, using \eqref{err4-1} we have
\begin{eqnarray}
|\widehat{X}_{i,e}^{*}(t) - \widehat{X}_{i}(t)| &\leq& |\widehat{X}_{i,w}(\widehat{T}_{i,e}(t)) - \widetilde{X}_{i}(\widehat{T}_{i,e}(t))| + |\widetilde{X}_{i}(\widehat{T}_{i,e}(t)) - \widetilde{X}_{i}(\widehat{T}_{i}(t))| \nonumber \\
&\leq& ||\widehat{X}_{i,w} - \widetilde{X}_{i}||_{\infty} + ||\widetilde{X}_{i}'||_{\infty}||\widehat{T}_{i,e} - \widehat{T}_{i}||_{\infty} \nonumber \\
\Rightarrow \ \ ||\widehat{X}_{i,e}^{*} - \widehat{X}_{i}||_{\infty} &\leq& ||\widehat{X}_{i,w} - \widetilde{X}_{i}||_{\infty} + const.|\xi_{i}|~||T_{i}'||_{\infty}\{R_{r} + |\xi_{i}|^{-\alpha}A_{i,r}^{\alpha}\}.  \label{err10}
\end{eqnarray}
The proof of part (c) of this theorem now follows from \eqref{err2}, \eqref{err4}, \eqref{err8} and part (c) of Theorem \ref{thm3}. \\
(d) Observe that by \eqref{err10}, we have
\begin{eqnarray*}
&& ||\overline{X}_{e*} - n^{-1}\sum_{i=1}^{n} \widehat{X}_{i}||_{\infty} \\
&\leq& n^{-1}\sum_{i=1}^{n} ||\widehat{X}_{i,w} - \widetilde{X}_{i}||_{\infty} + const.\left\{R_{r}\left(n^{-1}\sum_{i=1}^{n} |\xi_{i}|~||T_{i}'||_{\infty}\right) + n^{-1}\sum_{i=1}^{n} |\xi_{i}|^{1-\alpha}||T_{i}'||_{\infty}A_{i,r}^{\alpha}\right\}.
\end{eqnarray*}
The third term on the right hand side can be bounded using H\"older's inequality and \eqref{err3} as earlier. The bounds on the first two terms are given by \eqref{err9} and \eqref{err4}, respectively. The proof of this part of the theorem is now completed upon using these bounds along with part (e) of Theorem \ref{thm3}. \\
(e) For the proof of this part of theorem, we will use a decomposition of $\widehat{\mathscr{K}}_{e*}$ similar to that of $\widehat{\mathscr{K}}_{r}$ in the proof of part (f) of Theorem \ref{thm3}. In the same notation, we obtain the following bounds on $W_{1}, W_{2}$ and $W_{3}$. First, note that $|||W_{1}||| \leq n^{-1} \sum_{i=1}^{n} ||\widehat{X}_{i,e}^{*} - \widehat{X}_{i}||^{2} \leq 2n^{-1} \sum_{i=1}^{n} ||\widehat{X}_{i,w} - \widetilde{X}_{i}||_{\infty}^{2} + const.n^{-1} \sum_{i=1}^{n} \xi_{i}^{2}||T_{i}'||_{\infty}^{2}\{R_{r} + |\xi_{i}|^{-\alpha}A_{i,r}^{\alpha}\}^{2}$. Applying H\"older's inequality and using \eqref{err3}, \eqref{err4} and \eqref{err8}, we get that $|||W_{1}||| = O_{P}\{h_{2}^{4} + (rh_{2})^{-1} + h_{1}^{4\alpha} + (rh_{1}^{3})^{-\alpha}\}$. Next, using part (d) of this theorem and part (e) of Theorem \ref{thm3}, it follows that $|||W_{2}||| \leq ||\overline{X}_{e*} - \mu||^{2} \leq 2||\overline{X}_{e*} - n^{-1}\sum_{i=1}^{n} \widehat{X}_{i}||^{2} + 2||n^{-1}\sum_{i=1}^{n} \widehat{X}_{i} - \mu||^{2} = O_{P}\{h_{1}^{4\alpha} + (rh_{1}^{3})^{-\alpha} + h_{2}^{4} + (rh_{2})^{-1} + n^{-1}\}$. In a similar manner, $|||W_{3}||| \leq 2n^{-1}\sum_{i=1}^{n} ||\widehat{X}_{i,e}^{*} - \widehat{X}_{i}||||\widehat{X}_{i} - \mu|| = O_{P}\{h_{2}^{2} + (rh_{2})^{-1/2} + h_{1}^{2\alpha} + (rh_{1}^{3})^{-\alpha/2}\}$ by the Cauchy-Schwarz inequality and the bounds obtained earlier. So, using part (f) of Theorem \ref{thm3}, we have $|||\widehat{\mathscr{K}}_{e*} - \mathscr{K}||| = O_{P}\{h_{2}^{2} + (rh_{2})^{-1/2} + h_{1}^{2\alpha} + (rh_{1}^{3})^{-\alpha/2} + n^{-1/2}\}$. The bounds for the leading eigenvalue and eigenfunction follow directly by standard bounds in the theory of perturbation of operators.
\end{proof}

\begin{proof}[Proof of Theorem \ref{thm5}]
\indent First assume that $\mu' \neq 0$. Then, define $G(t) = \int_{0}^{t} |\gamma_{1}^{-1}\mu'(u)|du/\int_{0}^{1} |\gamma_{1}^{-1}\mu'(u)|du = \int_{0}^{t} |\mu'(u)|du/\int_{0}^{1} |\mu'(u)|du$ and $\widetilde{G}_{i}(t) = G(T_{i}^{-1}(t))$ for $t \in [0,1]$ and $i=1,2,\ldots,n$. Some algebraic manipulations yield
\begin{eqnarray*}
&& |F_{i}(t) - G(t)| \\
&\leq& \frac{\int_{0}^{t} |Y_{i1}\phi_{1}'(u) + {\eta}Y_{i2}\phi_{2}'(u)|du}{\int_{0}^{1} |\gamma_{1}^{-1}\mu'(u) + Y_{i1}\phi_{1}'(u) + {\eta}Y_{i2}\phi_{2}'(u)|du} + \left|\frac{\int_{0}^{t} |\gamma_{1}^{-1}\mu'(u)|du}{\int_{0}^{1} |\gamma_{1}^{-1}\mu'(u) + Y_{i1}\phi_{1}'(u) + {\eta}Y_{i2}\phi_{2}'(u)|du} \right. \\
&& \hspace{8cm} \left. \ - \ \frac{\int_{0}^{t} |\gamma_{1}^{-1}\mu'(u)|du}{\int_{0}^{1} |\gamma_{1}^{-1}\mu'(u)|du}\right| \\
&\leq& \frac{2\int_{0}^{1} |Y_{i1}\phi_{1}'(u) + {\eta}Y_{i2}\phi_{2}'(u)|du}{\int_{0}^{1} |\gamma_{1}^{-1}\mu'(u) + Y_{i1}\phi_{1}'(u) + {\eta}Y_{i2}\phi_{2}'(u)|du} = Z_{i}.
\end{eqnarray*}
Thus, $||F_{i} - G||_{\infty} \leq Z_{i}$ almost surely for each $i$. So $||\widetilde{F}_{i} - \widetilde{G}_{i}||_{\infty} = \sup_{t \in [0,1]} |F_{i}(T_{i}^{-1}(t)) - G(T_{i}^{-1}(t))| = \sup_{t \in [0,1]} |F_{i}(t) - G(t)| \leq Z_{i}$, where the last equality holds because $T_{i}$ is a bijection on $[0,1]$. \\
\indent Next, let $c_{i} = F_{i}^{-1}(t)$ and $c = G^{-1}(t)$. So, $t = F_{i}(c_{i}) = G(c)$. Also, $G(c) - G(c_{i}) = G(c) - F_{i}(c_{i}) + F_{i}(c_{i}) - G(c_{i}) = F_{i}(c_{i}) - G(c_{i})$ so that $|G(c) - G(c_{i})| \leq ||F_{i} - G||_{\infty} \leq Z_{i}$. The conditions of the theorem and arguments as in Lemma \ref{lem1} earlier show that $G^{-1}$ is $\alpha$-H\"older continuous for $\alpha = \epsilon/(1+\epsilon)$. Thus, for a finite, positive constant $C_{\mu}$, we have
\begin{eqnarray*}
|F_{i}^{-1}(t) - G^{-1}(t)| = |c_{i} - c| = |G^{-1}(G(c_{i})) - G^{-1}(G(c))| \leq C_{\mu}|G(c_{i}) - G(c)|^{\alpha} \leq C_{\mu}Z_{i}^{\alpha}.
\end{eqnarray*}
Thus, $||F_{i}^{-1} - G^{-1}||_{\infty} \leq C_{\mu}Z_{i}^{\alpha}$ almost surely. Consequently, $||\widetilde{F}_{i}^{-1} - \widetilde{G}_{i}^{-1}||_{\infty} = \sup_{t \in [0,1]} |T_{i}(F_{i}^{-1}(t)) - T_{i}(G^{-1}(t))| \leq ||T_{i}'||_{\infty}||F_{i}^{-1} - G^{-1}||_{\infty} \leq C_{\mu}||T_{i}'||_{\infty}||Z_{i}^{\alpha}$ almost surely. Further, 
\begin{eqnarray*}
||\widehat{F}^{-1} - \widehat{G}^{-1}||_{\infty} \ \leq \ \frac{1}{n}\sum_{i=1}^{n} ||\widetilde{F}_{i}^{-1} - \widetilde{G}_{i}^{-1}||_{\infty} \ \leq \ \frac{C_{\mu}}{n}\sum_{i=1}^{n} ||T_{i}'||_{\infty}Z_{i}^{\alpha} \ \leq \ 2C_{\mu}E(||T_{1}'||_{\infty})E(Z_{1}^{\alpha}),
\end{eqnarray*}
as $n \rightarrow \infty$ almost surely. Here, the last inequality follows from the moment assumptions in the theorem, the Cauchy-Schwarz inequality, the strong law of large numbers and the fact that the $Y_{il}$'s (and hence the $X_{i}$'s) are independent of the $T_{i}$'s. Thus,
\begin{eqnarray*}
|\widehat{T}_{i}^{-1}(t) - T_{i}(t)| &=& |\widehat{F}^{-1}(F_{i}(T_{i}^{-1}(t))) - T_{i}^{-1}(t)| \\
&\leq& |\widehat{F}^{-1}(F_{i}(T_{i}^{-1}(t))) - \widehat{G}^{-1}(F_{i}(T_{i}^{-1}(t)))| + |\widehat{G}^{-1}(F_{i}(T_{i}^{-1}(t))) - \widehat{G}^{-1}(G(T_{i}^{-1}(t)))| \\
&& \hspace{6.5cm}+ \ |\widehat{G}^{-1}(G(T_{i}^{-1}(t))) - T_{i}^{-1}(t)| \\
&\leq& ||\widehat{F}^{-1} - \widehat{G}^{-1}||_{\infty} + |\overline{T}(G^{-1}(F_{i}(T_{i}^{-1}(t)))) - \overline{T}(G^{-1}(G(T_{i}^{-1}(t))))| \\
&& \hspace{6.5cm}+ \ |\overline{T}(G^{-1}(G(T_{i}^{-1}(t)))) - T_{i}^{-1}(t)| \\
&\leq& ||\widehat{F}^{-1} - \widehat{G}^{-1}||_{\infty} + ||\overline{T}'||_{\infty}C_{\mu}|F_{i}(T_{i}^{-1}(t)) - G(T_{i}^{-1}(t))|^{\alpha} \\
&&\hspace{6.5cm}+ \ |\overline{T}(T_{i}^{-1}(t)) - T_{i}^{-1}(t)| \\
&\leq& ||\widehat{F}^{-1} - \widehat{G}^{-1}||_{\infty} + C_{\mu}n^{-1}\left\{\sum_{j=1}^{n} ||T_{j}'||_{\infty}\right\} ||F_{i} - G||_{\infty} + ||\overline{T} - Id||_{\infty} \\
&\leq& const.\left\{E(Z_{1}^{\alpha}) + Z_{i} + ||\overline{T} - Id||_{\infty}\right\}, \\
\Rightarrow ||\widehat{T}_{i}^{-1} - T_{i}^{-1}||_{\infty} &\leq& const.\left\{E(Z_{1}^{\alpha}) + Z_{i} + ||\overline{T} - Id||_{\infty}\right\}
\end{eqnarray*}
as $n \rightarrow \infty$ almost surely, where the constant term is uniform in $i$. \\
\indent Next, let $t = \widehat{F}^{-1}(u)$. Then, $n^{-1}\sum_{i=1}^{n} T_{i}(F_{i}^{-1}(u)) = t$. Let $t_{*} = n^{-1}\sum_{i=1}^{n} T_{i}(G^{-1}(u)) = \overline{T}(G^{-1}(u)) = \widehat{G}^{-1}(u)$ so that $u = \widehat{G}(t_{*})$. Note that $\widehat{F}(t) - \widehat{G}(t) = \widehat{F}(t) - \widehat{G}(t_{*}) + \widehat{G}(t_{*}) - \widehat{G}(t) = \widehat{G}(t_{*}) - \widehat{G}(t) = G(\overline{T}^{-1}(t_{*})) - G(\overline{T}^{-1}(t))$. Thus, using the assumptions in the theorem and arguments similar to those used in the proof of part (b) of Theorem \ref{thm2}, we have
\begin{eqnarray*}
|\widehat{F}(t) - \widehat{G}(t)| &=& |G(\overline{T}^{-1}(t_{*})) - G(\overline{T}^{-1}(t))| \\
&\leq& ||G'||_{\infty}|\overline{T}^{-1}(t_{*}) - \overline{T}^{-1}(t)| \\
&\leq& ||G'||_{\infty}\delta^{-1}|t_{*} - t| \\
&\leq& ||G'||_{\infty}\delta^{-1}n^{-1}\sum_{i=1}^{n} \left|T_{i}(F_{i}^{-1}(u)) - T_{i}(G^{-1}(u))\right| \\
&\leq& ||G'||_{\infty}\delta^{-1}C_{\mu}n^{-1}\sum_{i=1}^{n} ||T_{i}'||_{\infty}Z_{i}^{\alpha} \ \leq \  const.E(||T_{1}'||_{\infty})E(Z_{1}^{\alpha}) \\
\Rightarrow \ \ ||\widehat{F} - \widehat{G}||_{\infty} &\leq& const.E(Z_{1}^{\alpha})
\end{eqnarray*}
as $n \rightarrow \infty$ almost surely. Therefore,  
\begin{eqnarray*}
|\widehat{T}_{i}(t) - T_{i}(t)| &=& |\widetilde{F}_{i}^{-1}(\widehat{F}(t)) - T_{i}(t)| \\
&\leq& |\widetilde{F}_{i}^{-1}(\widehat{F}(t)) - \widetilde{G}_{i}^{-1}(\widehat{F}(t))| + |\widetilde{G}_{i}^{-1}(\widehat{F}(t)) - \widetilde{G}_{i}^{-1}(\widehat{G}(t))| + |\widetilde{G}_{i}^{-1}(\widehat{G}(t)) - T_{i}(t)| \\
&\leq& ||\widetilde{F}_{i}^{-1} - \widetilde{G}_{i}^{-1}||_{\infty} + |T_{i}(G^{-1}(\widehat{F}(t)) - T_{i}(G^{-1}(\widehat{G}(t))| + |T_{i}(G^{-1}(\widehat{G}(t)) - T_{i}(t)| \\ 
&\leq& ||\widetilde{F}_{i}^{-1} - \widetilde{G}_{i}^{-1}||_{\infty} + ||T_{i}'||_{\infty}C_{\mu}|\widehat{F}(t) - \widehat{G}(t)|^{\alpha} + |T_{i}(\overline{T}^{-1}(t)) - T_{i}(t)| \\
&\leq& ||\widetilde{F}_{i}^{-1} - \widetilde{G}_{i}^{-1}||_{\infty} + ||T_{i}'||_{\infty}C_{\mu}||\widehat{F} - \widehat{G}||^{\alpha} + ||T_{i}'||_{\infty}||\overline{T}^{-1} - Id||_{\infty} \\
&=& ||\widetilde{F}_{i}^{-1} - \widetilde{G}_{i}^{-1}||_{\infty} + ||T_{i}'||_{\infty}C_{\mu}||\widehat{F} - \widehat{G}||^{\alpha} + ||T_{i}'||_{\infty}||\overline{T} - Id||_{\infty} \\
&\leq& const.||T_{i}'||_{\infty}\left\{Z_{i}^{\alpha} + E^{\alpha}(Z_{1}^{\alpha}) + ||\overline{T} - Id||_{\infty}\right\} \\
\Rightarrow \ \ ||\widehat{T}_{i} - T_{i}||_{\infty} &\leq& const.||T_{i}'||_{\infty}\left\{Z_{i}^{\alpha} + E^{\alpha}(Z_{1}^{\alpha}) + ||\overline{T} - Id||_{\infty}\right\}
\end{eqnarray*}
as $n \rightarrow \infty$ almost surely, where the constant term is uniform in $i$. \\
\indent Next, note that $\widehat{X}_{i} = \widetilde{X}_{i} \circ \widehat{T}_{i} = X_{i} \circ T_{i}^{-1} \circ \widehat{T}_{i} = \mu \circ T_{i}^{-1} \circ \widehat{T}_{i} + \gamma_{1}Y_{i1}\phi_{1} \circ T_{i}^{-1} \circ \widehat{T}_{i} + \gamma_{2}Y_{i2}\phi_{2} \circ T_{i}^{-1} \circ \widehat{T}_{i}$. So,
\begin{eqnarray*}
|\widehat{X}_{i}(t) - X_{i}(t)| &\leq& |\mu(T_{i}^{-1}(\widehat{T}_{i}(t))) - \mu(t)| + \gamma_{1}|Y_{i1}|~|\phi_{1}(T_{i}^{-1}(\widehat{T}_{i}(t))) - \phi_{1}(t)| \\
&& \hspace{3.5cm}+ \ \gamma_{2}|Y_{i2}|~|\phi_{2}(T_{i}^{-1}(\widehat{T}_{i}(t))) - \phi_{2}(t)| \\
&\leq& |T_{i}^{-1}(\widehat{T}_{i}(t)) - t|\left\{||\mu'||_{\infty} + \gamma_{1}|Y_{i1}|~||\phi_{1}'||_{\infty} + \gamma_{1}|Y_{i2}|~||\phi_{2}'||_{\infty}\right\} \\
\Rightarrow \ \ ||\widehat{X}_{i} - X_{i}||_{\infty} &\leq& ||\widehat{T}_{i}^{-1} - T_{i}^{-1}||_{\infty}\left\{||\mu'||_{\infty} + \gamma_{1}|Y_{i1}|~||\phi_{1}'||_{\infty} + \gamma_{1}|Y_{i2}|~||\phi_{2}'||_{\infty}\right\} \\
&\leq& O_{P}(1)\left\{E(Z_{1}^{\alpha}) + Z_{i} + ||\overline{T} - Id||_{\infty}\right\}
\end{eqnarray*}
as $n \rightarrow \infty$ almost surely, where the $O_{P}(1)$ term is independent on $n$. \\
\indent Next, consider the case when $\mu' = 0$. Then, define $G(t) = \int_{0}^{t} |\phi_{1}'(u)|du/\int_{0}^{1} |\phi_{1}'(u)|du$. Some algebraic manipulations yield
\begin{eqnarray*}
|F_{i}(t) - G(t)| &=& \left|\frac{\int_{0}^{t} |Y_{i1}\phi_{1}'(u) + {\eta}Y_{i2}\phi_{2}'(u)|du}{\int_{0}^{1} |Y_{i1}\phi_{1}'(u) + {\eta}Y_{i2}\phi_{2}'(u)|du} - \frac{\int_{0}^{t} |\phi_{1}'(u)|du}{\int_{0}^{1} |\phi_{1}'(u)|du}\right| \\
&\leq& \frac{2{\eta}\int_{0}^{1} |Y_{i2}\phi_{2}'(u)|du}{\int_{0}^{1} |Y_{i1}\phi_{1}'(u) + {\eta}Y_{i2}\phi_{2}'(u)|du} \ = \ Z_{i}.
\end{eqnarray*}
Similar arguments as in the case of $\mu' \neq 0$ now yield the error bounds on the estimators
\end{proof}

\newpage
\begin{center}
\large{Plots of the registered curves for the simulated and real data sets in the main paper using some other procedures}
\end{center}

\begin{figure}[H]
\begin{center}
{
\includegraphics[scale=0.46]{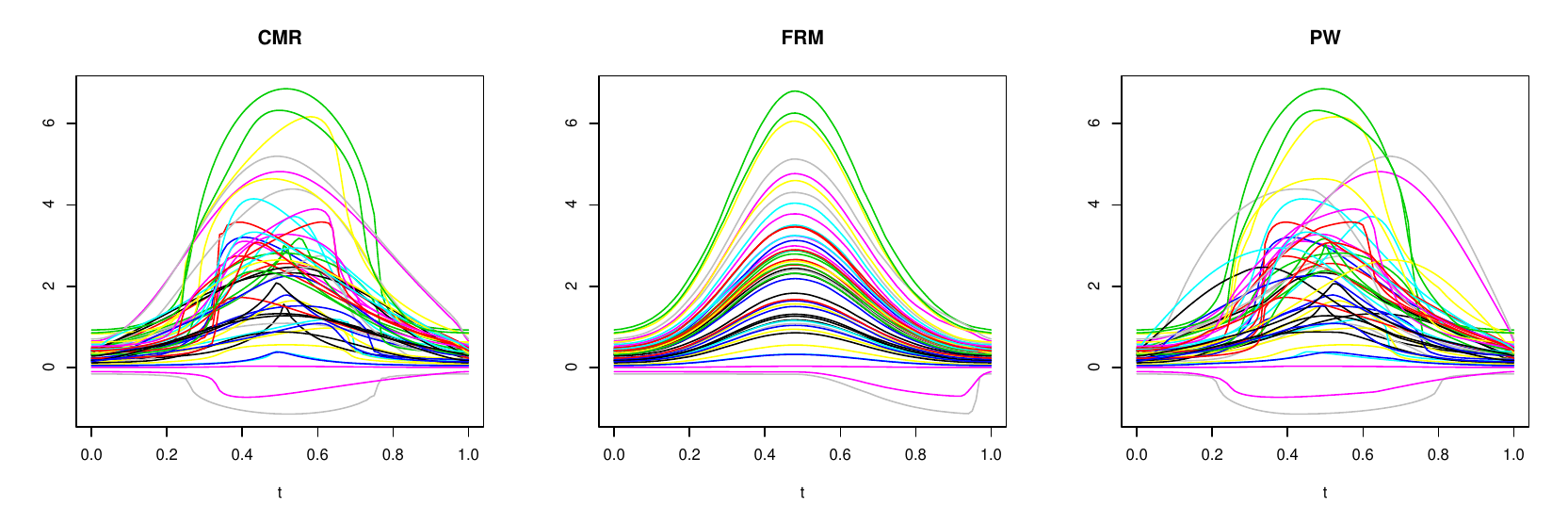}
}
\end{center}
\vspace{-0.3in}
\caption{Plots of the registered data curves using some other procedures under Model $1$ without measurement error.}
\label{Figsupp-model1}
\end{figure}
\begin{figure}[H]
\vspace{-0.2in}
\begin{center}
{
\includegraphics[scale=0.46]{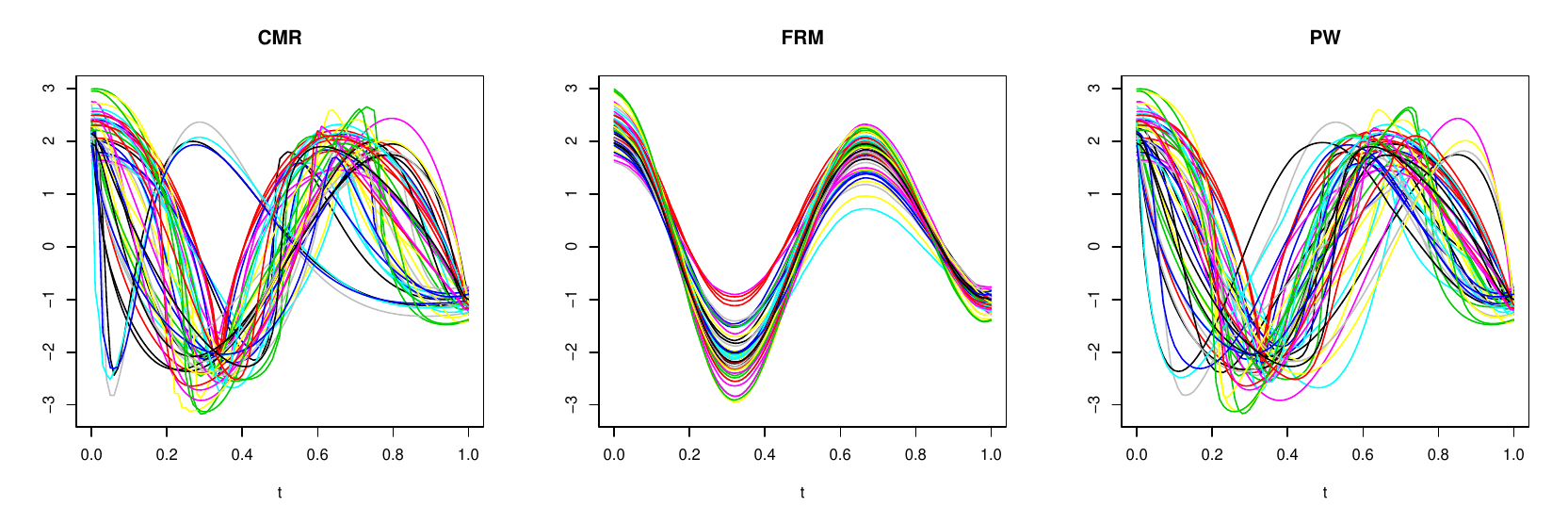}
}
\end{center}
\vspace{-0.3in}
\caption{Plots of the registered data curves using some other procedures under Model $2$ without measurement error.}
\label{Figsupp-model2}
\end{figure}
\begin{figure}[H]
\vspace{-0.2in}
\begin{center}
{
\includegraphics[scale=0.46]{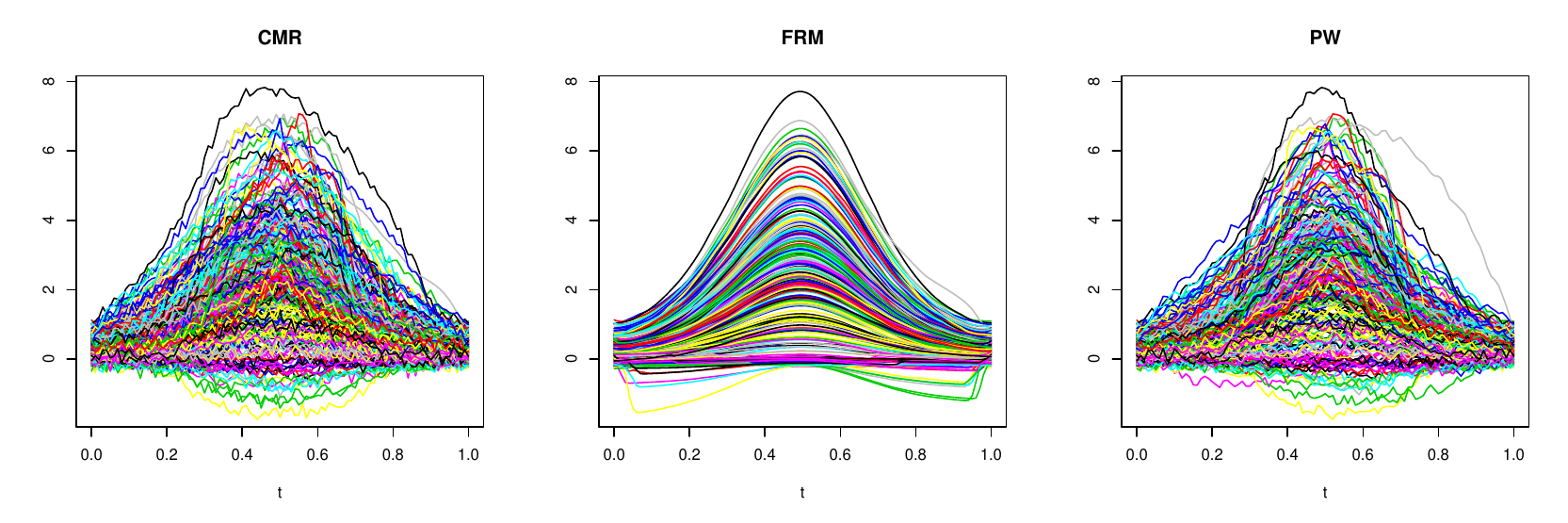}
}
\end{center}
\vspace{-0.3in}
\caption{Plots of the registered data curves using some other procedures under Model $1$ in the presence of measurement error.}
\label{Figsupp-model1err}
\end{figure}
\begin{figure}[H]
\vspace{-0.2in}
\begin{center}
{
\includegraphics[scale=0.46]{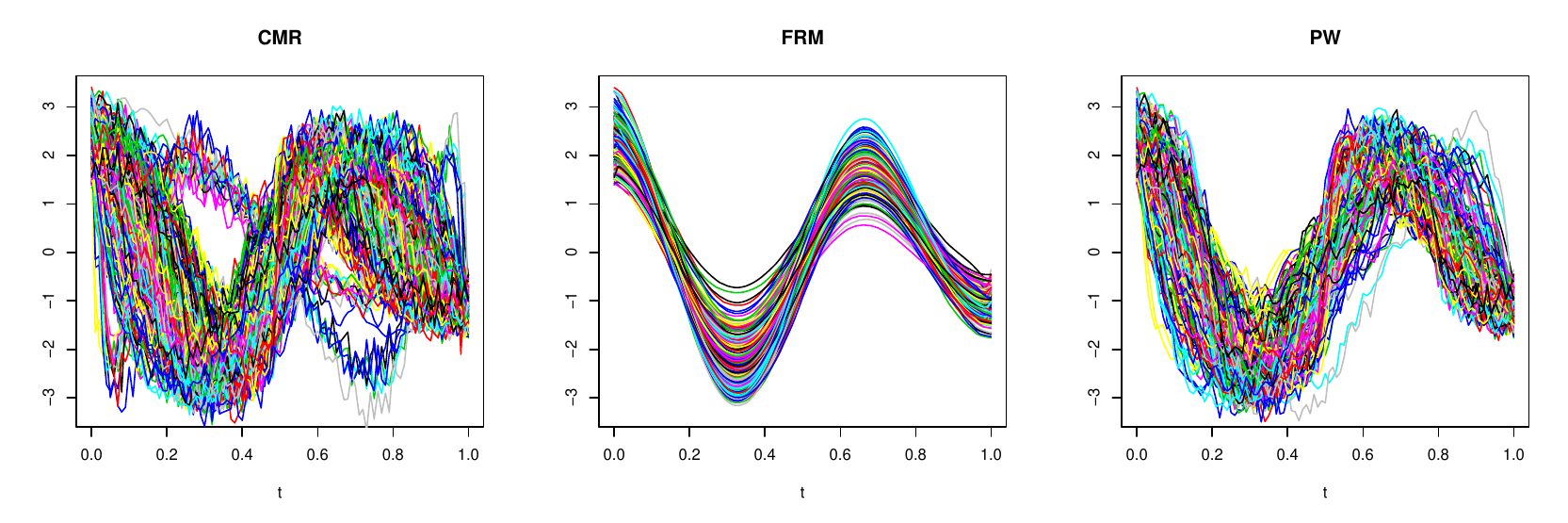}
}
\end{center}
\vspace{-0.3in}
\caption{Plots of the registered data curves using some other procedures under Model $2$ in the presence of measurement error.}
\label{Figsupp-model2err}
\end{figure}
\begin{figure}[H]
\vspace{-0.2in}
\begin{center}
{
\includegraphics[scale=0.46]{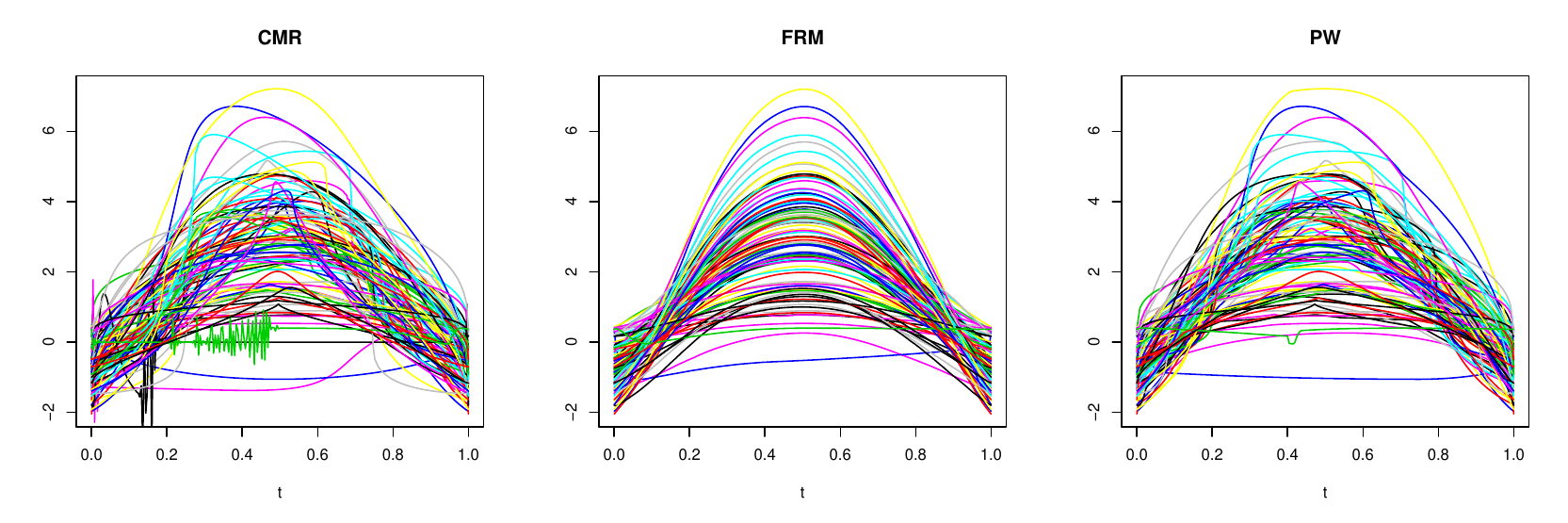}
}
\end{center}
\vspace{-0.3in}
\caption{Plots of the registered data curves using some other procedures under the rank $2$ model.}
\label{Figsupp-rank2}
\end{figure}
\begin{figure}[H]
\vspace{-0.2in}
\begin{center}
{
\includegraphics[scale=0.46]{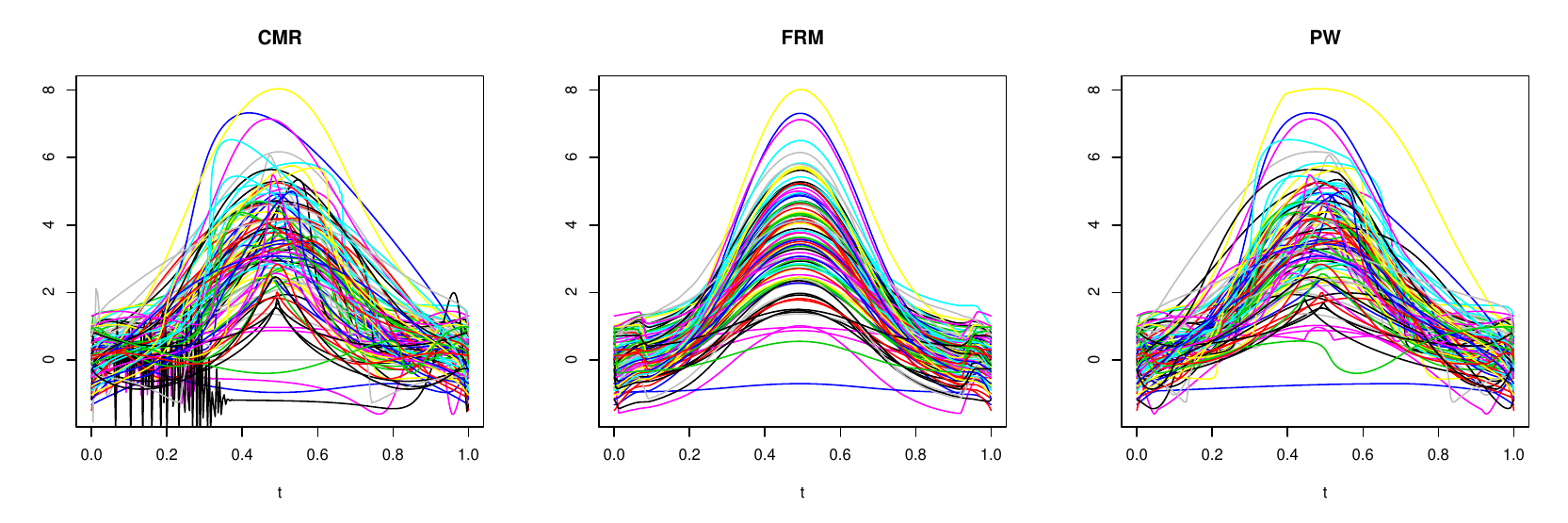}
}
\end{center}
\vspace{-0.3in}
\caption{Plots of the registered data curves using some other procedures under the rank $3$ model.}
\label{Figsupp-rank3}
\end{figure}
\begin{figure}[H]
\vspace{-0.2in}
\begin{center}
{
\includegraphics[scale=0.46]{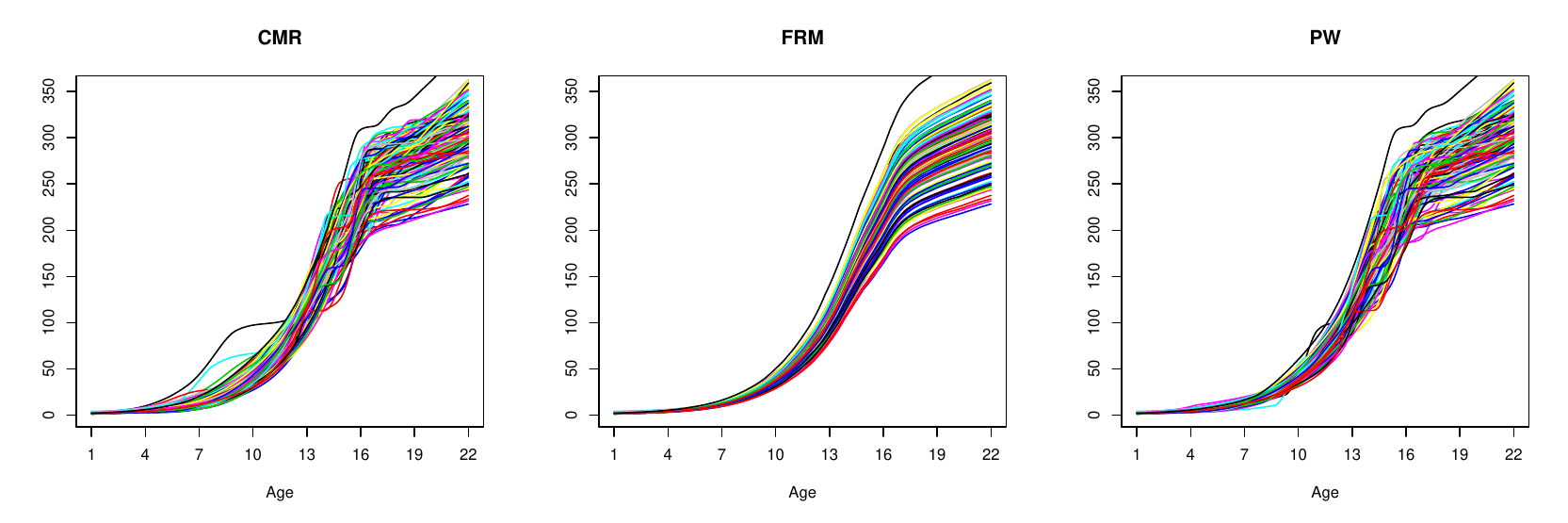}
}
\end{center}
\vspace{-0.3in}	
\caption{Plots of the registered data curves using some other procedures for the real data.}
\label{Figsupp-beetle}
\end{figure}

%
%
%
%
%

\end{document}